\documentclass{article}

\PassOptionsToPackage{sort&compress,numbers}{natbib}

\usepackage[final,main]{neurips_2026}
\usepackage{neurips-mystyle}
\usepackage[belowskip=0pt]{caption}
\usepackage[utf8]{inputenc} 
\usepackage[T1]{fontenc}    
\usepackage{xurl}            

\usepackage{hyperref}
\usepackage{booktabs}       

\newcommand{\pmv}[2]{$#1{\scriptscriptstyle\pm}#2$}
\newcommand{\bestpmv}[2]{$\mathbf{#1}{\scriptscriptstyle\pm}\mathbf{#2}$}

\usepackage{amsfonts}       
\usepackage{nicefrac}       
\usepackage{xcolor}         
\usepackage{comment}
\usepackage{subcaption}
\allowdisplaybreaks

\newcommand{\titletext}{Heterogeneous Judge-Aware Ranking with\\
Sensitivity, Disagreement, and Confidence}

\title{\titletext}

\author{Shibo Yu\\The University of Hong Kong\\\texttt{shibo-yu01@connect.hku.hk}
\And
Yingzhou Wang\\The University of Hong Kong\\\texttt{wyingz@connect.hku.hk}
\And
Yan Chen\\The University of Hong Kong\\\texttt{yanc25@hku.hk}
\And
Guodong Li\\The University of Hong Kong\\\texttt{gdli@hku.hk}
\And
Jin-Hong Du\\The University of Hong Kong\\\texttt{jinhongd@hku.hk}
}
\date{}

\begin{document}

\maketitle

\begin{abstract}
Pairwise comparisons from multiple judges are central to large language model evaluation and preference modeling, yet standard ranking pipelines often pool judgments into a single score vector, treating systematic judge disagreement as noise. 
We propose \emph{Heterogeneous Judge-Aware (HJA) ranking}, a structured multi-judge ranking framework that separates consensus ranking, judge-specific sensitivity to consensus, and residual preference disagreement. 
HJA thereby treats ranking, judge sensitivity, and structured disagreement as separate inferential targets.
We establish conditions under which this decomposition is identifiable and develop an anchored alternating algorithm that preserves the identifying geometry. 
For confidence quantification, we study a fixed-panel repeated-comparison regime in which the judge panel may remain fixed or modest while information grows through repeated judgments. 
This yields uncertainty statements for consensus and judge-specific ranking contrasts, sensitivity parameters, pairwise probabilities, and summaries of residual disagreement.
Experiments on synthetic and real multi-judge comparison data show that HJA improves recovery, robustness, uncertainty calibration, and near-tie performance relative to pooled and sensitivity-only baselines. 
The fitted model also provides diagnostics for judge disagreement and model-affinity patterns, giving a statistically grounded framework for ranking under heterogeneous comparative judgments.
\end{abstract}

\begin{bibunit}[unsrtnat]

\section{Introduction}

Comparative judgments have become a standard interface for evaluating and aligning modern language models. Pairwise preferences are used to build leaderboards, benchmark systems, and provide supervision for post-training and preference optimization \citep{zheng2023mtbench,christiano2017preferences,ouyang2022instructgpt,rafailov2023dpo}. 
Increasingly, however, these judgments are not produced by a single homogeneous evaluator, but by panels of human annotators, model-based judges, prompting protocols, or evaluation settings applied to the same collection of candidate outputs \citep{chen2013crowdbt,otani2016irt,zhao2025lmc,dhurandhar2024ranking}. 
In such settings, disagreement among judges is not merely exchangeable noise around one latent ranking. Recent evidence shows substantial variation across judge models, languages, tasks, and evaluation protocols, with especially large consequences for subjective evaluations and near-tie comparisons \citep{bavaresco2025llmjudges,fu2025multilingualjudge,guerdan2025indeterminacy,gao2025reevaluating}. 
Thus, moving from one judge to many does not simply increase sample size for a classical ranking problem; it changes the object of inference from a single pooled ordering to a structured multi-judge ranking problem.

To address this problem, we develop a heterogeneous judge-aware ranking framework that treats multi-judge comparisons as evidence for several distinct ranking targets rather than as noisy votes for one pooled ordering. 
As summarized in \Cref{fig:overview}, the framework separates a consensus ranking shared across judges, judge-specific sensitivity to that consensus, and structured residual disagreement across judges and items. 
This distinction is important because two judges may differ not only in how reliably they track the common ranking signal, but also in which types of items or responses they systematically favor. 
The resulting model supports uncertainty quantification for consensus scores, judge-specific comparisons, consensus sensitivity, and latent disagreement patterns.

A key technical distinction is that the low-rank structure is not used as an arbitrary latent preference space. 
As formalized in \Cref{sec:HJA}, HJA decomposes the judge-specific Bradley--Terry--Luce (BTL) score matrix as \(S=\gamma\mu^\top+UV^\top\), where \(\mu\) is the row-average consensus direction, \(\gamma_k\) is judge \(k\)'s loading on that direction, and \(UV^\top\) is constrained to represent residual disagreement orthogonal to consensus. 
Thus \(\gamma_k\) measures sensitivity to the shared ranking signal, rather than global judge quality, while the low-rank component captures structured deviations after consensus has been removed. 
The resulting targets are therefore not merely an improved predictor of pairwise preferences, but the identifiable components of this canonical decomposition.

The main technical challenge is to make this separation statistically well-defined, estimable, and inferentially useful. 
We first show that the consensus component, judge-specific consensus sensitivities, and structured disagreement subspace are identifiable under explicit normalization and rank conditions (\Cref{prop:ident}). 
We then estimate the model by constrained maximum likelihood and develop an anchored proximal alternating algorithm that preserves the identifying geometry of the decomposition while enjoying a local linear convergence guarantee (\Cref{alg:mle} and \Cref{thm:main_convergence}). 
For inference, we study a repeated-comparison asymptotic regime in which the number of judges may remain fixed or modest while information grows through repeated comparisons, and establish asymptotic normality for the constrained estimator (\Cref{thm:asymptotics}). 
This yields confidence intervals for smooth ranking targets, including consensus score contrasts, judge-specific score contrasts, consensus sensitivity, pairwise preference probabilities, and summaries of latent disagreement.


\begin{figure}[!t]
    \centering
    \includegraphics[width=\textwidth]{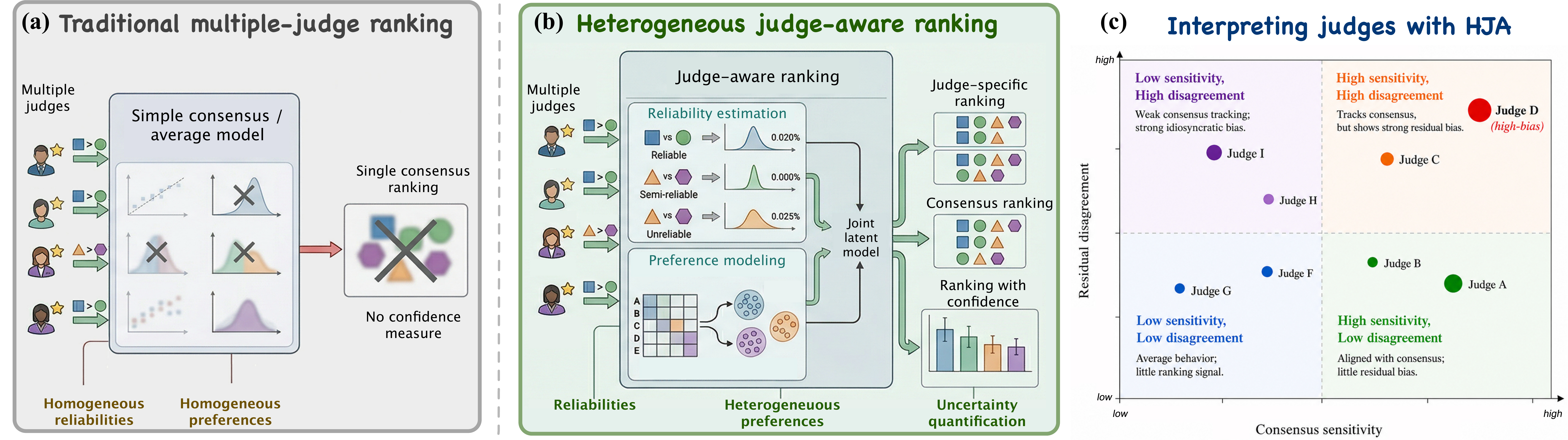}
    \caption{
    Overview of Heterogeneous Judge-Aware (HJA) ranking.
    Unlike pooled multi-judge ranking, which collapses heterogeneous judgments into a single consensus score, HJA models judge-specific latent scores through the structured decomposition.
    This yields consensus and judge-specific rankings, sensitivity assessment, structured disagreement analysis, and uncertainty quantification.
    }
    \label{fig:overview}
    \vspace{-2pt}
\end{figure}

Empirically, we evaluate HJA through synthetic recovery studies and real-data benchmarks (\Cref{subsec:comparison}). 
Synthetic experiments show that HJA recovers consensus rankings, judge-specific behavior, structured heterogeneity, and uncertainty calibration as comparison information increases, while misspecified baselines degrade under stronger heterogeneity. 
Across four real LLM-evaluation datasets, HJA improves robustness to noisy judges, held-out pairwise prediction, and near-tie performance, especially when judge disagreement is consequential. 
Beyond predictive benchmarking, we use HJA as a diagnostic tool to report judge-specific rankings with confidence intervals, identify high-leverage judges, and localize structured judge-model affinity patterns, including possible self- or family-affinity effects without treating them as universal findings (\Cref{subsec:real-analysis}).

\section{Related Work}
\vspace{-3pt}
\noindent\textbf{Pairwise ranking, annotator heterogeneity, and judge-aware aggregation.}
Our starting point is the BTL model for ranking from pairwise comparisons \citep{bradley1952rank,hunter2004mm}. 
Beyond the homogeneous BTL setting, prior work in crowdsourcing and LLM evaluation shows that aggregation should account for annotator-specific sensitivity, calibration, or sharpness when comparisons are collected from multiple annotators \citep{chen2013crowdbt,otani2016irt}. 
Motivated by modern multi-judge evaluation settings, we allow disagreement to involve structured preference heterogeneity beyond scalar sensitivity.

\noindent\textbf{LLM-as-a-judge and evaluator disagreement.}
Pairwise LLM judging became a practical standard through benchmarks and platforms such as MT-Bench and Chatbot Arena \citep{chiang2024chatbot,zheng2023mtbench}, and early evidence suggested that strong LLM judges could sometimes serve as useful proxies for human evaluation \citep{chiang2023alternative}. 
More recent work has weakened that optimistic view. Large-scale empirical studies report substantial variation across tasks and benchmarks \citep{bavaresco2025llmjudges}; multilingual evaluations show poor cross-language stability \citep{fu2025multilingualjudge}; and rating-indeterminacy analyses argue that collapsing multiple defensible judgments into a single forced label can itself introduce evaluation bias \citep{guerdan2025indeterminacy}. 
Near-tie settings are especially fragile, since small prompt or decoding changes can induce large changes in pairwise outcomes and hence in final leaderboards \citep{gao2025reevaluating}. 
These observations suggest that disagreement among judges should be explicitly modeled, rather than averaged away as if all judges were interchangeable.

\noindent\textbf{Judge-aware ranking and uncertainty quantification.}
JA-Ranking \citep{xu2026judgeaware} introduces judge-specific discrimination parameters for ranking LLM outputs without ground truth, capturing scalar sensitivity variation across judges. 
In our notation, this is the sensitivity-only structure \(S=\gamma\mu^\top\), the \(r=0\) special case of HJA. 
Collaborative ranking and preference-completion methods instead learn structured heterogeneous preferences from pairwise data \citep{lu2015individualized,park2015preference}, with related work imposing low-rank structure on pairwise preference probabilities or heterogeneous preference populations \citep{rajkumar2016can,jin2020rank,wu2015clustering}. 
These approaches show that preference heterogeneity can be low-dimensional, but typically treat the individualized or low-rank preference object as the primary target. 
HJA uses low-rank structure for a different inferential purpose: it canonically decomposes the judge-specific BTL score matrix into a row-average consensus direction, judge-specific sensitivity to that direction, and an orthogonal residual disagreement component. 
This yields separate targets for consensus ranking, judge-specific ranking, and latent disagreement, with fixed-panel uncertainty quantification for identifiable components rather than only pairwise prediction. 
Complementary work on comparative uncertainty argues that evaluation pipelines should retain confidence information rather than reduce judgments to deterministic verdicts \citep{fathullah2025generalised,wang2025judgmentdist}; HJA connects this perspective to structured multi-judge ranking with fixed or modest judge panels.
See \Cref{app:alternative-formulations} for a detailed comparison.

\section{Problem Setup}

\noindent\textbf{Pairwise judgments from multiple judges.}
We consider repeated pairwise comparison settings with multiple judges.
Let $[N]=\{1,\ldots,N\}$ and $[K]=\{1,\ldots,K\}$ index the items and the judges.
For each judge $k\in[K]$ and unordered item pair $(i,j)$ with $1\le i<j\le N$, $n_{kij}$ denotes the number of repeated comparisons between items $i$ and $j$, and $Y_{kij}\in\{0,\ldots,n_{kij}\}$ denotes the number of times judge $k$ prefers item $i$ to item $j$. 
In LLM evaluation and preference modeling, pairwise comparisons are routinely used to benchmark systems, rank models without reference answers, and train preference-aligned models \citep{zheng2023mtbench,dhurandhar2024ranking,christiano2017preferences,ouyang2022instructgpt,rafailov2023dpo}.
Write
\begin{align}
    \Omega := \{(k,i,j): 1\le k\le K,\ 1\le i<j\le N,\ n_{kij}>0\}
    \label{eq:Omega}
\end{align}
for the set of observed judge-item-pair triples.
Throughout, statistical information grows through repeated judgments on observed triples, while the number of judges $K$ may remain fixed or modest.

\noindent\textbf{Motivation.}  A pooled ranking model treats all judges as interchangeable and collapses their responses into a single latent score vector. But this is often too restrictive in multi-judge evaluation.
See, e.g., crowdsourced pairwise ranking and MT evaluation, where annotator-specific sensitivity or calibration affects aggregation \citep{chen2013crowdbt,otani2016irt}, and is amplified in LLM-as-a-judge settings with judge, language, and protocol sensitivity \citep{bavaresco2025llmjudges,fu2025multilingualjudge,gao2025reevaluating}.
Judges differ not only in how strongly they align with an overall ranking signal, but also in systematic ways that cannot be explained by a common score plus independent noise. For example, as will be shown in \Cref{subsec:model}, the well-documented difference in domain expertise \cite{szymanski2025limitations,li2026llmevaluation} can easily confuse the pooled-based models. These considerations suggest that the target of inference should not be a single pooled ranking alone, but a decomposition that preserves a consensus ranking while allowing structured judge-specific departures from it.

\noindent\textbf{Consensus ranking with sensitivity, preference heterogeneity, and confidence.}
A judge-aware ranking model should do more than output a single pooled ordering of the items. 
In many evaluation settings, the substantive questions are not only which items rank highest on average, but also which judges track the shared ranking signal most strongly, where systematic preference disagreement remains, and how confident we should be in the ranking conclusions.

Accordingly, we seek a representation that separates three statistical targets: 
(i) \emph{consensus ranking}, which summarizes the ordering most strongly supported across judges;
(ii) \emph{consensus sensitivity}, which is judge-specific sensitivity to the shared consensus ranking direction and measures how strongly a judge tracks the common ranking signal, rather than the judge's overall quality or calibration across all possible preference dimensions;
and (iii) \emph{structured preference heterogeneity}, which allows judges to differ not only in sharpness or noise level, but also in systematic preferences that cannot be explained by a scalar rescaling alone.
This perspective is aligned with collaborative ranking and preference-completion models, which treat individualized preferences as structured low-dimensional signals rather than independent noise \citep{lu2015individualized,park2015preference,fan2025uncertainty}.

In addition to recovering these latent components, the model should support confidence statements for the inferential outputs they induce. In particular, we would like to quantify uncertainty in the consensus ranking, in judge-specific comparisons, and in the estimated sensitivity and disagreement structure. This is especially important in near-tie regimes, where small perturbations in judgments can produce large changes in the final ordering \citep{guerdan2025indeterminacy,gao2025reevaluating}.
These considerations motivate a structured latent-score model that separates shared ranking signal from judge-specific alignment and residual disagreement. 
The next section introduces such a decomposition.

\section{Heterogeneous Judge-Aware Ranking Model}\label{sec:HJA}

\subsection{Judge-specific latent variable model} \label{subsec:model}
We model each judge as having a latent score vector over the $N$ items, with pairwise judgments determined by differences in those scores. 
Let $S\in\RR^{K\times N}$ denote the score matrix where the $k$th row $S_{k,:}$ represents the score vector for judge $k$.
Conditional on $S$, the counts $\{Y_{kij}:(k,i,j)\in\Omega\}$ are assumed to be independent and follow Bradley--Terry--Luce (BTL) model \citep{bradley1952rank,luce1959individual}:
\begin{equation}
    Y_{kij}\sim \mathrm{Binomial}(n_{kij},p_{kij}), \qquad \mathrm{logit}(p_{kij})=S_{ki}-S_{kj}, \qquad (k,i,j)\in\Omega. \label{eq:binomial-model}
\end{equation}
Equivalently, \(p_{kij}={\exp(S_{ki})}/\{\exp(S_{ki})+\exp(S_{kj})\}.\)
For each judge $k$, let $\cG_k=([N],\cE_k)$ denote the comparison graph with edge set
\(
\cE_k=\{\{i,j\}:(k,i,j)\in\Omega\}.
\)
For the main development, we assume that $\cG_k$ is connected for every $k$, as is standard in BTL estimation \citep{hunter2004mm}. 
This condition implies that each judge-specific score vector is identifiable up to an additive constant from that judge's observed comparisons. Together with the row-centering constraint
\begin{equation}
    S \one_N = \zero_K,
    \label{eq:S-id}
\end{equation}
this removes the usual additive invariance of the BTL model, so the collection $\{p_{kij}:(k,i,j)\in\Omega\}$ uniquely identifies $S$.
We allow for disconnected individual comparison graphs in the partial sampling regime of \Cref{subsec:partial-sample}, under additional identification constraints for the latent score matrix.

To separate shared ranking information from structured disagreement across judges, we model the latent score matrix through a consensus-plus-heterogeneity decomposition:
\begin{equation}
S = \gamma \mu^\top + UV^\top,
\label{eq:score-decomp}
\end{equation}
where $\mu \in \RR^N$ is the consensus item score vector, $\gamma \in \RR^K$ collects judge-specific consensus-sensitivity parameters, $U \in \RR^{K \times r}$ contains judge loadings on latent disagreement directions, and $V \in \RR^{N \times r}$ contains the corresponding item coordinates. Equivalently, one has
\begin{equation}
\mathrm{logit}(p_{kij})
=
\gamma_k(\mu_i-\mu_j)
+
U_{k,:}(V_{i,:}-V_{j,:})^\top,\qquad (k,i,j)\in\Omega.
\label{eq:judge-aware-link}
\end{equation}
This decomposition separates two distinct sources of variation across judges. The term
$\gamma_k(\mu_i-\mu_j)$ captures how strongly judge $k$ responds to the common consensus signal, and therefore determines how sharply judge $k$ discriminates along the shared ranking direction $\mu$. 
In the unlikely case where $\gamma_k<0$, the judge $k$ is anti-aligned with the consensus direction or may be \emph{malicious} \citep{chen2013crowdbt,jing2018hybrid}.
The residual term $U_{k,:}(V_{i,:}-V_{j,:})^\top$ captures systematic departures from that consensus that cannot be explained by a scalar rescaling alone. As a result, the model distinguishes judges who differ mainly in sharpness from judges who exhibit structured preference heterogeneity across items.

This distinction is particularly useful in LLM-as-a-judge settings. Some judges may be weakly aligned with the common ordering and therefore behave as comparatively noisy evaluators, while others may display systematic deviations tied to answer style, language, task domain, or other latent attributes. Under \eqref{eq:score-decomp}, the consensus vector $\mu$ represents the population-level ranking target, whereas $UV^\top$ encodes interpretable disagreement structure around that target.

We also offer an intuitive example of latent expertise to exemplify our notion of heterogeneity, and to show why simpler models fail in this common phenomenon. 
\begin{example}[Latent expertise under HJA]
    Consider an example of coding skill evaluation with three responses \(A,B,C\), where \(A\) is technically correct but terse, \(B\) is polished but contains a subtle fatal bug, and \(C\) is incomplete. 
        Judge \(1\) is a coding expert, whereas judges \(2\) and \(3\) are general-purpose evaluators. 
        Suppose the HJA model is correctly specified with judgment-specific preference given by        
        \begin{equation}
        \textstyle
            \Pr_k(i\succ j)=\sigma(S_{ki}-S_{kj}),
        \qquad
        S=\gamma\mu^\top+UV^\top,
        \end{equation}
        where \(S_{ki}\) denotes judge \(k\)'s latent BTL score for response \(i\), and ``$\succ$'' denotes preference over.
        With coordinates ordered as \((A,B,C)\), set
        \(
        \gamma= 1_3, \,
        \mu=(0,1,-1)^\top,\,
        U=(\sqrt 2/2,-\sqrt2/4,
        -\sqrt2/4)^\top,\,
        V=(\sqrt 2,-\sqrt2/2,
        -\sqrt2/2)^\top.
        \)
        The resulting score matrix is
        \[
        S=
        \begin{pmatrix}
        1 & 1/2 & -3/2\\
        -1/2 & 5/4 & -3/4\\
        -1/2 & 5/4 & -3/4
        \end{pmatrix}.
        \]
        Thus, the coding expert ranks \(A\succ B\succ C\), whereas both general-purpose judges rank \(B\succ A\succ C\).
        Under equal sampling across judges, the pooled probability that \(A\) defeats \(B\) is
        \[
        \bar p_{AB}
        =
        \{\sigma(1/2)+\sigma(-7/4)+\sigma(-7/4)\}/3
        \approx 0.306<1/2.
        \]
        Consequently, the population pooled-BTL fit ranks \(B\succ A\), and increasing the number of comparisons only makes this misleading conclusion more certain for an expertise-conditioned coding target. Since \(\gamma=\one_3\), a sensitivity-only model cannot represent the reversal either, because the conflict arises entirely from the latent expertise--response interaction \(UV^\top\).
        HJA instead separates the panel consensus \(B\succ A\succ C\) from the judge-specific conclusion \(A\succ B\), thereby exposing the disagreement that pooling removes. Our model in effect abstracts this type of phenomenon. 
\end{example}
Another example which can serve as an extension is to model the domain expertise $d$ in another dimension, when domain-specific comparisons are available, rather than aggregate all comparisons into a marginal judge-by-candidate score matrix. For example, one may have a mixture model
    \begin{equation} \label{eq:known-domain}
    \textstyle
    \operatorname{logit}(p_{kijd}) =
    \gamma_{kd}(\mu_{id}-\mu_{jd})
    +
    \sum_{\ell=1}^{r}
    U_{k\ell}W_{d\ell}(V_{i\ell}-V_{j\ell}),
    \end{equation}
    where \(W_{d\ell}\) links latent direction \(\ell\) to a domain such as coding or mathematics. The resulting judge-by-candidate-by-domain tensor would distinguish judges reliable in particular domains from candidate LLMs strong in those domains. \Cref{{subsec:latent-domain}}(B) and details in \Cref{app:subsec:latent-domain} showcase our model's performance under such misspecified data, when domain labels are hidden from the model.

\subsection{Identification of latent components}
\label{subsec:id}

The decomposition in \eqref{eq:score-decomp} introduces interpretable latent components for consensus, consensus sensitivity, and structured heterogeneity, but the factorization is not identified by the likelihood alone: many parameter tuples can generate the same row-centered score matrix \(S\). 
The role of the following constraints is therefore not merely to choose a convenient scale, rotation, or centering convention. 
They define the canonical representative in which \(\mu\) is the cross-judge consensus direction, \(\gamma\) measures judge-specific sensitivity to that direction, and \(UV^\top\) represents residual disagreement.

\begin{condition}[Normalization]\label{cond:norm}
    The parameters $(\gamma,\mu,U,V)$ satisfy: \begin{enumerate}[label=(\roman*)]
        \item\label{cond:center} {Centering:} $\one_N^\top \mu = 0$ and $V^\top \one_N = \zero_r$.
        \item\label{cond:scale} {Scaling:} $\one_K^\top U = \zero_r^\top$ and $\one_K^\top \gamma = K$.
        \item\label{cond:orth} {Orthogonality:} $\mu^\top V = \zero_r^\top$.
        \item\label{cond:subspace} {Subspace anchoring:} $N^{-1}V^\top V = I_r$ and $K^{-1}U^\top U = D$ for some diagonal matrix $D$ with strictly decreasing positive diagonal entries.
    \end{enumerate}
\end{condition}

\Cref{cond:norm}\ref{cond:center}-\ref{cond:scale} fix a canonical representative of the decomposition in \eqref{eq:score-decomp}. 
The centering constraints remove the additive indeterminacy inherited from the BTL model and align 
the decomposition with the row-centering normalization \(S\one_N=\zero_K\). The scaling constraints 
then calibrate the consensus component so that \(\mu\) represents the cross-judge consensus direction 
and \(\gamma\) measures judge-specific alignment with that direction. Indeed, under \Cref{cond:norm}\ref{cond:scale}, if \(S^\top\one_K\neq \zero_N\), then
\[
\one_K^\top S/K
=
\one_K^\top \gamma \mu^\top/K
+
\one_K^\top U V^\top/K
=
\mu^\top,
\qquad
\gamma
=
{S\mu}/{\|\mu\|_2^2}
=
{K}SS^\top \one_K \big/ {\|\one_K^\top S\|_2^2}.
\]
Thus \(\mu\) is the row-average BTL score direction induced by the judge-specific score matrix. 
Correspondingly, \(\gamma_k\) is the consensus sensitivity of judge \(k\)'s loading on this consensus direction.

\Cref{cond:norm}\ref{cond:orth}-\ref{cond:subspace} separate consensus from heterogeneity and anchor the residual factorization. 
The constraint \(\mu^\top V=\zero_r^\top\) requires the heterogeneous item directions to be orthogonal to the consensus direction, while the subspace-anchoring condition removes rotational ambiguity up to column-wise sign changes. 
Define the heterogeneous component of the latent score matrix:
\begin{equation} \label{eq:S-residual}
    \tilde{S} := SP_{S^\top \one_K}^\perp, 
\end{equation}
where $P_A^\perp$ denotes the projection matrix onto the orthogonal complement of $A$'s column space.
In the degenerate case \(S^\top\one_K=\zero_N\), the induced model reduces to pure heterogeneity.

Accordingly, for a fixed rank \(r\le R_{\max}:=\min\{K-1,N-2\}\), we define
\[
\Theta_r
=
\left\{
(\gamma,\mu,U,V):
\mu\in\RR^N,\ \gamma\in\RR^K,\ U\in\RR^{K\times r},\ V\in\RR^{N\times r},
\ \text{\Cref{cond:norm} holds}
\right\}.
\]
While \Cref{cond:norm} specifies a normalized parameterization, it does not guarantee that an arbitrary centered score matrix \(S\) admits such a decomposition. 
To recover admissible \((\gamma,\mu,U,V)\) uniquely from \(S\), we therefore impose regularity conditions on $S$ and its heterogeneous component \(\tilde{S}\).
\begin{condition}\label{cond:S-rank}
    The cross-judge aggregate score vector of $S$ is nonzero, i.e., \(S^\top\one_K\neq \zero_N\), and the heterogeneous score matrix $\tilde{S}$ has $r$ distinct positive singular values.
\end{condition}

The next proposition formalizes the identifiability of the model. 
First, the BTL comparison law identifies the row-centered judge-specific score matrix \(S\). 
Second, \Cref{cond:norm,cond:S-rank} recover from \(S\) the canonical consensus-plus-heterogeneity decomposition, uniquely up to column-wise sign changes in \((U,V)\). 
Compared to generic low-rank preference modeling \citep{rajkumar2016can,jin2020rank,wu2015clustering}, the latent factors in HJA are therefore not introduced merely as a flexible approximation, but are identifiable components with distinct meanings as consensus, consensus sensitivity, and residual disagreement.
\begin{proposition}[Identification] \label[proposition]{prop:ident}
    Consider the probabilistic preference model \eqref{eq:binomial-model} subject to \eqref{eq:S-id}. Then,
    \begin{enumerate}
        \item[(i)] (BTL identification)
        The latent score matrix $S$ can be identified from the generative probability distribution of pairwise comparisons in \eqref{eq:binomial-model}.
        
        \item[(ii)] 
        (Parameter identification) If \Cref{cond:S-rank} holds for $S$, then there exists $(\gamma, \mu,U,V)\in\Theta_r$ such that \eqref{eq:score-decomp} holds. Such existence is unique up to some sign changes in the columns of $U$ and $V$.
    \end{enumerate}
\end{proposition}

\Cref{prop:ident} shows that the model parameters correspond to uniquely defined features of the data-generating mechanism, rather than to an arbitrary factorization of \(S\). The only residual ambiguity can be eliminated via some sign conventions. This identification result provides the foundation for the constrained estimation and uncertainty quantification procedures developed in the next section.

\section{Estimation and Uncertainty Quantification}\label{sec:est-uq}
\subsection{Constrained maximum likelihood estimation}\label{subsec:estimation}

We first consider estimation for a fixed heterogeneity rank \(r \leq R_{\max}:=\min\{K-1,N-2\}\), with practical rank-selection criteria discussed in
\Cref{app:subsec:rank-selection}.
When the heterogeneity components are negligible or unnecessary, $r=0$ may be selected, and \Cref{app:zero-rank} shows our rank selection and inferential procedures can indeed spot the case where a simpler model is preferred.

Let \(\theta=(\gamma,\mu,U,V)\in\Theta_r\). 
For each observed triple
\((k,i,j)\in\Omega\), define
\(
\eta_{kij}(\theta)=
\gamma_k(\mu_i-\mu_j)
+
U_{k,:}(V_{i,:}-V_{j,:})^\top .
\)
Dropping constants independent of \(\theta\), the negative log-likelihood is
\begin{equation} \label{eq:ll}
\textstyle
\mathcal{L}_n(\theta)
:=
\sum_{(k,i,j)\in\Omega}
n_{kij}
\big[
-\bar{Y}_{kij}\eta_{kij}(\theta)
+
\log\{1+\exp(\eta_{kij}(\theta))\}
\big],
\end{equation}
where \(\bar{Y}_{kij}=Y_{kij}/n_{kij}\). When $r>0$, we estimate the identifiable components of the HJA model using the constrained maximum-likelihood estimator (MLE)
\begin{equation}
\widehat{\theta}
=
(\widehat{\gamma},\widehat{\mu},\widehat{U},\widehat{V})
\in
\arg\min_{\theta\in\Theta_r}
\mathcal{L}_n(\theta).
\label{eq:mle}
\end{equation}
The constraint set \(\Theta_r\) encodes the normalizations used in \Cref{subsec:id} to make the consensus direction, judge-specific consensus sensitivities, and low-rank disagreement component statistically identifiable, so that \(\widehat{\theta}\) estimates the canonical representative of the decomposition \eqref{eq:score-decomp}.
The resulting optimization problem is nonconvex because of the bilinear factorization in \((U,V)\) and the nonlinear anchoring constraints in \(\Theta_r\). In addition, direct alternating updates over the judge and item blocks do not preserve the nonlinear identification constraints, including \(\mu^\top V=0\), \(N^{-1}V^\top V=I_r\), and \(K^{-1}U^\top U\).
To address this, we develop a proximal alternating algorithm for constrained MLE that maintains the identifying structure of the decomposition while exploiting the blockwise smoothness of the loss function, shown in \Cref{alg:mle}. 
Each of its likelihood-and-gradient evaluation costs up to $\cO(KNr+N^2)$ in time complexity.
We combine proximal block updates with a re-anchoring step to obtain the canonical representative in \(\Theta_r\). 
\Cref{app:est-main} offers further implementation details.

\begin{algorithm}[t]
    \caption{
    Heterogeneous judge-aware (HJA) ranking by proximal anchored alternating MLE
    }
    \label{alg:mle}
    \small
    \begin{algorithmic}[1]
    \Require Data $\{(n_{kij},Y_{kij})\}_{(k,i,j)\in\Omega}$, rank \(r\), tolerance \(\varepsilon\), iterations \(T\), and parameters \(\{\tau^{(t)}\}_{t\geq 0}\).
    \Ensure Estimated parameter $\hat\theta$, Fisher information matrix $\cI_n(\hat\theta)$, and a smooth target estimand \(q\).
    \State $\theta^{(0)} \gets \textsc{Initialize}(\{(n_{kij},Y_{kij})\}_{(k,i,j)\in\Omega},r)$ based on \Cref{alg:initialize}.
    \For{\(t=0,1,\ldots,T-1\)}
        \State Set \(\tau \gets \tau^{(t)}\).
        \State Compute \((\tilde{\gamma},\tilde{U})\) as the unique minimizer of
        \(
        \mathcal{L}_n(\gamma,\mu^{(t)},U,V^{(t)})
        + \frac{\tau}{2}\|\gamma-\gamma^{(t)}\|_2^2
        + \frac{\tau}{2}\|U-U^{(t)}\|_F^2
        \)
        subject to \(\one_K^\top \gamma=K\) and \(\one_K^\top U=\zero_r^\top\). \label{line:argmin1}
        \State Compute \((\tilde{\mu},\tilde{V})\) as the unique minimizer of
        \(
        \mathcal{L}_n(\tilde{\gamma},\mu,\tilde{U},V)
        + \frac{\tau}{2}\|\mu-\mu^{(t)}\|_2^2
        + \frac{\tau}{2}\|V-V^{(t)}\|_F^2
        \)
        subject to \(\one_N^\top \mu=0\) and \(\one_N^\top V=\zero_r^\top\).\label{line:argmin2}
        \State \(\theta^{(t+1)} \gets \textsc{ReAnchor}(\tilde{\gamma},\tilde{\mu},\tilde{U},\tilde{V})\) based on \Cref{alg:reanchor}.
        
        \State \textbf{If}
        \(
        {
        |\mathcal{L}_n(\theta^{(t+1)})-\mathcal{L}_n(\theta^{(t)})|
        }/{
        [1+|\mathcal{L}_n(\theta^{(t)})|]
        }
        <\varepsilon
        \)
        \textbf{then break}.
    \EndFor
    \State \(\widehat{\theta}\gets \theta^{(t+1)}\) and $\cI_n(\hat\theta) \gets \textsc{EstimateFisherInfo}(\hat\theta)$ based on \Cref{alg:covariance}.
    \State Compute \(\widehat q=q(\widehat\theta)\), \(\widehat{\mathrm{se}}(\widehat q)\), and the Wald interval \(\widehat q\pm z_{1-\alpha/2}\widehat{\mathrm{se}}(\widehat q)\) using \Cref{alg:uq}.
    \end{algorithmic}
\end{algorithm}

To analyze \Cref{alg:mle} in the local nonconvex regime, we work in an identified chart around the target constrained MLE. 
\Cref{asm:algo-chart} ensures that the re-anchoring map is locally well defined and nondegenerate; \Cref{asm:algo-local-min} imposes an isolated constrained local minimizer with positive curvature in this chart; and \Cref{asm:algo-init} places the initialization in its basin of attraction.
Let \(\widehat{\theta}\) be defined by \eqref{eq:mle}, and let \(\{\theta^{(t)}\}_{t\geq 0}\) denote the anchored sequence generated by \Cref{alg:mle}. 
Throughout the remainder of \Cref{sec:est-uq}, we fix the local column signs of \((U,V)\) by the anchoring convention (e.g., by \Cref{rem:sign-id}) and work in the corresponding local identified parameterization.

\begin{theorem}[Locally linear convergence]
\label{thm:main_convergence}    
    Under \Cref{asm:algo-chart,asm:algo-local-min,asm:algo-init}, the iterates of \Cref{alg:mle} remain in the identified local chart, the likelihood
    \(\{\mathcal{L}_n(\theta^{(t)})\}_{t\geq 0}\) are nonincreasing, and
    \(
        \theta^{(t)} \to \widehat{\theta} .
    \) 
    Moreover, there exists \(\rho\in(0,1)\) such that for \(t\) large enough, 
    \(
        \|\theta^{(t)}-\widehat{\theta}\|_2
        +
        \mathcal{L}_n(\theta^{(t)})-\mathcal{L}_n(\widehat{\theta})
        \lesssim
        \rho^{t}.
    \)
\end{theorem}

\Cref{thm:main_convergence} suggests that the re-anchored alternating updates of \Cref{alg:mle} are stable within the identified local chart, i.e., the objective decreases monotonically, the iterates converge to the intended constrained MLE, and the final local convergence is geometric. See \Cref{app:est:subsec:conv} for its proof.

\vspace{-1mm}
\subsection{Uncertainty quantification} \label{subsec:UQ}
\vspace{-1mm}
Point estimates alone are insufficient for judge-aware ranking: in near-tie regimes, small perturbations in the observed comparisons can change the consensus ordering, judge-specific rankings, or estimated disagreement structure. 
We therefore quantify uncertainty for the score contrasts that determine ranking conclusions, rather than only for the latent components themselves. 
We study this problem in a fixed-rank repeated-comparison regime, where \(K\) and \(N\) may remain fixed while information accumulates through repeated comparisons on observed triples. 
This regime matches LLM evaluation settings in which the same judge panel can be queried repeatedly across prompts, random seeds, or judging instances, even when adding trusted judges is costly.

We impose the following regularity conditions for statistical inference in this setting. 
\Cref{asm:sampling} assumes stable asymptotic sampling fractions \(n_{kij}/n\to w_{kij}\), and \Cref{asm:graph} requires the active comparison graph for each judge to be connected.
They identify the row-centered score matrix from active pairwise differences and pass this information to model components. 

With the local sign convention fixed as in \Cref{subsec:estimation} and
\(\Omega_w\) defined in \eqref{eq:active-graph}, we define the Fisher information
for the anchored true parameter \(\theta^*\in\Theta_r\) by
\begin{align}
    \textstyle
    \cI(\theta^*)
    =
    \sum_{(k,i,j)\in\Omega_w}
    w_{kij} p_{kij}^*(1-p_{kij}^*)
    \nabla_{\theta}\eta_{kij}(\theta^*)
    \nabla_{\theta}\eta_{kij}(\theta^*)^\top .\label{eq:fisher-aggregate}
\end{align}
The next theorem shows the inferential properties of the estimate and its smooth target functionals. 

\begin{theorem}[Fixed-rank asymptotics]
\label{thm:asymptotics}
    For the observed-data model \eqref{eq:binomial-model} and \eqref{eq:score-decomp} with $\mu^* \ne 0$,
    under
    \Cref{asm:sampling,asm:graph},
    the constrained MLE \eqref{eq:mle}  satisfies
    \(\widehat{\theta}\pto\theta^*\),
    \(\cI(\theta^*)\) is positive definite, and
    \[
        \sqrt{n}\big(\widehat{\theta}-\theta^*\big)
        \dto
        \mathcal{N}\!\left(
            0,
            \cI(\theta^*)^{-1}
        \right).
    \]
    Consequently, for any smooth scalar target \(q:\Theta_r\to\mathbb{R}\),
    \[
        \sqrt{n}\{q(\widehat{\theta})-q(\theta^*)\}
        \dto
        \mathcal{N}\!\left(
            0,
            \nabla_{\theta}q(\theta^*)^\top
            \cI(\theta^*)^{-1}
            \nabla_{\theta}q(\theta^*)
        \right).
    \]
\end{theorem}

\Cref{thm:asymptotics} shows that repeated judgments can support inference on the structured components of the HJA model even when the judge panel itself does not grow.
The Fisher information in \eqref{eq:fisher-aggregate} can be consistently estimated by the plug-in observed design estimator
\begin{equation}
\label{eq:fisher-plugin}
    \mathcal{I}_{n}(\widehat{\theta})
    =
    {\textstyle\sum_{(k,i,j)\in\Omega}}
    \frac{n_{kij}}{n}
    \widehat p_{kij}(1-\widehat p_{kij})
    \nabla_{\theta}\eta_{kij}(\widehat{\theta})
    \nabla_{\theta}\eta_{kij}(\widehat{\theta})^\top ,
\end{equation}
where
\(
    \widehat p_{kij}
    =
    \{1+\exp[-\eta_{kij}(\widehat{\theta})]\}^{-1}.
\)
Under the assumptions of \Cref{thm:asymptotics},
\(
    \mathcal{I}_{n}(\widehat{\theta})
    \pto
    \mathcal{I}(\theta^*).
\)
Thus Wald standard errors for smooth targets \(q(\theta)\) are given by
\(
    \widehat{\mathrm{se}}\{q(\widehat{\theta})\}
    =
    [
    \frac{1}{n}
    \nabla q(\widehat{\theta})^\top
    \mathcal{I}_n(\widehat{\theta})^{-1}
    \nabla q(\widehat{\theta})
    ]^{1/2}.
\)
Typical choices of \(q\) include consensus contrasts \(\mu_i-\mu_j\), judge-specific contrasts \(S_{ki}-S_{kj}\), consensus-sensitivity parameters \(\gamma_k\), pairwise probabilities \(p_{kij}\), and smooth summaries of \(UV^\top\).
These are the natural inferential objects for ranking, as full rankings are discontinuous at ties, whereas pairwise score differences remain smooth and directly quantify ranking uncertainty.
Note that when $r^*=0$, the inferential results above still apply to $\mu$ and $\gamma$.

\section{Empirical Evaluation}\label{sec:emp}

\subsection{Benchmarking across synthetic and real data}\label{subsec:comparison}
\noindent\textbf{(A) Synthetic simulations.}
We first use synthetic data to isolate the effect of structured judge heterogeneity under a known latent score matrix and known consensus-plus-heterogeneity decomposition. 
Data are generated from \eqref{eq:binomial-model} and \eqref{eq:score-decomp} with \(N=8\) items, \(K=4\) judges, rank \(r=1\), and 50 random seeds. We compare 4 methods: HJA, JA-Ranking~\citep{xu2026judgeaware}, pooled BTL \((\gamma=\one_K,r=0)\), and an unstructured judge-wise BTL estimator followed by a rank-\((r+1)\) truncated SVD of the estimated score matrix; and consider 2 regimes: a sample-size experiment varying the total number of comparisons \(n_{\mathrm{cmp}}\) at fixed heterogeneity scale \(h=1\), and a heterogeneity experiment varying \(h\) at fixed \(n_{\mathrm{cmp}}=800\), where \(h\) scales the residual component \(UV^\top\). 
Full data-generating details, metric definitions, convergence diagnostics, and sensitivity-only simulations are deferred to \Cref{app:simulation}.

Figure~\ref{fig:synthetic} shows that HJA improves score recovery, consensus ranking, uncertainty calibration, and judge-specific sign prediction as comparison information increases. The advantage is most pronounced as heterogeneity grows: pooled BTL and JA-Ranking are structurally misspecified because they cannot represent residual judge-item preference structure, and hence concentrate around biased score matrices in larger samples. 
The unstructured BTL+SVD baseline is less accurate and has higher score-level variance because it estimates the full score matrix before imposing low-rank structure. 
The ablation in \Cref{app:subsec:ablation} further shows that the structured optimization and reanchor step are necessary to stabilize the identified decomposition and improve score-level recovery.

\begin{figure}[!t]
    \centering
    \includegraphics[width=0.9\linewidth]{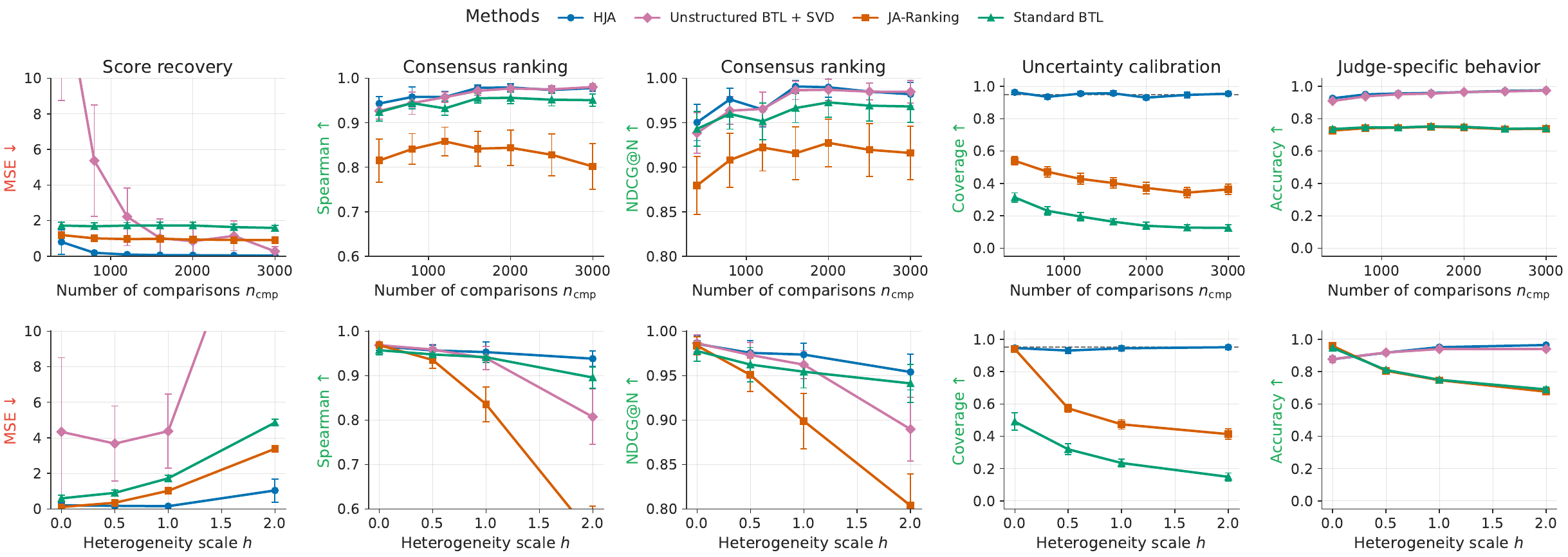}
    \caption{
    Benchmarking on synthetic simulations, with (a) varying numbers of pairwise comparisons \(T \in \{400,800,1200,1600,2000,2500,3000\}\) at heterogeneity scale 1.0, and (b) varying heterogeneity scales over \(\{0,0.5,1,2\}\) and fixed \(T=800\) comparisons. 
    Each panel reports the mean over 50 independent repetitions with 95\% normal-approximation error bars. Columns show performance on score recovery (MSE), consensus rank quality (Spearman and NDCG@N), uncertainty calibration (nominal 95\% score-entry coverage), and judge-specific behavior (sign accuracy).
    }\label{fig:synthetic}
    \vspace{-2mm}
\end{figure}

\begin{table*}[!b]
    \centering
    \scriptsize
    \setlength{\tabcolsep}{2.4pt}
    \renewcommand{\arraystretch}{1.0}
    \caption{
    Real-data comparison. Robustness is evaluated under increasing noisy-judge perturbation; near-tie accuracy is evaluated across increasingly less ambiguous near-tie groups.
    Entries are mean \(\pm\) standard deviation in 20 random seeds; bold marks the best mean per dataset-metric group.
    }\label{tab:dataset_comparison}
    \begin{tabular*}{\textwidth}{@{\extracolsep{\fill}}llcccccccc@{}}
    \toprule
    \multirow{2}{*}{\textbf{Dataset}}
    & \multirow{2}{*}{\textbf{Method}}
    & \multicolumn{3}{c}{\textbf{Noisy-judge robustness}}
    & \multirow{2}{*}{\shortstack[l]{\textbf{Hold-out}\\ \textbf{accuracy}}}
    & \multicolumn{3}{c}{\textbf{Near-tie accuracy}} \\
    \cmidrule(lr){3-5} \cmidrule(lr){7-9}
    & & Low & Mid & High & & Closest & Mid & Farthest \\
    \midrule
    
    \multirow{3}{*}{\shortstack[l]{Chatbot\\Arena}}
    & HJA
    & \bestpmv{0.95}{0.07}
    & \bestpmv{0.85}{0.09}
    & \bestpmv{0.82}{0.07}
    & \bestpmv{0.65}{0.01}
    & \bestpmv{0.68}{0.05}
    & \bestpmv{0.69}{0.04}
    & \bestpmv{0.67}{0.06} \\
    & JA-Ranking
    & \pmv{0.77}{0.09}
    & \pmv{0.51}{0.12}
    & \pmv{0.46}{0.11}
    & \pmv{0.58}{0.01}
    & \pmv{0.58}{0.03}
    & \pmv{0.54}{0.05}
    & \pmv{0.50}{0.06} \\
    & BTL
    & \pmv{0.47}{0.16}
    & \pmv{0.33}{0.13}
    & \pmv{0.21}{0.11}
    & \pmv{0.58}{0.01}
    & \pmv{0.57}{0.04}
    & \pmv{0.54}{0.05}
    & \pmv{0.51}{0.06} \\
    
    \cmidrule{2-9}
    
    \multirow{3}{*}{MT-Bench}
    & HJA
    & \bestpmv{1.00}{0.00}
    & \bestpmv{1.00}{0.00}
    & \bestpmv{1.00}{0.00}
    & \bestpmv{0.76}{0.01}
    & \bestpmv{0.74}{0.02}
    & \bestpmv{0.84}{0.01}
    & \bestpmv{0.89}{0.01} \\
    & JA-Ranking
    & \pmv{0.67}{0.00}
    & \pmv{0.67}{0.00}
    & \pmv{0.67}{0.00}
    & \pmv{0.70}{0.01}
    & \pmv{0.62}{0.02}
    & \pmv{0.78}{0.01}
    & \pmv{0.86}{0.01} \\
    & BTL
    & \pmv{0.92}{0.15}
    & \pmv{0.87}{0.17}
    & \pmv{0.85}{0.20}
    & \pmv{0.70}{0.01}
    & \pmv{0.62}{0.02}
    & \pmv{0.78}{0.01}
    & \pmv{0.86}{0.01} \\
    
    \cmidrule{2-9}
    
    \multirow{3}{*}{\shortstack[l]{Ultra-\\Feedback}}
    & HJA
    & \pmv{0.81}{0.06}
    & \bestpmv{0.65}{0.11}
    & \bestpmv{0.67}{0.13}
    & \bestpmv{0.70}{0.01}
    & \bestpmv{0.69}{0.05}
    & \bestpmv{0.68}{0.04}
    & \bestpmv{0.71}{0.06} \\
    & JA-Ranking
    & \bestpmv{0.83}{0.11}
    & \pmv{0.63}{0.11}
    & \pmv{0.62}{0.16}
    & \pmv{0.62}{0.01}
    & \pmv{0.55}{0.06}
    & \pmv{0.58}{0.06}
    & \pmv{0.60}{0.04} \\
    & BTL
    & \pmv{0.57}{0.09}
    & \pmv{0.32}{0.12}
    & \pmv{0.30}{0.10}
    & \pmv{0.61}{0.01}
    & \pmv{0.55}{0.06}
    & \pmv{0.58}{0.06}
    & \pmv{0.60}{0.04} \\
    
    \cmidrule{2-9}
    
    \multirow{3}{*}{In-house}
    & HJA
    & \bestpmv{0.94}{0.05}
    & \bestpmv{0.83}{0.05}
    & \pmv{0.33}{0.05}
    & \bestpmv{0.62}{0.01}
    & \bestpmv{0.60}{0.07}
    & \bestpmv{0.56}{0.10}
    & \pmv{0.58}{0.10} \\
    & JA-Ranking
    & \pmv{0.65}{0.06}
    & \pmv{0.43}{0.09}
    & \bestpmv{0.36}{0.08}
    & \pmv{0.61}{0.01}
    & \pmv{0.58}{0.07}
    & \bestpmv{0.56}{0.08}
    & \bestpmv{0.59}{0.09} \\
    & BTL
    & \pmv{0.35}{0.08}
    & \pmv{0.18}{0.06}
    & \pmv{0.12}{0.04}
    & \pmv{0.61}{0.01}
    & \pmv{0.57}{0.07}
    & \bestpmv{0.56}{0.08}
    & \bestpmv{0.59}{0.09} \\
    
    \bottomrule
    \end{tabular*}
\end{table*}

\noindent\textbf{(B) Misspecified synthetic data.} 
\label{subsec:latent-domain}
We use a hidden-expertise setting with $N=12$ items and $K = 9$ judges to illustrate how our model can identify the potential latent domains of judges' expertise, when the synthetic data do not agree with our model. With the two preference (expertise) blocs (6 maths-reliable, 3 coding-reliable, each judge nearly random outside its own domain) expressed by a mixture model based on \eqref{eq:known-domain} with its $U = 0$, the likelihood cannot be expressed via \eqref{eq:ll}, thus lying outside our model's data generating mechanism. Every judge contributes 4,800 comparisons, and no domain label $d$ is supplied to any model. The details in \Cref{app:subsec:latent-domain} will show that a pooled-ranking can reverse the pattern, while HJA's $UV^\top$ term approximates the pattern accurately with correct interpretation; indeed, across 50 repetitions, HJA correctly recovers the ranking in 48 runs, versus 35/0 for the sensitivity-only/pooled-BTL, with high held-out accuracy. We show that HJA both mitigates the reversal and identifies the (misspecified) latent structure responsible for it.

\noindent\textbf{(C) Real-data robustness evaluation.}
\label{subsec:real-evaluation}
We next use four representative real-world evaluation datasets to test whether modeling structured judge heterogeneity improves ranking under noisy, uneven, and near-tie judgments: Chatbot Arena \citep{chiang2024chatbot}, MT-Bench \citep{zheng2023mtbench}, UltraFeedback \citep{cui2024ultrafeedback}, and an in-house judgment dataset \citep{xu2026judgeaware}. 
We compare the induced consensus rankings by HJA, JA-Ranking, and BTL along three axes: robustness to injected noisy judges, held-out pairwise prediction, and accuracy on empirically ambiguous near-tie pairs, where automatic rankings are especially fragile \citep{guerdan2025indeterminacy,gao2025reevaluating}.
Noisy-judge robustness measures how stable the estimated ranking remains as additional noisy judges are injected; held-out prediction evaluates record-level generalization; and near-tie accuracy focuses on item pairs with empirical win rates close to \(0.5\), where small modeling biases can flip the induced ordering. 
See \Cref{app:subsec:real-evaluation} for dataset and experimental details. 

\Cref{tab:dataset_comparison} shows that HJA performs strongly on the three public benchmarks,  Chatbot Arena, MT-Bench, and UltraFeedback, with clear gains in held-out prediction and consistent improvements in robustness and near-tie accuracy. On the in-house dataset \cite{xu2026judgeaware}, though HJA does not dominate JA-Ranking under the strongest perturbations or all near-tie slices, it remains competitive. 
Overall, these results support the benefit of modeling structural judge heterogeneity beyond scalar sensitivity. 

\noindent\textbf{(D) Stress test.} \label{subsec:stress-test}
We inject a source with a systematic preference for a fixed candidate cluster and with approximately the comparison volume of the entire clean panel to each of the four datasets, mimicking, say, strong self-preference bias \cite{chen2025beyond}. Table~\ref{tab:cluster_stress} shows that HJA better preserves clean rankings on the three public datasets. In in-house data, however, all methods reverse the ranking due to the volume of the biased injection.
See implementation details in \Cref{app:subsec:real-evaluation}.

\begin{table}[!t]
\centering
\caption{Stress-test ranking preservation: Spearman correlation with each method's clean ranking (mean $\pm$ sd over 20 replications)}
\label{tab:cluster_stress}
\footnotesize
\begin{tabular}{lccc}
\toprule
Dataset & HJA & Sensitivity-only & Pooled BTL \\
\midrule
Chatbot Arena & $\mathbf{\hphantom{-}0.1748\pm0.0164}$ & $\hphantom{-}0.0553\pm0.0102$ & $\hphantom{-}0.0339\pm0.0073$ \\
MT-Bench & $\mathbf{\hphantom{-}0.2343\pm0.0469}$ & $-0.1229\pm0.0280$ & $-0.1086\pm0.0287$ \\
UltraFeedback & $\mathbf{\hphantom{-}0.6245\pm0.3283}$ & $-0.2705\pm0.0130$ & $-0.2636\pm0.0145$ \\
In-house & $\mathbf{-0.3263\pm0.0046}$ & $-0.3324\pm0.0024$ & $-0.3293\pm0.0020$ \\
\bottomrule
\end{tabular}
\end{table}

\subsection{Consensus, sensitivity, and disagreement in real data}
\label{subsec:real-analysis}
\vspace{-2mm}
\begin{figure}[!t]
    \centering
    \includegraphics[width=\linewidth]{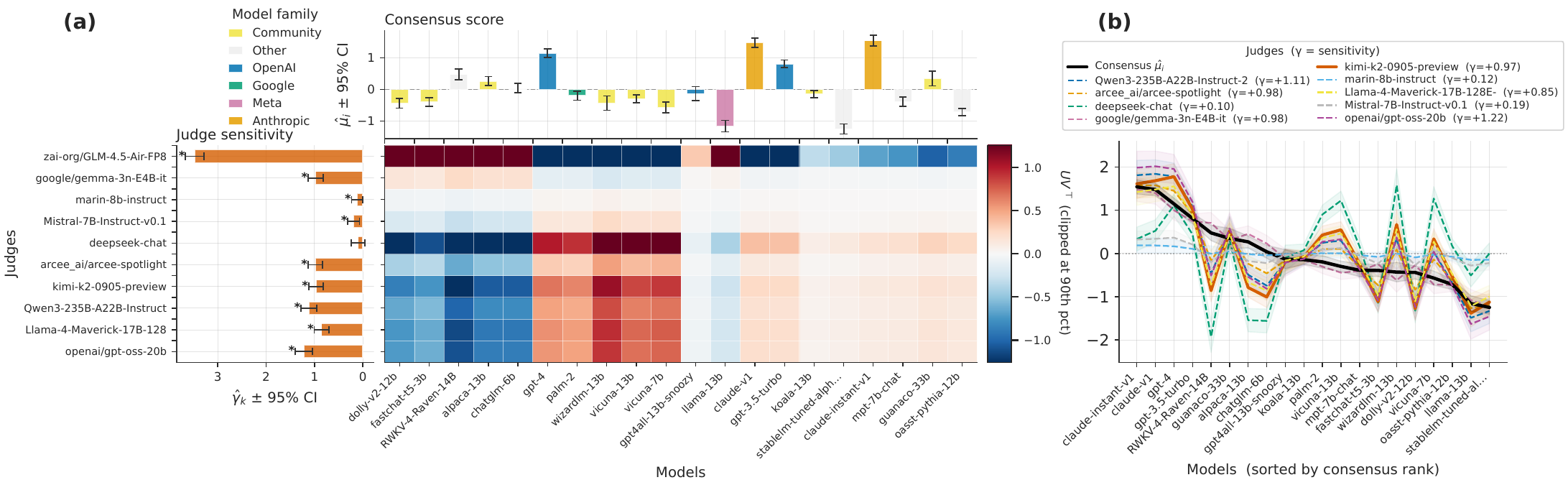}
    \caption{
    HJA diagnostic analysis on Chatbot Arena.
    (a) Heatmap shows the heterogeneous preference matrix \(\hat U\hat V^\top\); barplots show judge sensitivity \(\hat\gamma\) and consensus score \(\hat\mu\), with 95\% confidence intervals.
    The models are colored by LLM model families.
    (b) Judge-specific score profiles \(\hat S_{k,:}\), ordered by consensus score \(\hat\mu\) with the high-leverage outlier excluded.
    }
    \label{fig:chatbot_arena}
    \vspace{-2pt}
\end{figure}

Beyond benchmark comparison, we use HJA as a full-data diagnostic tool for real multi-judge evaluation. 
We focus on Chatbot Arena, a large crowdsourced LLM comparison dataset, and report analogous analyses for the other real datasets in \Cref{app:subsec:interpretation}. 
The fitted decomposition separates three quantities: the consensus ranking, judge sensitivity, and residual judge-model interactions. As shown in \Cref{fig:chatbot_arena}(a), this separation is useful for judge assessment: a large \(\hat\gamma_k\) indicates a strong response to the consensus direction, but reliability also requires small residual leverage. 
The judge \texttt{zai-org/GLM-4.5-Air-FP8} has both unusually large consensus sensitivity and unusually large residual disagreement, making it a high-leverage judge rather than simply a strong judge; this agrees with the broader diagnostic in \Cref{fig:leverage-glm}.
HJA also provides judge-specific rankings through the fitted score profiles
\(
\hat S_{k,:}=\hat\gamma_k\hat\mu^\top+\hat U_{k,:}\hat V^\top .
\)
\Cref{fig:chatbot_arena}(b) shows these profiles ordered by the consensus ranking. Unlike fitting each judge independently, HJA estimates these rankings jointly through a shared consensus direction and a low-rank disagreement structure. This pooling is important in sparse or noisy pairwise data: judge-specific deviations are allowed, but they are regularized through common latent disagreement directions rather than estimated as unrelated score vectors.

Finally, the residual component \(\hat U\hat V^\top\) localizes structured judge-model affinities. After identifying high-leverage outliers, we can refit or visualize the model with those judges filtered to obtain a more stable consensus and a less saturated interaction heatmap; see \Cref{fig:chatbot-uvt-heatmap}. The same residual structure can be inspected by clustering judges and models, as in \Cref{fig:bicluster,fig:biclustering-all}, to reveal groups with similar evaluation behavior. These analyses suggest that HJA can flag unreliable or high-leverage judges, recover judge-specific ranking preferences, and expose systematic affinity patterns, while treating self- or family-preference effects as diagnostic hypotheses rather than universal conclusions.

\vspace{-2mm}
\section{Conclusion and discussion}
\label{sec:conclusion}
\vspace{-1mm}
We introduced HJA ranking, an identifiable multi-judge BTL framework that separates consensus ranking, judge sensitivity, and structured residual disagreement. 
With constrained estimation and fixed-panel uncertainty quantification, HJA yields rankings that are more robust, interpretable, and diagnostic under heterogeneous judgments. 
However, the framework is not intended to certify rankings under adversarial judge panels: severe coordinated bias may affect the consensus estimate, and the nonconvex optimization remains locally initialized. Hence, deliberate manipulations may cause HJA to output biased rankings, which poses a potential negative impact.
Thus, HJA should be viewed as a statistically grounded diagnostic for heterogeneous evaluation before benchmark design.

\section*{Acknowledgements and disclosures}
We appreciate the Area Chair and reviewers' thoughtful suggestions, which greatly improved the content and clarity of this paper. 
The last author acknowledges support from the Institute of Data Science (IDS) at the University of Hong Kong through the HKU100 award.
\putbib[references]
\end{bibunit}

\clearpage

\appendix
\begin{bibunit}[unsrtnat]
\makeatletter
\def\@extra@b@citeb{.appendix}
\def\@extra@binfo{.appendix}
\makeatother
\counterwithin{theorem}{section}
\counterwithin{equation}{section}

\renewcommand{\thesection}{\Alph{section}}
\setcounter{table}{0}
\renewcommand{\thetable}{\thesection\arabic{table}}
\setcounter{figure}{0} 
\renewcommand\thefigure{\thesection\arabic{figure}}
\setcounter{algorithm}{0}
\renewcommand{\thealgorithm}{\thesection.\arabic{algorithm}}
\renewcommand{\theHalgorithm}{\thesection.\arabic{algorithm}}

\begin{center}
\Large
{\bf Appendix}
\end{center}

This serves as an appendix to the main paper.
Below, we provide an outline for the appendix along with a summary of the notation used in the main paper and the appendix.

\paragraph{Organization.}
The content of the appendix is organized as follows.

\begin{table}[!ht]
\centering \small
\begin{tabularx}{0.95\textwidth}{l l D}
    \toprule    
    \multicolumn{2}{c}{\textbf{Appendix}} & \textbf{Content} \\
    \midrule
    \Cref{app:alternative-formulations} &  & Relation to alternative formulations.\\\midrule \addlinespace[0.5ex]
    \Cref{app:ident} & \ref{app:ident:prop:ident} & Proof of \Cref{prop:ident} on model identification.\\
    \midrule \addlinespace[0.5ex]     
    \multirow{3}{*}{\Cref{app:est}} &\ref{app:subsec:rank-selection} & Rank selection and reduction to simpler models.\\
    & \ref{app:est-main}& Algorithms and implementation details. \\    
    & \ref{app:est:subsec:conv} & Proof of \Cref{thm:main_convergence} on convergence analysis. \\    
    \midrule \addlinespace[0.5ex]
    \multirow{4}{*}{\Cref{app:stat}} & \ref{app:consistency} & Proof of consistency part of \Cref{thm:asymptotics}. \\
    & \ref{app:normality} & Proof of asymptotic normality part in \Cref{thm:asymptotics} via graph theory. \\
    & \ref{app:lemmas} & Auxiliary lemmas used for the proofs.\\
    & \ref{subsec:partial-sample} & Partial-sampling (pooled-graph treatment) regime. \\\midrule \addlinespace[0.5ex]

    \multirow{5}{*}{\Cref{app:simulation}} & \ref{app:sim-main-dgp} & Main heterogeneous DGP specification.\\
    & \ref{app:sim-protocol} & Comparison design, fitting protocol, and metrics.\\
    & \ref{app:subsec:sim-converge} & Convergence diagnostics.\\
    & \ref{app:sim-ja-dgp} & JA-Ranking DGP specification and results.\\
    & \ref{app:subsec:ablation} & Ablation study. \\
    & \ref{app:subsec:latent-domain} & Latent-expertise (misspecified) data experiment.\\
    \midrule \addlinespace[0.5ex]
    
    \multirow{2}{*}{\Cref{app:sec:real}} & \ref{app:subsec:real-evaluation} & Real-data evaluation details.\\
    & \ref{app:subsec:interpretation} & Interpretation on real data analyses. \\

    \bottomrule
\end{tabularx}
\end{table}
\addcontentsline{toc}{part}{\appendixname}

\paragraph{Notation.}
An overview of some general notations used in the main paper and the appendix is as follows.

In \(\RR^d\), the $j$-th standard basis vector is denoted by $e_j$ and the zero (or all-one) vector is denoted by \(\zero_d\) (or \(\one_d\)) or simply \(\zero\) (or \(\one\)) if the dimension is clear from the context.
The cardinality of a set $\cS$ is denoted by $|\cS|$. The indicator function is denoted by $\ind_{\{\cdot\}}$.



We use ``$o$'' and ``$\cO$'' to denote the little-o and big-O notations; ``$\op$'' and ``$\Op$'' are their probabilistic counterparts. For sequences $\{a_n\}$ and $\{b_n\}$, we write $a_n\lesssim b_n$ if $a_n=\cO(b_n)$; and $a_n\asymp b_n$ if $a_n=\cO(b_n)$ and $b_n=\cO(a_n)$. 
Convergence almost surely, in probability, and in distribution is denoted by ``$\xrightarrow{\mathrm{a.s.}}$'', ``$\pto$'', and ``$\dto$'', respectively. When ``positivity'' property is imposed or stated on a matrix, it means positive definiteness. We use $\mathrm{KL}(p\,\|\,q)$ to denote the Kullback–Leibler (KL) divergence between two distributions $p$ and $q$.

\paragraph{Reproducibility.} All experiment results in this paper can be reproduced with datasets openly accessible. The code files are available at \url{https://github.com/JinHongDu-Lab/HJA-Ranking}.

\clearpage
\section{Relation to Alternative Formulations}
\label{app:alternative-formulations}

\noindent\textbf{Relation to JA-Ranking.}
JA-Ranking \citep{xu2026judgeaware} and related judge-aware BTL models represent judge heterogeneity primarily through a judge-specific discrimination or sensitivity parameter, leading to a sensitivity-only score matrix of the form \(S=\gamma\mu^\top\). HJA contains this model as the special case \(r=0\), up to the normalization \(1_K^\top\gamma=K\). The main distinction is that, for \(r>0\), HJA allows judges to deviate from the consensus direction through the residual component \(UV^\top\). Thus \(\gamma_k\) should be interpreted as judge \(k\)'s sensitivity to the consensus ranking direction, rather than as a global sensitivity parameter scaling all aspects of the judge-specific score vector. A judge can be weakly aligned with the consensus while still having strong, structured preferences in the residual disagreement subspace. This differs from formulations in which a single positive discrimination parameter scales the entire latent score vector.
If one instead wishes $\gamma$ to scale the full judge-specific score, one may set
\(
U=\diag(\gamma)\tilde U
\)
and write
\[
S
=
\gamma\mu^\top+\diag(\gamma)\tilde U V^\top
=
\diag(\gamma)\bigl(\one_K\mu^\top+\tilde U V^\top\bigr).
\]
Under this alternative parameterization, $\gamma$ plays the role of a global sensitivity parameter, since it scales both the consensus and heterogeneous components simultaneously.

Additionally, our formulation does not impose strict positivity on the coefficients $\gamma_k$. This is useful in settings with noisy or partially adversarial judges: $\gamma_k=0$ means that judge $k$ does not respond to the consensus direction and contributes only through structured disagreement around consensus, whereas $\gamma_k<0$ indicates systematic opposition to the consensus signal. By contrast, JA-ranking imposes a strict positivity restriction on the corresponding discrimination parameters.

Thirdly, JA-Ranking assumes a balanced comparison design in which each unordered model pair is sampled uniformly for each judge. HJA instead works directly with the observed triple set \(\Omega\) in \eqref{eq:Omega}, allowing arbitrary and unbalanced cell counts \(n_{kij}\), including designs in which some judge--pair cells are observed much more often than others or are missing entirely. This flexibility is important in realistic LLM evaluation pipelines, where certain comparisons may be difficult to obtain because of API rate limits, model availability, annotation cost, safety filters, or adaptive sampling policies that focus effort on near-tie or high-disagreement pairs. In the main theory, identification and inference are guaranteed when each judge-specific active comparison graph is connected. Beyond this regime, \Cref{subsec:partial-sample} gives a pooled-graph treatment: disconnected judge-specific graphs can still be handled when the pooled comparison graph is connected and the structural low-rank constraints rule out unidentified local perturbations.

\noindent\textbf{Relation to low-rank BTL and collaborative ranking models.}
Low-rank preference and collaborative ranking models \citep{rajkumar2016can,jin2020rank,wu2015clustering} also exploit low-dimensional structure in heterogeneous preferences. However, HJA uses low rank for a different inferential purpose. Rather than treating the factorization of the full score matrix as the primary object, HJA first defines the consensus direction as the row-average score vector and then constrains the residual factor \(UV^\top\) to be orthogonal to that consensus direction. Consequently, the low-rank component represents structured disagreement around consensus, not an arbitrary rotated latent preference space. The normalization and rank conditions in Proposition~\ref{prop:ident} make this separation identifiable from the BTL comparison law, up to the usual column-wise sign ambiguity of the factors.

\noindent\textbf{Relation to two-step BTL plus spectral decomposition.}
A natural alternative is to first estimate an unrestricted judge-specific BTL score matrix \(\hat S\), and then recover \((\hat\gamma,\hat\mu,\hat U,\hat V)\) by applying the decomposition in Proposition~\ref{prop:ident} followed by a truncated singular value decomposition (SVD). This two-step procedure is useful as a diagnostic and as a computational baseline, especially when each judge has many repeated comparisons over a dense comparison graph. However, it separates score estimation from the structural constraints and therefore does not exploit the low-rank consensus-plus-heterogeneity model during likelihood fitting. In finite samples, this can lead to high variance in the estimated score matrix before the spectral step is applied, as shown in \Cref{fig:synthetic}, when the heterogeneity level increases. 
By contrast, the constrained MLE estimates the identifiable components jointly, borrows strength across judges and items through the structural model, and directly supports the Fisher-information-based uncertainty quantification.

\noindent\textbf{Relation to random graph theory.}
The inference properties of HJA rely on the connectivity structure of the active comparison graphs $\mathcal{G}_k^w$ (where edge $(i,j)$ exists if $(k,i,j)\in\Omega$) for each judge $k$. The weighted graph-Laplacian framework in \Cref{prop:fisher_graph_nondegenerate} shows that Fisher information nondegeneracy is controlled by properties of the graph, similar to spectral properties in random graph theory \citep{chung1997spectral}. With these techniques, we can establish the properties of MLE without requiring a fully connected comparison graph, which is important for real-world evaluation settings, as judges may only compare subsets of items.

In the asymptotic regime where the comparison graph $\mathcal{G}$ (where edge $(i,j)$ exists if $(k,i,j)\in\Omega$ for some $k$) is fixed, connectivity is imposed as a deterministic structural assumption (\Cref{asm:graph}). However, when comparisons are allocated randomly, or the graph grows dynamically (e.g., active learning over time), one could appeal to results from random graph theory, such as properties of Erdős-Rényi \citep{erdos1959graph}, to establish that connectivity holds with high probability. 
Such connections suggest that HJA's asymptotic theory could be extended to settings with random or growing comparison graphs, though the current paper focuses on fixed-panel designs common in LLM evaluation.

\clearpage
\section{Identification}\label{app:ident}
\subsection[Proof of Proposition 1]{Proof of \Cref{prop:ident}}\label{app:ident:prop:ident}

\begin{proof}[Proof of \Cref{prop:ident}]
    The proof consists of three parts. First for (i), given any set of pairwise probabilities, $S$ is uniquely identified if each of its row is restricted to have zero sum (i.e., Eq. \eqref{eq:S-id}). The second part proves that the mapping induced by \eqref{eq:score-decomp} from $\Theta_r$ to the space of $S$ is injective. Finally, Part (iii) imposes \eqref{eq:S-id} on the space of $S$, which makes such mapping bijective (one-to-one). Part (iii), combined with Part (ii), establishes the unique existence, which is identification.
    \paragraph{Part (i): Identification of BTL score matrix.}
    For each fixed judge $k$, the pairwise probabilities determine all score differences $S_{ki}-S_{kj}$ through the inverse logit \eqref{eq:binomial-model}. Hence, if $\tilde{S}$ is another matrix generating the same pairwise probabilities, then $$\tilde{S}_{ki}-\tilde{S}_{kj}=S_{ki}-S_{kj} \qquad \text{for all } i,j.$$
    Therefore, for each $k$, the vector $\tilde{S}_{k,:}-S_{k,:}$ is constant across items, i.e., there exists $c_k\in\RR$ such that
    $$\tilde{S}_{k,:}=S_{k,:}+c_k\one_N^\top.$$
    Right-multiplying both sides by $\one_N$ yields 
    $$\tilde{S}_{k,:}\one_N=S_{k,:}\one_N+c_k\one_N^\top\one_N=S_{k,:}\one_N+Nc_k.$$
    But it is required that $S\one_N=\zero$, with the same applying to $\tilde{S}$, and therefore $c_k=0$ for all $k\in\{1,\ldots,K\}$, i.e., $S=\tilde{S}$.

    \paragraph{Part (ii): Identification of latent components.}
    Below, we prove that if $S=\gamma\mu^\top+UV^\top$ with $(\gamma,\mu,U,V)\in\Theta_r$, the parameters are uniquely identified subject to some sign changes to be specified later.

    Multiplying both sides of \eqref{eq:score-decomp} with $\frac{1}{K}\one_K^\top$ yields 
    \begin{equation} \label{eq:mu}
    \frac{1}{K}\one_K^{\top}S = \frac{1}{K}\one_K^{\top}\gamma\mu^{\top} + \frac{1}{K}\one_K^{\top}UV^{\top} = \mu^{\top},
    \end{equation}
    so $\mu$ is immediately identified by the row average of $S$.
    
    To further separate the consensus item direction from the heterogeneity item subspace, by \Cref{cond:norm}\ref{cond:orth},
    \begin{equation} \label{eq:gamma}
        S\mu = \gamma\|\mu\|^2 + U(V^{\top}\mu) = \gamma\|\mu\|^2 \Longrightarrow \gamma = \frac{S\mu}{\|\mu\|^2},
    \end{equation}
    and thus $\gamma$ is identified.

    We proceed to identify $U$ and $V$ from $\tilde{S}:=S-\gamma\mu^\top=UV^\top$. Write $U=[u_1,\dots,u_r]$ and $V=[v_1,\dots,v_r]$. By \Cref{cond:norm}\ref{cond:subspace}, the vectors $v_1,\dots,v_r$ are orthonormal after normalization by $\sqrt{N}$, and the columns of $U$ are mutually orthogonal with $\|u_j\|_2^2=Kd_j$. Together with the orthogonality condition (\Cref{cond:norm}\ref{cond:orth}), this yields the orthogonal decomposition
    $$\tilde{S}=\sum_{j=1}^r u_jv_j^\top,$$
    where the left vectors $\gamma,u_1,\dots,u_r$ are mutually orthogonal and the right vectors $\mu,v_1,\dots,v_r$ are mutually orthogonal.
    Therefore, after normalization, this is a SVD of $\tilde S$. Its nonzero singular values are
    $$
    \sigma_j=\|u_j\|_2\|v_j\|_2=\sqrt{KN\,d_j}, \qquad j=1,\dots,r.
    $$
    The strict inequalities $d_1>\cdots>d_r>0$ imply that $\sigma_1,\dots,\sigma_r$ are pairwise distinct and therefore simple.

    Now let $\tilde{S}=\tilde{U}\tilde{V}^\top$ be another representation satisfying the same conditions. The same argument shows that this representation also produces an SVD of $\tilde{S}$ with the same simple nonzero singular values. By the uniqueness of singular vectors associated with simple singular values, there exist signs $\varepsilon_1,\dots,\varepsilon_r\in\{\pm1\}$ such that
    $$ 
        \tilde{u}_j=\varepsilon_j u_j,\qquad
        \tilde{v}_j=\varepsilon_j v_j,\quad j=1,\dots,r.
    $$
    For each $j$, the pair $(\tilde{u}_j,\tilde{v}_j)$ may be replaced by $(-\tilde{u}_j,-\tilde{v}_j)$ without changing the product $\tilde{u}_j\tilde{v}_j^\top$. And therefore $U$ and $V$ are identified up to the sign changes in their corresponding columns.
    
    \begin{remark}[Sign of factors and loadings] \label{rem:sign-id}
        In subsequent analysis where such sign ambiguities also need ruling out, we may require the first non-zero elements of each column of $U$ to be positive, so that the signs of $U$ and $V$ are anchored and fixed.
    \end{remark}
    
    \paragraph{Part (iii): Existence of parameters recovered from $S$.} 
    We are left to show that given any matrix $S\in\RR^{K\times N}$ satisfying $S\one_N=\zero$ and \Cref{cond:S-rank}, these parameters can indeed be constructed. Let $\mu=\frac{1}{K}S^\top\one_K$, and we have $\one_N^\top\mu=0$.
    Then, $$\gamma:=\frac{1}{\|\mu\|^2}S\mu=\frac{K}{\|\one_K^\top S\|^2}SS^\top\one_K$$ satisfies $\one_K^\top \gamma = K$, which is the second part of \Cref{cond:norm}\ref{cond:scale}.

    After recovering the consensus part, the remaining $\tilde{S}=S-\alpha SS^\top\one_K \one_K^\top S$, where $\alpha=1/\|\one_K^TS\|^2$. Hence, \(\tilde{S}=SP_{S^\top \one_K}^\perp\) where $P_A^\perp$ is the projection matrix onto the orthogonal complement of $\text{col}(A)$.
    Therefore, $\one_K^{\top}\tilde{S} = \one_K^\top SP_{S^\top \one_K}^\perp=\zero^\top$ automatically holds. 
    Take the SVD of $\tilde{S}$ to get $\tilde{S}=P\Sigma R^\top$. Let $V=\sqrt{N}R$ and $U= \frac{1}{\sqrt{N}} P\Sigma$, and \Cref{cond:norm}\ref{cond:subspace} is satisfied. The rank $r$ is also determined to be $\rank(\tilde{S})$.

    To verify the first part of \Cref{cond:norm}\ref{cond:scale}, we left-multiply $\one_K^\top$ to the equation $\tilde{S}=UV^\top$ and get $(\one_K^\top U)V^\top=\one_K^{\top}\tilde{S}=\zero^\top$, i.e., $V (\one_K^\top U)^\top=\zero$. But $V\in\RR^{N\times r}$ is constructed to have full column rank, and consequently \Cref{cond:norm}\ref{cond:scale} must hold.

    Finally, we study $V^\top\mu$. Note that by \eqref{eq:S-residual}, $\tilde{S}\mu=\frac{1}{K}SP_{S^\top \one_K}^\perp\one_K^\top S=\zero$, and hence $UV^\top\mu=\zero$. But $U$ is also constructed to have full column rank, and $V^\top\mu$ must be zero. Taking the transpose yields \Cref{cond:norm}\ref{cond:orth}.
    A similar argument verifies $V^\top\one_N=\zero$.

    Parts (ii) and (iii) jointly corroborate \Cref{prop:ident}(ii).
\end{proof}

\clearpage
\section{Estimation}\label{app:est}

\subsection{Rank selection}
\label{app:subsec:rank-selection}
\subsubsection{Rank selection methods}
The rank $r$ of the factorization $UV^\top$ serves as a structural hyperparameter that controls the model's capacity to capture structured disagreement around consensus, thereby affecting model complexity. 
We offer three ways to select the rank $r$ of HJA. Note that all methods below allow the case of the $r=0$, i.e., when our model is reduced to the sensitivity-only (JA) model.

\paragraph{Cross Validation.}
In experiments with dense comparison graphs, we can mask a validation subset of the pairwise comparisons $\Omega_{\text{val}} \subset \Omega$. For each candidate rank $r$, we estimate the parameters $\hat{\theta}_r$ on the training set and select the $r$ that minimizes the negative log-likelihood on $\Omega_{\text{val}}$. This guards against overfitting by measuring generalizability and is hence suitable for empirical practices with real-world datasets. We adopt this method when dealing with real-world datasets in this paper; see \Cref{app:subsec:real-evaluation}.

\paragraph{Spectral Residual Analysis.}
According to \eqref{eq:S-residual}, estimation of the preference residuals $UV^\top$ as a whole does not depend on $r$. Therefore, as a straightforward diagnostic, $r$ can be selected by analyzing the spectral decay of $UV^\top$, which can be estimated by fitting the consensus-only model ($S_0 = \gamma \mu^\top$) and computing $S-S_0$. A scree plot of the SVD of this residual matrix can help locate the optimal $r$, which typically corresponds to the elbow of the plot, reflecting the dominant aspects of disagreement. However, since the optimization is structural, this method may serve as a heuristic for future work. Another potential problem with this approach lies in the instability of $\hat S$ directly estimated from the data, which is shown in \Cref{fig:synthetic}.

\paragraph{Information-theoretic approach: BIC.}

We also propose the BIC method for hyperparameter selection, which is statistically guaranteed to consistently estimate the true rank $r^*$ of the underlying judge preference matrix $(UV^\top)^*$.
Note that the BIC for a candidate rank $r$ is defined as
\begin{equation} \label{eq:BIC}
    \text{BIC}(r) = 2 \mathcal{L}_n(\hat{\theta}_r) + d_r \log n,
\end{equation}
where $\mathcal{L}_n(\hat{\theta}_r)$ is the minimized negative log-likelihood under rank $r$, and $d_r := r(K + N - r-3)$ is the effective degrees of freedom (having accounted for the identifiability constraints).
The following theorem shows the consistency of BIC under the setting of \Cref{thm:asymptotics}.

\begin{theorem}[Consistency of rank selection]
\label{thm:bic_consistency}
Under \Cref{asm:sampling,asm:graph} with fixed comparison graph and $\mu^*\ne0$, the rank selected by minimizing BIC defined in \eqref{eq:BIC} is statistically consistent, i.e.,
\begin{equation*}
    \lim_{n \to \infty} \mathbb{P} \left\{ \arg\min_{r \in \{0, \dots, R_{\max}\}} \mathrm{BIC}(r) = r^* \right\} = 1.
\end{equation*}
\end{theorem}
We now provide the formal proof for Theorem \ref{thm:bic_consistency}. The goal is to show that minimizing the Bayesian Information Criterion (BIC) consistently selects the true heterogeneity rank $r^*$ as the total number of comparisons $n \to \infty$ while the comparison graph $\Omega$ remains fixed. Note that the sample log-likelihood is $-\cL_n(\theta)$ throughout the paper.
\begin{proof}[Proof of \Cref{thm:bic_consistency}]
Let $\mathbb{P}^*$ denote the true data-generating distribution parameterized by $\theta^* \in \Theta_{r^*}$, where $\Theta_r$ is the constrained parameter space corresponding to rank $r$ (adhering to the identifiability constraints in \Cref{prop:ident}). 
To prove $\lim_{n \to \infty} \mathbb{P}\left(r^*=\arg\min_r \text{BIC}(r)\right) = 1$, we must show that for any $r \neq r^*$, $\mathbb{P}\left(\text{BIC}(r) - \text{BIC}(r^*) > 0\right) \to 1$. We partition the analysis into two cases: the under-parameterized regime ($r < r^*$) and the over-parameterized regime ($r > r^*$).

\paragraph{Case 1: Under-parameterized regime ($r < r^*$).}
When $r < r^*$, the candidate model space $\Theta_r$ does not contain the true parameter $\theta^*$. Because the BTL model is a strictly concave Generalized Linear Model (GLM) mapping bounded logits to probabilities, and our identifiability constraints make the parameterization injective, $\theta^*$ cannot be perfectly approximated by any $\theta \in \Theta_r$. 

Following the argument of \Cref{app:consistency} (especially \eqref{eq:kld-score}), by the properties of the KL divergence, the maximum expected log-likelihood under the restricted space $\Theta_r$ is strictly less than under the true parameter:
\begin{equation*}
    \delta_r := M(\theta^*) - \sup_{\theta \in \Theta_r} M(\theta) > 0
\end{equation*}
By the Uniform Law of Large Numbers (since $\Theta$ is compact and the BTL log-likelihood is continuous and bounded), we have $\frac{1}{n}\mathcal{L}_n(\hat{\theta}_r) \pto \sup_{\theta \in \Theta_r} M(\theta)$ and $\frac{1}{n}\mathcal{L}_n(\hat{\theta}_{r^*}) \pto M(\theta^*)$. 
Hence \(\mathcal L_n(\hat\theta_r)-\mathcal L_n(\hat\theta_{r^*})=n\delta_r+\op(n)\).

The difference in BIC is
\begin{align}
    \text{BIC}(r) - \text{BIC}(r^*) &= -2\left(\mathcal{L}_n(\hat{\theta}_{r^*}) - \mathcal{L}_n(\hat{\theta}_r)\right) + (d_r - d_{r^*})\log(n) \label{eq:bic-diff}
    \\
    &= -2n \left( \frac{1}{n}\mathcal{L}_n(\hat{\theta}_{r^*}) - \frac{1}{n}\mathcal{L}_n(\hat{\theta}_r) \right) + (d_r - d_{r^*})\log(n). \notag
\end{align}
Since the term inside the first parenthesis converges in probability to $\delta_r > 0$, the likelihood difference scales as $\mathcal{O}_p(n)$. But the penalty difference $(d_r - d_{r^*})\log(n)$ only scales as $\mathcal{O}(\log n)$, making the linear growth of the likelihood deficit dominate. Thus, for any $r < r^*$,
\begin{equation*}
    \lim_{n \to \infty} \mathbb{P}\left( \text{BIC}(r) - \text{BIC}(r^*) > 0 \right) = 1.
\end{equation*}

\paragraph{Case 2: Over-parameterized regime ($r > r^*$)}
When $r > r^*$, the candidate model space $\Theta_r$ strictly nests the true model space $\Theta_{r^*}$ in terms of closure (i.e., $\overline{\Theta_{r^*}} \subset \overline{\Theta_r}$). Therefore, $\theta^* \in \overline{\Theta_r}$, and the model can be viewed as correctly specified but over-parameterized.

Because the true model is nested within the larger model, we invoke theory on the Likelihood Ratio (LR) test. Under the identifiability constraints (\Cref{subsec:id}) and assumptions required by asymptotic normality, the Fisher Information Matrix evaluated at $\theta^*$ is non-singular. 
Note that under the null hypothesis that the true parameter lies in the restricted space $\Theta_{r^*}$, the log-likelihood ratio statistic may not converge in law to a chi-squared distribution, as the true parameter $\theta^*$ essentially lies on the boundary of $\Theta_r$. However,  $-2\left(\mathcal{L}_n(\hat{\theta}_r) - \mathcal{L}_n(\hat{\theta}_{r^*})\right)$ is still $\cO_p(1)$ by theory on the boundary behaviour of LR [\citealp[Theorem 3]{liang1987asymptotic}] and tests on reduced-rank regression [\citealp[Theorem 3.2]{Robin_Smith_2000}].
Therefore, the maximum improvement in log-likelihood gained by fitting $r - r^*$ spurious components is bounded in probability. But $r > r^*$ leads to $d_r > d_{r^*}$ and consequently the penalty term $(d_r - d_{r^*})\log(n) \to \infty$ as $n \to \infty$. 
Thus, for any $r > r^*$:
\begin{equation*}
    \lim_{n \to \infty} \mathbb{P}\left( \text{BIC}(r) - \text{BIC}(r^*) > 0 \right) = 1.
\end{equation*}

Combining Case 1 and Case 2, the true rank $r^*$ strictly minimizes the BIC with probability tending to 1 as the sample size $n \to \infty$, i.e., $\lim_{n \to \infty} \mathbb{P}\left(\arg\min_r \text{BIC}(r)=r^*\right) = 1$.
\end{proof}

Note that the theory above also applies to the case of $r^*=0$. We also add a proof-of-concept simulation on the rank selection consistency of BIC under the oracle model; see \Cref{fig:bic-rank-consistency}.

\begin{figure}
    \centering
    \includegraphics[width=0.5\linewidth]{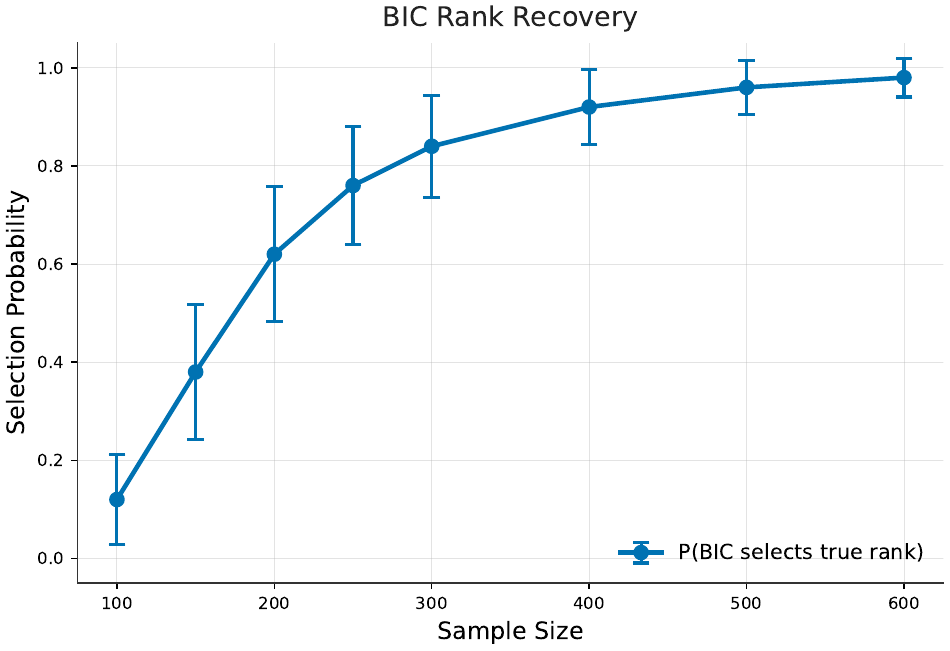}
    \caption{BIC rank selection consistency. The selection probability (y-axis) is estimated by the frequency of correct rank estimation over 50 repetitions. Error bars denote $\pm$ twice the standard error.}
    \label{fig:bic-rank-consistency}
\end{figure}

\subsubsection{Practical suggestions and reduction to simpler models} \label{app:zero-rank}
Based on the theory and simulations, we prefer BIC when the true data fit our method well or when computational resources are too limited to run cross-validation.

It is possible that the selected rank is zero, indicating that a simpler model might suffice and be preferred for sake of parsimony. We provide a targeted model-reduction experiment with 3000 balanced comparisons and five-fold validation over $r\in\{0,1,2,3\}$; see \Cref{tab:reducetion-to-sensitivity}.
\begin{table}[h] 
\centering
\caption{Testing reduction to sensitivity-only model ($r^*=0$)}
\begin{tabular}{l c c}
\toprule
\textbf{Data-generating regime} & \textbf{$r=0$ selected} & \textbf{Interpretation} \\ \midrule
Exact sensitivity-only model & 50/50 & No residual heterogeneity \\ 
Very small residual component & 48/50 & Residual signal is negligible \\ 
10\% independent label flips & 48/50 & Unstructured noise only \\ 
Structured rank-one alternative & 0/50 & Reproducible residual heterogeneity \\
\bottomrule
\end{tabular}
\label{tab:reducetion-to-sensitivity}
\end{table}

Thus, cross-validation does not introduce a disagreement component merely to absorb random variation: it selects $r=0$ when the data contain no meaningful residual structure, and $r>0$ in every structured run. A selected $r=0$ implies the data do not support a reproducible $UV^\top$ component, though not that population heterogeneity is exactly zero. After selecting $r=0$, we evaluate the further reduction to pooled BTL by checking whether a joint 95\% confidence region for the $K-1$ contrasts of $\gamma-\one_K$ contains zero; see \Cref{tab:reduction-to-BTL}.

\begin{table}[h]
\centering
\caption{Further reduction to pooled-BTL}
\begin{tabular}{l c c}
\toprule
\textbf{Data-generating regime} & \textbf{Zero contained} & \textbf{Suggested reduction} \\ \midrule
Exact pooled model & 49/50 & Pooled BTL \\
Slight sensitivity variation & 46/50 & Often compatible with pooled BTL \\
Independent symmetric label flips & 45/50 & Often compatible with pooled BTL \\
Material sensitivity heterogeneity & 0/50 & Retain judge-specific $\gamma$ \\ \bottomrule
\end{tabular}
\label{tab:reduction-to-BTL}
\end{table}

This suggests a practical hierarchy: select $r$ by validation likelihood or BIC; if $r=0$, assess whether $\gamma=\one_K$ is compatible; and report held-out negative log-likelihood alongside either simplification.

\subsection{Algorithm} \label{app:est-main}


\subsubsection{Initialization}
We initialize our parameters by fitting a pooled BTL model for $\mu$, all ones for $\gamma$, and a judgewise BTL model for each row of $S$. See \Cref{alg:initialize}.
\begin{algorithm}[!ht]
    \caption{\textsc{Initialize}: construct $(\gamma^{(0)},\mu^{(0)},U^{(0)},V^{(0)})$}
    \label{alg:initialize}
    \begin{algorithmic}[1]
    \Require Aggregated counts $\{(n_{kij},Y_{kij})\}_{(k,i,j)\in\Omega}$, target rank $r$
    \Ensure Initial value $\theta^{(0)}=(\gamma^{(0)},\mu^{(0)},U^{(0)},V^{(0)}) \in \Theta_r$
    
    \State Fit a pooled centered BTL model to obtain $\mu^{(0)}$.
    \State Re-center $\mu^{(0)} \gets \mu^{(0)} - \frac{1}{N}(\one_N^\top \mu^{(0)})\one_N$.
    \State Set $\gamma^{(0)} \gets \one_K$.
    
    \For{$k=1,\ldots,K$}
        \State Fit a centered judgewise BTL model using judge $k$'s comparisons to obtain $\tilde S^{(0)}_{k,:}$.
    \EndFor
    
    \State Form the initial residual matrix $R^{(0)} \gets \tilde S^{(0)} - \gamma^{(0)}(\mu^{(0)})^\top$.
    \State Judge-center the residual:
    \(
    R_c^{(0)} \gets R^{(0)} - \frac{1}{K}\one_K \one_K^\top R^{(0)}.
    \)
    
    \State Compute the rank-$r$ truncated SVD $R_c^{(0)} \approx P_r \Sigma_r Q_r^\top$.
    \State Set $U^{(0)} \gets P_r \Sigma_r^{1/2}$ and $V^{(0)} \gets Q_r \Sigma_r^{1/2}$.
    
    \State Compute $(\gamma^{(0)},\mu^{(0)},U^{(0)},V^{(0)}) \gets \textsc{ReAnchor}( \gamma^{(0)},\mu^{(0)},U^{(0)},V^{(0)} )$ based on \Cref{alg:reanchor}.

    
    \State \Return $\theta^{(0)}=(\gamma^{(0)},\mu^{(0)},U^{(0)},V^{(0)})$.
    \end{algorithmic}
\end{algorithm}


\subsubsection{Re-anchor and enforce identifiability}

\begin{algorithm}[!ht]
    \caption{\textsc{ReAnchor}: canonical re-anchoring with feasibility check}
    \label{alg:reanchor}
    \begin{algorithmic}[1]
    \Require Tentative $(\tilde{\gamma},\tilde{\mu},\tilde{U},\tilde{V})$ from the two block updates, tolerance $\delta_\mu>0$, spectral tolerance $\delta_\sigma>0$.
    \Ensure Either a feasible anchored point $(\gamma^+,\mu^+,U^+,V^+)\in\Theta_r$ or \textsc{Fail}.

    \State Set $\mu^+ \gets \tilde{\mu}$.
    \If{$\|\mu^+\|_2 < \delta_\mu$}
        \State \Return \textsc{Fail} \Comment{Avoid instability incurred on updating $V$}
    \EndIf

    \State Compute $a \gets \tilde{V}^\top \mu^+ / \|\mu^+\|_2^2$.
    \State Set $\bar V \gets \tilde V - \mu^+ a^\top$ and $\gamma^{+} \gets \tilde\gamma + \tilde U a$.

    \State Form $\bar H \gets \tilde U \bar V^\top$.
    \State Compute the thin SVD $\bar H = P\Sigma Q^\top$, with diagonal entries of $\Sigma$ arranged in decreasing order.
    \If{$\sigma_r(\Sigma) < \delta_\sigma$}
        \State \Return \textsc{Fail} \Comment{Avoid ill-conditioned matrices caused by a large leap}
    \EndIf

    \State Set $U^+ \gets P\Sigma/\sqrt{N}$, and $V^+ \gets \sqrt{N}\,Q$.
    Flip the corresponding columns of $(U^+,V^+)$ so that the first nonzero entry of each column of $U^+$ is positive.
    \State \Return $(\gamma^+,\mu^+,U^+,V^+)$
    \end{algorithmic}
\end{algorithm}
The re-anchoring step restores the nonlinear identification constraints without changing the implied score matrix $\tilde S = \tilde\gamma \tilde\mu^\top + \tilde U \tilde V^\top$. Concretely, \textsc{ReAnchor} removes the component of $\tilde V$ along $\tilde\mu$, compensates for this change through the update $\gamma^{\mathrm{+}} = \tilde\gamma + \tilde U a$, and then refactorizes the heterogeneity term by a normalized SVD. If the candidate loading violates the local nondegeneracy conditions, the routine returns \textsc{Fail}, and the main algorithm increases the proximal parameter before recomputing the block updates.

Unlike JA-Ranking \citep{xu2026judgeaware}, we do not log-reparameterize the judge-specific coefficient \(\gamma\). 
In our model, \(\gamma\) is an affine-anchored consensus-sensitivity vector recovered from the structured score matrix, not a strictly positive global discrimination parameter. 
A log transform would impose unnecessary positivity and would not be compatible with the \(S\)-preserving re-anchoring map. The proximal terms make the block updates well posed and provide descent, while re-anchoring restores the canonical representative without changing the likelihood.

In exact arithmetic, the post-anchor candidate automatically satisfies
$\one_K^\top \gamma^{\mathrm{+}} = K$ because the judge block enforces
$\one_K^\top \tilde\gamma = K$ and $\one_K^\top \tilde U = \zero_r^\top$.

\subsubsection{Confidence intervals for target functionals}

\begin{algorithm}[!ht]
\caption{\textsc{EstimateFisherInfo}: plug-in Fisher information}
\label{alg:covariance}
\begin{algorithmic}[1]
\Require Final estimator \(\widehat\theta=(\widehat\gamma,\widehat\mu,\widehat U,\widehat V)\), data \(\{(n_{kij},Y_{kij})\}_{(k,i,j)\in\Omega}\).
\Ensure Plug-in Fisher information matrix \(\mathcal I_n(\widehat\theta)\).
\State Use the local identified parameterization fixed in \Cref{subsec:estimation}; all gradients below are taken in this parameterization.
\For{each \((k,i,j)\in\Omega\)}
    \State Compute
    \(
        \widehat\eta_{kij}
        =
        \widehat\gamma_k(\widehat\mu_i-\widehat\mu_j)
        +
        \widehat U_{k,:}(\widehat V_{i,:}-\widehat V_{j,:})^\top .
    \)
    \State Set
    \(
        \widehat p_{kij}
        =
        \{1+\exp(-\widehat\eta_{kij})\}^{-1}.
    \)
    \State Compute the local gradient
    \(
        \widehat g_{kij}
        =
        \nabla_\theta \eta_{kij}(\widehat\theta).
    \)
\EndFor
\State Set
\(
    \mathcal I_n(\widehat\theta)
    \gets
    \sum_{(k,i,j)\in\Omega}
    \frac{n_{kij}}{n}
    \widehat p_{kij}(1-\widehat p_{kij})
    \widehat g_{kij}\widehat g_{kij}^{\top}.
\)
\State \Return \(\mathcal I_n(\widehat\theta)\).
\end{algorithmic}
\end{algorithm}

\begin{algorithm}[!ht]
\caption{\textsc{UncertaintyQuantification}: delta-method confidence intervals}
\label{alg:uq}
\begin{algorithmic}[1]
\Require Final estimator \(\widehat\theta\), plug-in information matrix \(\mathcal I_n(\widehat\theta)\), target functionals \(\mathcal Q\), confidence level \(1-\alpha\).
\Ensure Wald confidence intervals for \(\{q(\theta):q\in\mathcal Q\}\).
\State Use the local identified parameterization fixed in \Cref{subsec:estimation}.
\State Set
\(
    \widehat\Sigma_{\theta}
    \gets
    \frac{1}{n}\mathcal I_n(\widehat\theta)^{-1}.
\)
\For{each smooth scalar target \(q\in\mathcal Q\)}
    \State Compute \(\widehat q=q(\widehat\theta)\).
    \State Compute the local gradient
    \(
        \widehat a_q
        =
        \nabla_\theta q(\widehat\theta).
    \)
    \State Compute
    \(
        \widehat{\mathrm{se}}(\widehat q)
        =
        \left(
        \widehat a_q^\top
        \widehat\Sigma_{\theta}
        \widehat a_q
        \right)^{1/2}.
    \)
    \State Report
    \(
        \widehat q
        \pm
        z_{1-\alpha/2}\widehat{\mathrm{se}}(\widehat q).
    \)
\EndFor
\end{algorithmic}
\end{algorithm}

\subsection{Convergence analysis}\label{app:est:subsec:conv}

\begin{assumption}[Local algorithmic chart]
\label{asm:algo-chart}
There exists a compact set $\Theta_{r,0}\subset \Theta_r$ such that:
(i) $\theta^\ast$ lies in the relative interior of $\Theta_{r,0}$ after fixing the local anchor signs of the columns of $U$ and $V$;
(ii) for every $\theta=(\gamma,\mu,U,V)\in\Theta_{r,0}$, one has $\|\mu\|_2\ge \delta_\mu$, and the nonzero singular values of $UV^\top$ are separated from each other and from zero by at least $\delta_\sigma$;
(iii) the re-anchoring map is well-defined and maps $\Theta_{r,0}$ into itself.
\end{assumption}

\begin{assumption}[Local isolated minimizer]
\label{asm:algo-local-min}
For the empirical criterion $\cL_n$, there exists a unique local minimizer $\hat\theta_n\in\Theta_{r,0}$, after fixing the anchor signs, and $\nabla_\xi^2 \cL_n(\hat\theta_n)$ is positive definite in the local free-coordinate chart $\xi$.
\end{assumption}

\begin{assumption}[Initialization in basin]
\label{asm:algo-init}
The initialized iterate $\theta^{(0)}$ lies in the basin of attraction of $\hat\theta_n$ within $\Theta_{r,0}$.
\end{assumption}

\begin{proof}[Proof of \Cref{thm:main_convergence}]
    The proof of \Cref{thm:main_convergence} follows from \Cref{prop:algo-conv}.
\end{proof}

\begin{proposition}[Convergence of the proximal anchored alternating algorithm]
\label[proposition]{prop:algo-conv}
Consider the sequence $\{\theta^{(t)}\}_{t\ge 0}$ generated by \Cref{alg:mle}, where each block subproblem includes the proximal terms
\[
\frac{\tau_t}{2}\|\gamma-\gamma^{(t)}\|_2^2+\frac{\tau_t}{2}\|U-U^{(t)}\|_F^2
\quad\text{and}\quad
\frac{\tau_t}{2}\|\mu-\mu^{(t)}\|_2^2+\frac{\tau_t}{2}\|V-V^{(t)}\|_F^2,
\]
with proximal parameters uniformly bounded away from zero ($\tau_t \ge \underline{\tau} > 0$) if $\tau$ is allowed to change during different iterations.

\textbf{Part I: Global descent properties.} Suppose the structural assumptions (\Cref{asm:algo-chart}) hold. Then:
\begin{enumerate}
    \item the sequence $\{\cL_n(\theta^{(t)})\}_{t\ge 0}$ is nonincreasing and convergent;
    \item the iterates satisfy $\sum_{t=0}^\infty \|\theta^{(t+1)}-\theta^{(t)}\|^2 < \infty$.
\end{enumerate}

\textbf{Part II: Local sequence convergence.} Suppose further that \Cref{asm:algo-local-min,asm:algo-init} concerning the basin hold. Then:
\begin{enumerate}
    \setcounter{enumi}{2}
    \item the iterates remain in $\Theta_{r,0}$ and the full anchored sequence $\theta^{(t)}$ converges to $\hat\theta_n$;
    \item equivalently, before anchoring, the factor pair $(U^{(t)},V^{(t)})$ converges to the equivalence class of $(\hat U_n,\hat V_n)$ up to the column-sign convention, and after re-anchoring the convergence is ordinary parameter convergence.
\end{enumerate}

\textbf{Part III: Local linear rate.} Following Part II, the convergence is locally linear: there exist constants $C>0$ and $\rho\in(0,1)$ such that
\[
\|\theta^{(t)}-\hat\theta_n\| \le C\rho^t
\quad\text{and}\quad
\cL_n(\theta^{(t)})-\cL_n(\hat\theta_n) \le C\rho^t
\]
for all sufficiently large $t$.
\end{proposition}
In particular, whenever the sample criterion admits a unique local minimizer $\hat\theta_n$ in the identified chart $\Theta_{r,0}$, the output of \Cref{alg:mle} coincides with that local constrained MLE for all sufficiently accurate runs.
\begin{proof}[Proof of \Cref{prop:algo-conv}]
Let the parameter space be partitioned into two blocks: the judge block $\phi_1 = (\gamma, U)$ and the item block $\phi_2 = (\mu, V)$, such that the full parameter vector is $\theta = (\phi_1, \phi_2)$.

\paragraph{Part I: Global Descent Properties.} 
This part relies only on the structural conditions in \Cref{asm:algo-chart} as the properties of proximal alternating minimization.

\textit{Claim 1 (Nonincreasing and convergent objective):}
For the judge block $\phi_1^{(t+1)}$, the update satisfies:
\[
\phi_1^{(t+1)} = \argmin_{\phi_1} \left\{ \cL_n(\phi_1, \phi_2^{(t)}) + \frac{\tau_t}{2}\|\phi_1 - \phi_1^{(t)}\|^2 \right\}.
\]
Comparing the value at the minimum $\phi_1^{(t+1)}$ against the previous iterate $\phi_1^{(t)}$ yields the sufficient decrease inequality for the first block:
\[
\cL_n(\phi_1^{(t+1)}, \phi_2^{(t)}) \le \cL_n(\phi_1^{(t)}, \phi_2^{(t)}) - \frac{\tau_t}{2}\|\phi_1^{(t+1)} - \phi_1^{(t)}\|^2.
\]
Applying identical logic to the item block update $\phi_2^{(t+1)}$ yields
\[
\cL_n(\phi_1^{(t+1)}, \phi_2^{(t+1)}) \le \cL_n(\phi_1^{(t+1)}, \phi_2^{(t)}) - \frac{\tau_t}{2}\|\phi_2^{(t+1)} - \phi_2^{(t)}\|^2.
\]
Adding these two inequalities provides the global sufficient decrease condition for a full iteration:
\begin{equation}\label{eq:sufficient_decrease}
\cL_n(\theta^{(t+1)}) \le \cL_n(\theta^{(t)}) - \frac{\tau_t}{2}\|\theta^{(t+1)} - \theta^{(t)}\|^2\leq \cL_n(\theta^{(t)}).
\end{equation}
The re-anchoring step (\Cref{alg:reanchor} and by \Cref{asm:algo-chart}) acts purely on the rotational and scaling equivalence classes of the bilinear factorization. Because it is an invariant transformation with respect to $S$, the objective value $\cL_n$ remains unchanged before and after re-anchoring. Therefore, the sequence of objective values $\{\cL_n(\theta^{(t)})\}$ is monotonically nonincreasing. Since the negative log-likelihood is bounded from below (by zero), the sequence converges to some limit $L^\ast \ge 0$.

\textit{Claim 2 (Square-summability):}
Summing the sufficient decrease inequality \eqref{eq:sufficient_decrease} from iteration $t=0$ to $T$ yields a telescoping sum:
\[
\sum_{t=0}^T \frac{\tau_t}{2}\|\theta^{(t+1)} - \theta^{(t)}\|^2 \le \cL_n(\theta^{(0)}) - \cL_n(\theta^{(T+1)}) \le \cL_n(\theta^{(0)}).
\]
Using the strict lower bound $\tau_t \ge \underline{\tau} > 0$ and taking the limit as $T \to \infty$, we obtain:
\[
\sum_{t=0}^\infty \|\theta^{(t+1)} - \theta^{(t)}\|^2 \le \frac{2}{\underline{\tau}}\cL_n(\theta^{(0)})< \infty.
\]
This finite sum guarantees square-summability of the iterate differences, which also trivially implies that the step sizes converge to zero: $\|\theta^{(t+1)} - \theta^{(t)}\| \to 0$ as $t \to \infty$.
\paragraph{Part II: Local Sequence Convergence.}
This part invokes the statistical and initialization conditions in \Cref{asm:algo-local-min,asm:algo-init}.

\textit{Claim 3 (Convergence to the local minimizer $\hat{\theta}_n$):}
By \Cref{asm:algo-local-min}, $\hat{\theta}_n$ is an isolated local minimizer. Therefore, there exists a compact neighborhood $\mathcal{B} \subset \Theta_{r,0}$ around $\hat{\theta}_n$ such that $\hat{\theta}_n$ is the strictly unique stationary point within $\mathcal{B}$, and $\cL_n(\theta) > \cL_n(\hat{\theta}_n)$ for all $\theta \in \mathcal{B} \setminus \{\hat{\theta}_n\}$. 

By \Cref{asm:algo-init}, the initialization $\theta^{(0)}$ lies close enough to $\hat{\theta}_n$ such that the connected component of the sublevel set $\mathcal{S} := \{ \theta \in \mathcal{B} \mid \cL_n(\theta) \le \cL_n(\theta^{(0)}) \}$ is strictly contained in the interior of $\mathcal{B}$, given the positive definiteness of the Hessian matrix.
Because $\mathcal{B}$ is compact and $\cL_n$ is continuous, the sublevel set $\mathcal{S}$ is also compact. 

By the global descent property established in Part I, $\cL_n(\theta^{(t+1)}) \le \cL_n(\theta^{(t)})$ for all $t$. Thus, the sequence cannot enter the region where $\cL_n(\theta) > \cL_n(\theta^{(0)})$, meaning the entire sequence $\{\theta^{(t)}\}_{t=0}^\infty$ is strictly confined within the compact set $\mathcal{S}$. The re-anchoring step under \Cref{asm:algo-chart} preserves this containment.

Because $\mathcal{S}$ is compact, the sequence $\{\theta^{(t)}\}$ must possess at least one limit point $\theta^\ast \in \mathcal{S}$ by the Bolzano-Weierstrass theorem. To characterize any such limit point, we invoke the first-order optimality conditions of the exact proximal block updates. For the judge block,
\[
\nabla_{\phi_1} \cL_n(\phi_1^{(t+1)}, \phi_2^{(t)}) + \tau_t (\phi_1^{(t+1)} - \phi_1^{(t)}) = 0.
\]
Taking the norm yields $\| \nabla_{\phi_1} \cL_n(\phi_1^{(t+1)}, \phi_2^{(t)}) \| = \tau_t \| \phi_1^{(t+1)} - \phi_1^{(t)} \|$. 
By Part I and the boundedness of $\tau_t$, the right-hand side converges to zero, and thus the gradient mapping uniformly converges to zero. Because $\cL_n$ is continuously differentiable, any subsequence converging to $\theta^\ast$ satisfies $\nabla \cL_n(\theta^\ast) = 0$. 

Therefore, every limit point of the sequence is a stationary point. But $\mathcal{S} \subset \mathcal{B}$, and by construction, $\hat{\theta}_n$ is the strictly unique stationary point within $\mathcal{B}$. Therefore, $\theta^\ast = \hat{\theta}_n$ is the unique limit point of the sequence. Since $\{\theta^{(t)}\}$ is a sequence in a compact set with exactly one limit point, the entire sequence globally converges to $\hat{\theta}_n$.

\textit{Claim 4 (Pre- and post-anchoring equivalence):}
Because the joint parameters converge, the sequence of latent evaluation scores $S^{(t)} = \gamma^{(t)}(\mu^{(t)})^\top + U^{(t)}(V^{(t)})^\top$ converges to the optimal empirical score matrix $\hat{S}_n$. Before the re-anchoring step is applied, the unanchored factor matrices $(U^{(t)}, V^{(t)})$ converge to the equivalence class $[\hat{U}_n, \hat{V}_n]$ defined by orthogonal rotations and scale-sign permutations. The re-anchoring map is a continuous function acting on $\Theta_{r,0}$ that uniquely identifies a representative from this equivalence class by fixing local anchor signs, and hence guarantees the ordinary parameter-wise convergence in Euclidean space.

\paragraph{Part III: Local linear rate.} This part draws the final conclusion on local convergence rate from the assumed and established properties of the objective function, the algorithm, and the topological structures of the global optimum.

We restrict the domain of the objective function in $\mathcal{S}$, where $\cL$ demonstrates strong convexity. By [\citealp[Theorem 2]{karimi2016linear}], $\cL$ satisfies the Polyak-\L ojasiewicz (PL) inequality \citep{polyak1963gradient}, and by Appendix F of \cite{karimi2016linear}, satisfies the proximal-PL property, and by Appendix G of \cite{karimi2016linear}, has a local Kurdyka-\L ojasiewicz (KL) exponent $\theta=1/2$. The local linear convergence rate hence directly follows from [\citealp[Theorem 3.4]{attouch2010proximal}].
\end{proof}

\clearpage
\section{Statistical Inference}\label{app:stat}
\subsection{Consistency of constrained MLE} \label{app:consistency}
Recall that the structural parameter is estimated by the constrained maximum likelihood estimator
\[
\hat{\theta}:=(\hat{\gamma},\hat{\mu},\hat{U},\hat{V})
\in
\arg\min_{\theta\in\Theta_{r}}\mathcal{L}(S(\theta)),
\qquad
S(\theta):=\gamma\mu^\top+UV^\top.
\]
Since the likelihood depends on $\theta$ only through the induced score matrix $S(\theta)$, the proof is most naturally carried out on the centered score space
\[
\mathcal{S}:=\{S\in\RR^{K\times N}: S\one_N=\zero, \text{ and \Cref{cond:S-rank} holds}\}.
\]
The signs of the columns of $(U,V)$ can be anchored and fixed via, say, \Cref{rem:sign-id}, and then \Cref{prop:ident} ensures $\cS=\{S(\theta):\theta\in\Theta_r\}$.

We formalize the regularity conditions mentioned in \Cref{subsec:UQ} as the assumptions below.
\begin{assumption}[Asymptotic sampling fractions]
\label{asm:sampling}
There exist fixed constants $w_{kij} \geq 0$ such that
\({n_{kij}}/{n} \pto w_{kij}\) for all \((k,i,j)\in\Omega\) as \(n\to\infty\).
\end{assumption}
If $n_{kij}$ is non-random, i.e., the design of the experiment is deterministic, the convergence in probability is reduced to convergence in the ordinary sense.
Note that a triple can be observed ($n_{kij}>0$) but still have asymptotic fraction $w_{kij}=0$. To account for this in subsequent analysis, we define:
\begin{equation}\label{eq:active-graph}
    \Omega_w:=\{(k,i,j)\in\Omega:w_{kij}>0\}
\end{equation}
as the set of asymptotically active triples.
\begin{assumption}[Graph connectivity]
\label{asm:graph}
For each judge $k \in [K]$, the comparison graph $\mathcal{G}_k^w = ([N], \mathcal{E}_k^w)$ with edge set $\mathcal{E}_k^w = \{\{i,j\} : (k,i,j) \in \Omega_w\}$ is connected, and $0<p_{kij}^*<1$ for each active triple \((k,i,j)\in\Omega_w\).
\end{assumption}
By connectedness, we mean an undirected graph where a path (a sequence of connecting edges) exists between every pair of nodes (items). Therefore, under \Cref{asm:graph}, all the items assessed by a judge should be linked via one or more comparisons. However, it is not required that all possible pairs be directly compared, and this scenario is accounted for via $\Omega_w$.

These assumptions can guarantee the consistency of $\hat\theta$, which is shown below.
\begin{proof}[Proof of \Cref{thm:asymptotics}: consistency part]
    For $S\in\mathcal S$, write
\[
\eta_{kij}(S):=S_{ki}-S_{kj},
\qquad
p_{kij}(S):=\frac{1}{1+\exp\{-\eta_{kij}(S)\}}.
\]
Also write $\bar Y_{kij}:=Y_{kij}/n_{kij}$. Then the normalized negative log-likelihood, viewed as a function of $S$, is
\begin{equation}\label{eq:score-nll}
M_n(S):=\frac1n\mathcal L(S)
=
\sum_{(k,i,j)\in\Omega}\frac{n_{kij}}{n}
\Big[
-\bar Y_{kij}\eta_{kij}(S)+\log\bigl(1+\exp\{\eta_{kij}(S)\}\bigr)
\Big].
\end{equation}
By the strong law, for each fixed $(k,i,j)\in\Omega$,
\[
\bar Y_{kij}\xrightarrow{\mathrm{a.s.}} p_{kij}^\ast,
\]
and by \Cref{asm:sampling},
\[
\frac{n_{kij}}{n}\to w_{kij}.
\]
Since $\Omega$ is fixed and each summand in \eqref{eq:score-nll} is continuous in $S$, it follows that $M_n(S)$ converges uniformly on every compact subset of $\mathcal S$ to
\begin{equation}\label{eq:population-score-criterion}
M(S)
:=
\sum_{(k,i,j)\in\Omega}
w_{kij}
\Big[
-p_{kij}^\ast\eta_{kij}(S)+\log\bigl(1+\exp\{\eta_{kij}(S)\}\bigr)
\Big].
\end{equation}

Let $S^\ast=S(\theta^\ast)\in\cS$. Then, by the usual Bernoulli Kullback--Leibler identity,
\begin{equation}\label{eq:kld-score}
M(S)-M(S^\ast)
=
\sum_{(k,i,j)\in\Omega}
w_{kij}\,
\mathrm{KL}\!\left(
\mathrm{Bern}(p_{kij}^\ast)\,\|\,\mathrm{Bern}(p_{kij}(S))
\right)\ge 0.
\end{equation}
Equality holds if and only if $p_{kij}(S)=p_{kij}^\ast$ for all $(k,i,j)\in\Omega_w$, equivalently,
\[
S_{ki}-S_{kj}=S^\ast_{ki}-S^\ast_{kj}
\qquad\text{for all }(k,i,j)\in\Omega_w.
\]

By \Cref{lem:active-ident}, equality in \eqref{eq:kld-score} implies $S=S^\ast$. Hence $S^\ast$ is the unique minimizer of $M$ on $\mathcal S$.


We first prove the consistency of the score-space MLE $\hat S\in\arg\min_{S\in\cS}\mathcal L(S)$.

Fix any open neighborhood $\mathcal O$ of $S^\ast$ in $\mathcal S$. Since $M$ is continuous and uniquely minimized at $S^\ast$, there exists $\delta>0$ such that
\begin{equation}\label{eq:score-gap}
\inf_{S\in\cS\setminus\mathcal O} M(S)\ge M(S^\ast)+3\delta.
\end{equation}
By \Cref{lem:score-coercive}, there exists $R<\infty$ such that
\begin{equation}\label{eq:score-ball-gap}
\inf_{S\in\cS:\,\|S\|_F>R} M(S)\ge M(S^\ast)+4\delta.
\end{equation}
Define the truncated score set
\[
K_R:=\{S\in\cS:\|S\|_F\le R\}.
\]
Since $K_R$ is closed and bounded in the finite-dimensional Euclidean space $\RR^{K\times N}$, it is compact. By uniform convergence on compact subsets,
\[
\sup_{S\in K_R}|M_n(S)-M(S)|\pto 0.
\]
Hence, with probability tending to one,
\[
\sup_{S\in K_R}|M_n(S)-M(S)|<\delta.
\]
On this event, for every $S\in K_R$,
\[
M_n(S^\ast)\le M(S^\ast)+\delta,
\]
whereas for every $S\in K_R\setminus\mathcal O$, by \eqref{eq:score-gap},
\[
M_n(S)\ge M(S)-\delta\ge M(S^\ast)+2\delta.
\]
Thus no minimizer of $M_n$ over $K_R$ can lie in $K_R\setminus\mathcal O$.

Further, by the sample coercivity part of \Cref{lem:score-coercive}, with probability tending to one,
\[
\inf_{S\in\cS:\,\|S\|_F>R} M_n(S)>M(S^\ast)+2\delta.
\]
Hence no global minimizer of $M_n$ over $\cS$ can lie outside $K_R$. Combining the two conclusions, any global minimizer $\hat S$ of $M_n$ over $\cS$ must belong to $\mathcal O$ with probability tending to one, i.e.,
\[
\PP(\hat S\in\mathcal O)\to 1.
\]
Since $\mathcal O$ was arbitrary, we obtain
\[
\hat S\pto S^\ast.
\]


We now pass the consistency in the score space into the parameter space after fixing the sign convention of \Cref{rem:sign-id}. By \Cref{lem:score-coercive}, the criterion $M_n$ remains coercive on $\mathcal S$ with high probability, ensuring that any global minimizer of $M_n$ belongs to the identified score space $\mathcal S$.

Define the structured score map
\[
\Phi: \Theta_r \to \mathcal{S}, \quad \Phi(\theta) = \gamma\mu^\top + UV^\top.
\]
By \Cref{prop:ident}, after fixing column signs, $\Phi$ is invertible with smooth inverse 
$\Phi^{-1}: \mathcal{S} \to \Theta_r$.  
Since both $\Phi$ and its inverse are continuous, the continuous mapping theorem yields
\[
\hat{\theta} = \Phi^{-1}(\hat{S}) \pto \Phi^{-1}(S^*) = \theta^*.
\]

\end{proof}

\subsection{Asymptotic normality and nondegeneracy of the Fisher information}
\label{app:normality}
The asymptotic normality of the constrained maximum likelihood estimator requires the nonsingularity of the Fisher information matrix
\[
I_{\Theta}(\theta^*)
=
\sum_{(k,i,j)\in\Omega_w}
w_{kij} p_{kij}^*(1-p_{kij}^*)
\nabla_\theta \eta_{kij}(\theta^*)
\nabla_\theta \eta_{kij}(\theta^*)^\top.
\]
In this section, we first argue that such nondegeneracy can be derived from the active comparison graph structure together with the local regularity of the structured score map (\Cref{prop:fisher-nondegeneracy}), and then we establish the asymptotic normality.
The argument is inspired by the graph-Laplacian viewpoint for pairwise-comparison Fisher information developed in \cite{han2025statisticalinferencepairwisecomparison}, where positivity of the Fisher information is obtained from graph connectivity (after quotienting out the additive invariance of pairwise-comparison
models). In our setting, the argument proceeds in two layers:
\begin{enumerate}
    \item establish positive definiteness of the Fisher information in the ambient row-centered score space \(\mathcal S\) via judge-wise connected active graphs; and
    \item pull this positivity back to the structured parameter space \(\Theta_r\) through the differential of the map \(\Phi(\theta)=S(\theta)=\gamma\mu^\top + UV^\top\).
\end{enumerate}
The second step is needed because our inferential target is not the unrestricted score matrix
\(S\), but the constrained latent parameter \(\theta=(\gamma,\mu,U,V)\).

We first show the local regularity which passes score-space nondegeneracy to parameter-space nondegeneracy.

\begin{proposition}[Local map injectivity]
\label[proposition]{prop:fisher-nondegeneracy}
Let \(\theta^*=(\gamma^*,\mu^*,U^*,V^*)\in\Theta_r\) denote the true parameter with $\mu^*\ne0$.
After fixing the column signs of \((U,V)\) according to a local sign convention, there exists a neighborhood of \(\theta^*\) in \(\Theta_r\) that is parameterized by a smooth local chart on the manifold. Moreover, the differential of the structured score map
\[\Phi:\Theta_r \to \mathcal S,
\qquad
\Phi(\theta)=S(\theta)=\gamma\mu^\top + UV^\top,
\]
is injective at \(\theta^*\) on the tangent space \(T_{\theta^*}\Theta_r\). That is,
\[D\Phi(\theta^*)[\delta\theta]=0,\ \delta\theta\in T_{\theta^*}\Theta_r
\quad\Longrightarrow\quad\delta\theta=0.
\]
\end{proposition}
\begin{proof}[Proof of \Cref{prop:fisher-nondegeneracy}]
The proof is divided into two steps.

\paragraph{Step 1: Local chart of $\Theta_r$.}
By \Cref{cond:norm}, near $\theta^*$ the set $\Theta_r$ is the zero set of
\[
\psi(\theta):=\Big(\one_K^\top\gamma-K,\ \one_N^\top\mu,\ \one_K^\top U,\ V^\top\one_N,\ \mu^\top V,\ N^{-1}V^\top V-I_r,\ \operatorname{offdiag}(U^\top U)\Big),
\]
where $\operatorname{offdiag}(\cdot)$ keeps the off-diagonal entries, since the remaining requirement on the diagonal of $K^{-1}U^\top U$ is an open condition. Its differential at $\theta^*$ is
\[
D\psi(\theta^*)[\delta\theta]
=
\Big(\one_K^\top\delta\gamma,\ \one_N^\top\delta\mu,\ \one_K^\top\delta U,\ \delta V^\top\one_N,\ \delta\mu^\top V^*+{\mu^*}^\top\delta V,\ N^{-1}(\delta V^\top V^*+{V^*}^\top\delta V),\ \operatorname{offdiag}(\delta U^\top U^*+{U^*}^\top\delta U)\Big).
\]

We show that $D\psi(\theta^*)$ is surjective. Let $D=K^{-1}{U^*}^\top U^*=\diag(d_1,\dots,d_r)$ as in \Cref{cond:norm}. For any target $(a_1,a_2,b_1^\top,b_2,b_3^\top,\Xi_1,\Xi_2)$, where $\Xi_1$ is symmetric and $\Xi_2$ is symmetric with zero diagonal, take
\[
\delta\gamma=\frac{a_1}{K}\one_K,\quad
\delta\mu=\frac{a_2}{N}\one_N,\quad
\delta U=\frac1K\one_K b_1^\top+U^*X,\quad
\delta V=\frac1N\one_N b_2^\top+\frac{\mu^* b_3^\top}{\|\mu^*\|_2^2}+\frac12 V^*\Xi_1,
\]
where $X_{ii}=0$ and $X_{ij}=(\Xi_2)_{ij}/(2Kd_i)$ for $i\neq j$; this is well defined since $d_i>0$ and $\mu^*\neq\zero$. By \Cref{cond:norm}, the first six components of $D\psi(\theta^*)[\delta\theta]$ match the target directly, and for the last one,
\[
\delta U^\top U^*+{U^*}^\top\delta U=K(X^\top D+DX),
\qquad
K(d_jX_{ji}+d_iX_{ij})=(\Xi_2)_{ij}\quad(i\neq j).
\]
Hence $D\psi(\theta^*)$ is surjective. By the regular value theorem, after fixing the column signs of $(U,V)$, $\Theta_r$ is a smooth embedded submanifold near $\theta^*$, with
\[
T_{\theta^*}\Theta_r=\ker D\psi(\theta^*).
\]

\paragraph{Step 2: Injectivity of $D\Phi(\theta^*)$.}
Let $\delta\theta\in\ker D\psi(\theta^*)$ satisfy
\[
D\Phi(\theta^*)[\delta\theta]=\delta\gamma\,{\mu^*}^\top+\gamma^*\,\delta\mu^\top+\delta U\,{V^*}^\top+U^*\,\delta V^\top=0.
\]
Left-multiplying by $K^{-1}\one_K^\top$ and using $\delta\theta\in\ker D\psi(\theta^*)$ and \Cref{cond:norm}, we get
\[
0=\frac1K(\one_K^\top\delta\gamma){\mu^*}^\top+\frac1K(\one_K^\top\gamma^*)\delta\mu^\top+\frac1K(\one_K^\top\delta U){V^*}^\top+\frac1K(\one_K^\top U^*)\delta V^\top
=\delta\mu^\top.
\]
Then right-multiplying by $\mu^*$, and noting $\delta V^\top\mu^*=-{V^*}^\top\delta\mu=\zero$, we get
\[
0=\|\mu^*\|_2^2\,\delta\gamma+\delta U\,({V^*}^\top\mu^*)+U^*\,(\delta V^\top\mu^*)
=\|\mu^*\|_2^2\,\delta\gamma,
\]
so $\delta\gamma=\zero$ since $\mu^*\neq\zero$.

It remains that $\delta U\,{V^*}^\top+U^*\,\delta V^\top=0$. Let $A:=N^{-1}{V^*}^\top\delta V$, which is skew-symmetric since $\delta\theta\in\ker D\psi(\theta^*)$. Right-multiplying by $N^{-1}V^*$ gives $\delta U=-U^*A^\top=U^*A$, and therefore
\[
\operatorname{offdiag}(\delta U^\top U^*+{U^*}^\top\delta U)
=K\operatorname{offdiag}(DA-AD)=0
\quad\Longrightarrow\quad
(d_i-d_j)A_{ij}=0\quad(i\neq j).
\]
Since $d_1>\cdots>d_r$ and $A$ is skew-symmetric, $A=\zero$ and thus $\delta U=\zero$. Finally, $U^*\,\delta V^\top=0$ and ${U^*}^\top U^*=KD$ is invertible, so $\delta V=\zero$. Hence $\delta\theta=0$.
\end{proof}
\Cref{prop:fisher-nondegeneracy} is the infinitesimal analog of parameter identifiability at the local chart of $\theta^*$.
Now we can state the next proposition, which is crucial to the proof of asymptotic normality.
\begin{proposition}[Graph-induced nondegeneracy of the Fisher information]
    \label[proposition]{prop:fisher_graph_nondegenerate}
    Suppose \Cref{asm:sampling,asm:graph} hold. Let
    \[
    S^*:=S(\theta^*)=\gamma^*{\mu^*}^\top + U^*{V^*}^\top.
    \]
    
    Then the Fisher information matrix
    \[
    I_{\Theta}(\theta^*)
    =
    \sum_{(k,i,j)\in\Omega_w}
    \omega_{kij}\,
    \nabla_\theta \eta_{kij}(\theta^*)
    \nabla_\theta \eta_{kij}(\theta^*)^\top
    \]
    is positive definite on the identified tangent space \(T_{\theta^*}\Theta_r\); that is,
    \[
    \delta\theta^\top I_{\Theta}(\theta^*)\delta\theta>0
    \qquad
    \text{for every }0\neq \delta\theta\in T_{\theta^*}\Theta_r.
    \]
    Consequently, \(I_\Theta(\theta^*)\) is nonsingular in the local identified coordinates.
\end{proposition}

We begin to prove the second part of \Cref{thm:asymptotics}. Proof of \Cref{prop:fisher_graph_nondegenerate} follows afterwards.
\begin{proof}[Proof of \Cref{thm:asymptotics}: asymptotic normality]
In \Cref{app:consistency}, we have established $\hat{S} \pto S^*$. Asymptotic normality of $\hat{S}$ can follow this consistency. Since $\mathcal{S}$ is a finite-dimensional Euclidean space and $M_n$ is a smooth criterion (the binomial log-likelihood is twice continuously differentiable in $S$), standard asymptotic MLE theory applies, i.e., consistency combined with the local positive definiteness 
of the Hessian (equivalently, the Fisher information in the $S$-parameterization) implies asymptotic normality. But \Cref{prop:fisher_graph_nondegenerate} guarantees the nonsingularity of the Fisher information in the $S$-space, so the local positive definiteness condition is satisfied.
This yields $\sqrt{n}(\hat{S} - S^*) \dto \mathcal{N}(0, \mathcal{I}_{\mathcal{S}}(S^*)^{-1})$.

We now pass this normality from the $S$-space to the parameter space $\Theta_r$.
By the delta method applied to $\Phi^{-1}$,
\[
\sqrt{n}(\hat{\theta} - \theta^*) = \sqrt{n} D\Phi^{-1}(S^*)(\hat{S} - S^*) + o_p(1) 
\dto \mathcal{N}(0, D\Phi^{-1}(S^*) \Sigma_S [D\Phi^{-1}(S^*)]^\top),
\]

The limiting covariance in the $S$-space is $\Sigma_S = \mathcal{I}_{\mathcal{S}}(S^*)^{-1}$, where
\[
\mathcal{I}_{\mathcal{S}}(S^*) = \sum_{(k,i,j)\in\Omega_w} w_{kij} p_{kij}^*(1-p_{kij}^*) \nabla_S \eta_{kij}(S^*) \nabla_S \eta_{kij}(S^*)^\top
\]
is positive definite by \Cref{prop:fisher_graph_nondegenerate} (Step 1--2 of its Proof).


The likelihood depends on $\theta$ only through the linear predictor $\eta_{kij}(\theta)$. By the chain rule, for each active triple $(k,i,j) \in \Omega_w$,
\begin{equation}
\frac{\partial \eta_{kij}}{\partial \theta} 
= 
\frac{\partial \eta_{kij}}{\partial S} \cdot \frac{\partial S}{\partial \theta}
=
\nabla_S \eta_{kij}(S^*) \cdot D\Phi(\theta^*),
\label{eq:chain-rule-eta}
\end{equation}
where $D\Phi(\theta^*)$ is the Jacobian matrix of the score map $\Phi(\theta) = \gamma\mu^\top + UV^\top$ evaluated at $\theta^*$.

The Fisher information in the structured parameter space is therefore
\begin{align}
\mathcal{I}_\Theta(\theta^*)
&=
\sum_{(k,i,j)\in\Omega_w}
w_{kij} p_{kij}^*(1-p_{kij}^*)
\left(\frac{\partial \eta_{kij}}{\partial \theta}\right)
\left(\frac{\partial \eta_{kij}}{\partial \theta}\right)^\top
\nonumber\\
&=
\sum_{(k,i,j)\in\Omega_w}
w_{kij} p_{kij}^*(1-p_{kij}^*)
D\Phi(\theta^*)^\top 
\nabla_S \eta_{kij}(S^*) 
\nabla_S \eta_{kij}(S^*)^\top 
D\Phi(\theta^*)
\nonumber\\
&=
D\Phi(\theta^*)^\top 
\left[
\sum_{(k,i,j)\in\Omega_w}
w_{kij} p_{kij}^*(1-p_{kij}^*)
\nabla_S \eta_{kij}(S^*) 
\nabla_S \eta_{kij}(S^*)^\top
\right]
D\Phi(\theta^*)
\nonumber\\
&=
D\Phi(\theta^*)^\top \, \mathcal{I}_{\mathcal{S}}(S^*) \, D\Phi(\theta^*).
\label{eq:fisher-transform}
\end{align}

\Cref{prop:fisher-nondegeneracy} ensures that $D\Phi(\theta^*)$ is injective. Combined with the positive definiteness of $\mathcal{I}_{\mathcal{S}}(S^*)$, the conjugation by $D\Phi(\theta^*)$ preserves positive definiteness.

Therefore,
\[
\sqrt{n}(\hat\theta - \theta^*) \dto \mathcal{N}(0, \mathcal{I}_{\Theta}(\theta^*)^{-1}),
\]
and for any smooth scalar target functional $q(\theta)$, the delta method yields
\[
\sqrt{n}\{q(\widehat{\theta})-q(\theta^*)\}
\dto
\mathcal{N}\!\left(
    0,
    \nabla_{\theta}q(\theta^*)^\top
    \cI(\theta^*)^{-1}
    \nabla_{\theta}q(\theta^*)
\right),
\]
completing the proof.
\end{proof}

\begin{proof}[Proof of \Cref{prop:fisher_graph_nondegenerate}]
The proof is divided into three steps.

\paragraph{Step 1: Score-space Fisher information as a weighted Laplacian form.}
Temporarily parameterize the model by the row-centered score matrix
\[
\mathcal S:=\{S\in\mathbb R^{K\times N}:S\one_N=0\}.
\]
For a fixed active triple \((k,i,j)\in\Omega_w\), recall that
\[
\eta_{kij}(S):=S_{ki}-S_{kj}.
\]
Under the binomial Bradley--Terry--Luce model, the population Fisher information in the ambient score parameterization is
\[
I_{\mathcal S}(S^*)
=
\sum_{(k,i,j)\in\Omega_w}
\omega_{kij}\,
\nabla_S \eta_{kij}(S^*)
\nabla_S \eta_{kij}(S^*)^\top,
\]
where \(\omega_{kij}=w_{kij}p_{kij}^*(1-p_{kij}^*)>0\).

Let \(H=(H_{kn})\in\mathcal S\) be an arbitrary score perturbation. Then the corresponding quadratic form is
\begin{equation}
\label{eq:score_fisher_quadratic}
\langle H, I_{\mathcal S}(S^*)H\rangle=\sum_{(k,i,j)\in\Omega_w}\omega_{kij}(H_{ki}-H_{kj})^2.
\end{equation}
Indeed, for each triple \((k,i,j)\), the directional derivative of \(\eta_{kij}(S)\) along \(H\) is
\[
D\eta_{kij}(S^*)[H]=H_{ki}-H_{kj},
\]
which gives \eqref{eq:score_fisher_quadratic}.

Now fix a judge \(k\in[K]\), and write \(h_k=(H_{k1},\dots,H_{kN})^\top\in\mathbb R^N\).
Then the \(k\)-th summand in \eqref{eq:score_fisher_quadratic} becomes
\[
\sum_{\{i,j\}\in E_k^w}\omega_{kij}(h_{k,i}-h_{k,j})^2=h_k^\top L_k h_k,
\]
where \(L_k\) is the weighted graph Laplacian of the active graph \(G_k^w=([N],E_k^w)\) with edge weights \(\omega_{kij}\). Therefore
\begin{equation}
\label{eq:block_laplacian_decomp}
\langle H, I_{\mathcal S}(S^*)H\rangle
=
\sum_{k=1}^K h_k^\top L_k h_k.
\end{equation}

\paragraph{Step 2: Judge-wise graph connectivity implies positivity on \(\mathcal S\).}
By \Cref{asm:graph}, each active graph \(G_k^w\) is connected.
Because \(L_k\) is the weighted Laplacian of the connected graph \(G_k^w\) with strictly positive edge weights, for every \(x\in\mathbb R^N\),
\[
x^\top L_k x=\sum_{\{i,j\}\in E_k^w}\omega_{kij}(x_i-x_j)^2 \ge 0.
\]
Hence \(x\in\ker(L_k)\) implies \(x^\top L_k x=0\), so \(x_i=x_j\) for every edge \(\{i,j\}\in E_k^w\). Since \(G_k^w\) is connected, \(x\) must be constant on all vertices,
that is, \(x=c\one_N\). Conversely, \(L_k\one_N=0\). Therefore,
\[
\ker(L_k)=\mathrm{span}\{\one_N\}.
\]

Hence, for each \(k\),
\[
h_k^\top L_k h_k=0\quad\Longleftrightarrow\quad h_k=c_k \one_N
\text{ for some }c_k\in\mathbb R.
\]
But \(H\in\mathcal S\) means every row of \(H\) is centered, i.e.
\[
H\one_N=0.
\]
Equivalently,
\[
\one_N^\top h_k=0 \qquad\text{for all }k\in[K].
\]
If \(h_k=c_k\one_N\), then \(0=\one_N^\top h_k = Nc_k\), so \(c_k=0\), and therefore \(h_k=0\).
Since this holds for every \(k\), we conclude that
\[
\langle H, I_{\mathcal S}(S^*)H\rangle=0
\quad\Longrightarrow\quad
H=0.
\]
Thus \(I_{\mathcal S}(S^*)\) is positive definite on \(\mathcal S\):
\begin{equation}
\label{eq:score_PD}
\langle H, I_{\mathcal S}(S^*)H\rangle>0
\qquad
\forall\, 0\neq H\in\mathcal S.
\end{equation}

\paragraph{Step 3: Pulling back positivity to the structured parameter space $\Theta_r$.}
Let \(\delta\theta\in T_{\theta^*}\Theta_r\) be an arbitrary tangent perturbation, and define the induced score perturbation
\[
H:=D\Phi(\theta^*)[\delta\theta]\in\mathcal S.
\]
Since
\[
\Phi(\theta)=\gamma\mu^\top + UV^\top,
\]
the first-order variation of \(S(\theta)\) at \(\theta^*\) is
\begin{equation}
\label{eq:dPhi_expansion}
H
=
\delta\gamma\,{\mu^*}^\top
+
\gamma^*\,\delta\mu^\top
+
\delta U\,{V^*}^\top
+
U^*\,\delta V^\top.
\end{equation}

For each active triple \((k,i,j)\in\Omega_w\), the directional derivative of the structured logit difference satisfies
\[
D\eta_{kij}(\theta^*)[\delta\theta]
=
H_{ki}-H_{kj}.
\]
Equivalently,
\[
\nabla_\theta \eta_{kij}(\theta^*)^\top \delta\theta
=
H_{ki}-H_{kj}.
\]
Therefore,
\begin{align}
\delta\theta^\top I_\Theta(\theta^*)\delta\theta
&=
\sum_{(k,i,j)\in\Omega_w}\omega_{kij}\bigl(\nabla_\theta \eta_{kij}(\theta^*)^\top\delta\theta\bigr)^2\nonumber\\
&=\sum_{(k,i,j)\in\Omega_w}\omega_{kij}(H_{ki}-H_{kj})^2\nonumber\\
&=\langle H, I_{\mathcal S}(S^*)H\rangle.
\label{eq:pullback_quadratic}
\end{align}

Now suppose
\[
\delta\theta^\top I_\Theta(\theta^*)\delta\theta=0.
\]
Then \eqref{eq:pullback_quadratic} and \eqref{eq:score_PD} imply that \(H=0\). That is,
\[
D\Phi(\theta^*)[\delta\theta]=0.
\]
By \Cref{prop:fisher-nondegeneracy}, the differential \(D\Phi(\theta^*)\) is injective
on \(T_{\theta^*}\Theta_r\), hence \(\delta\theta=0\).

We have thus shown that
\[
\delta\theta^\top I_\Theta(\theta^*)\delta\theta>0
\qquad
\forall\, 0\neq \delta\theta\in T_{\theta^*}\Theta_r,
\]
which proves positive definiteness of \(I_\Theta(\theta^*)\) on the identified tangent space.
Therefore \(I_\Theta(\theta^*)\) is nonsingular in any local identified coordinate system.
\end{proof}

\begin{remark}[Interpretation of the two-layer argument]
\label{rem:two_layer_argument}
The proposition separates the source of nondegeneracy into two logically distinct parts.

Unlike in ordinary BTL models, graph connectivity alone is not sufficient for Fisher nondegeneracy in the structured latent-factor parameterization, due to parameter-level identifiability from local regularity of the structured map. It must be combined with local injectivity of the decomposition map, which is established in \Cref{prop:fisher-nondegeneracy}.
\end{remark}

\begin{remark}[Relation to the pairwise-comparison literature]
\label{rem:relation_pairwise_inference}
The proof above is inspired by recent asymptotic-inference theory for pairwise-comparison models, where the Fisher information is represented as a graph-Laplacian quadratic form and its nondegeneracy is deduced from graph connectivity after quotienting out the additive invariance of pairwise scores \citep{han2025statisticalinferencepairwisecomparison}. The present argument extends that idea to the heterogeneous judge-aware model by treating the structured parameter space \(\Theta_r\) as a smooth submanifold of the row-centered score space \(\mathcal S\), and then pulling back the ambient Fisher information through the differential of the map \(\Phi(\theta)=\gamma\mu^\top+UV^\top\).

In other words, the graph structure controls nondegeneracy in the ambient score geometry, whereas the constraint geometry of \(\Theta_r\) controls whether such nondegeneracy is retained after passing to the latent parameters.
\end{remark}

\subsection{Auxiliary lemmas} \label{app:lemmas}
Two lemmas used in the proof of \Cref{thm:asymptotics} are established below.

\begin{lemma}[Injectivity of active score differences]
\label[lemma]{lem:active-ident}
Under \Cref{asm:graph}, if $S,S'\in\mathcal S$ satisfy
\[
S_{ki}-S_{kj}=S'_{ki}-S'_{kj}
\qquad\text{for all }(k,i,j)\in\Omega_w,
\]
then $S=S'$.
\end{lemma}
\begin{proof}[Proof of \Cref{lem:active-ident}]
Fix a judge $k$ and let $h_k:=S_{k,:}-S'_{k,:}\in\RR^N$. The assumption implies
\[
h_{ki}-h_{kj}=0
\qquad\text{for every }\{i,j\}\in\mathcal E_k^w.
\]
Because the graph $\mathcal G_k^w$ is connected, $h_k$ must be constant across items, say
\[
h_k=c_k\one_N.
\]
Since both $S$ and $S'$ lie in the centered space $\mathcal S$, we also have
\[
h_k^\top\one_N=0.
\]
Therefore $Nc_k=0$, so $c_k=0$ and hence $h_k=0$. This holds for every judge $k$, and thus $S=S'$.
\end{proof}

\begin{lemma}[Coercivity]
\label[lemma]{lem:score-coercive}
Under \Cref{asm:graph}, the population criterion $M$ in \eqref{eq:population-score-criterion} is coercive on $\mathcal S$, i.e.
\[
\|S\|_F\to\infty,\quad S\in\mathcal S
\qquad\Longrightarrow\qquad
M(S)\to\infty.
\]
Moreover, with probability tending to one, the sample criterion $M_n$ in \eqref{eq:score-nll} is also coercive on $\mathcal S$.
\end{lemma}
\begin{proof}[Proof of \Cref{lem:score-coercive}]
We first prove the coercivity of the population criterion.

For each fixed active triple $(k,i,j)\in\Omega_w$, define
\[
\phi_{kij}(x):=-p_{kij}^\ast x+\log(1+e^x).
\]
Because $p_{kij}^\ast\in(0,1)$, the function $\phi_{kij}$ is continuous, nonnegative, and satisfies
\[
\phi_{kij}(x)\to\infty
\qquad\text{as }|x|\to\infty.
\]
Indeed, as $x\to+\infty$,
\[
\phi_{kij}(x)=(1-p_{kij}^\ast)x+o(x)\to\infty,
\]
while as $x\to-\infty$,
\[
\phi_{kij}(x)=-p_{kij}^\ast x+o(1)\to\infty.
\]

Now fix a judge $k$ and write the $k$th row of $S$ as
\[
s_k:=(S_{k1},\dots,S_{kN})^\top\in\RR^N.
\]
Since $S\in\mathcal S$, we have $s_k^\top\one_N=0$. Because $\mathcal G_k^w$ is connected and $N$ is fixed, we have
\begin{equation}\label{eq:row-graph-bound}
\|s_k\|_2
\le \sqrt{N}\|s_k\|_1 \le
2\sqrt{N}\sum_{\{i,j\}\in\mathcal E_k^w}|s_{k,i}-s_{k,j}|.
\end{equation}
To see this, fix a reference node in $\mathcal G_k^w$ and telescope along paths from the reference node to each other node; since the graph is finite and connected, all resulting path lengths are uniformly bounded.

Suppose now that $\|S\|_F\to\infty$ with $S\in\mathcal S$. Then for some judge $k$ we must have $\|s_k\|_2\to\infty$. By \eqref{eq:row-graph-bound}, this implies that for at least one active edge $\{i,j\}\in\mathcal E_k^w$,
\[
|S_{ki}-S_{kj}|=|s_{k,i}-s_{k,j}|\to\infty.
\]
Since the corresponding summand in $M(S)$ is
\[
w_{kij}\phi_{kij}(S_{ki}-S_{kj}),
\]
with $w_{kij}>0$, it follows that this term diverges to $+\infty$. All other summands are nonnegative, and therefore
\[
M(S)\to\infty.
\]
This proves population coercivity.

We next prove sample coercivity. On the event
\[
\mathcal A_n:=\{0<Y_{kij}<n_{kij}\text{ for all }(k,i,j)\in\Omega_w\},
\]
define
\[
\phi_{kij,n}(x):=-\bar Y_{kij}x+\log(1+e^x).
\]
Whenever $0<\bar Y_{kij}<1$, the same argument as above shows that
\[
\phi_{kij,n}(x)\to\infty
\qquad\text{as }|x|\to\infty.
\]
Since $p_{kij}^\ast\in(0,1)$ and $n_{kij}\to\infty$ for each active triple, we have
\[
\PP\bigl(Y_{kij}=0 \text{ or } Y_{kij}=n_{kij}\bigr)\to 0,
\]
and therefore, by a union bound over the finite set $\Omega_w$,
\[
\PP(\mathcal A_n)\to1.
\]
On $\mathcal A_n$, the same graph argument used for the population criterion shows that
\[
\|S\|_F\to\infty,\quad S\in\mathcal S
\qquad\Longrightarrow\qquad
M_n(S)\to\infty.
\]
Thus $M_n$ is coercive on $\mathcal S$ with probability tending to one.
\end{proof}

\subsection{Partial sampling scheme} \label{subsec:partial-sample}

Let $\mathcal{G}:=([N],\mathcal{E})$ denote the pooled graph where $\mathcal{E}:=\{\{i,j\}:(k,i,j)\in\Omega_w\text{ for some }k\}$ denotes the edge set. If the graph belonging to a judge, say, $\mathcal{G}_k$ is not connected, then the sampling is considered partial. 
For example, a judge may compare Item 1, 2, 3 and 4, 5, 6, but the graphs of 1--2--3 and 4--5--6 are separated. Under such a scenario, identification becomes the major obstacle to statistical analysis. The judge may assign higher scores to all 1, 2, and 3, and lower scores to all 4, 5, and 6 without breaking \Cref{eq:S-id} or altering the likelihood. Therefore, assumptions should be made to rule out such ambiguity: the true parameter shall lie in a region where such perturbations must make it impossible for $S$ to still admit the structured representation \eqref{eq:score-decomp} with unchanged $r$. 

We define the structured score manifold
\[\mathcal{M}_r:=\{S=\gamma\mu^\top+UV^\top:(\gamma,\mu,U,V)\in\Theta_r\}
\cap \mathcal{S},
\]
and let $\mathcal{T}^*:=T_{S^*}\mathcal{M}_r$ be its tangent space at $S^*$ to characterize the potential local perturbations on $S^*$.
\begin{assumption}[Identification under partial overlap]
\label{asm:graph-partial}
Every nonzero local perturbation of the structured score matrix changes at least one observed pairwise log-odds difference. Equivalently, if $H\in\mathcal{T}^*$ satisfies
\[
H_{ki}-H_{kj}=0
\qquad\text{for all }(k,i,j)\in\Omega_w,
\]
then we must have $H=0$.
\end{assumption}
On the other hand, we assume connectivity on the pooled graph to ensure that there are shared nodes (items) that connect the disconnected parts of each judge, so information is expected to be aligned with the help of these connecting nodes.
\begin{assumption}[Connectivity of the pooled graph]\label{asm:pooled-connect}
    The pooled graph $\mathcal{G}=([N],\mathcal{E})$ with edge set $\mathcal{E}:=\{\{i,j\}:(k,i,j)\in\Omega_w\text{ for some }k\}$ is connected.
\end{assumption}
\begin{remark}[How overlap helps]
The point of \Cref{asm:graph-partial,asm:pooled-connect} is that overlap lets different judges' local comparisons be aligned through the shared low-rank structure. For example, suppose judge 1 compares items $(1,2)$ and judge 2 compares items $(2,3)$, but neither judge compares all three pairs. Item 2 then serves as an overlap point linking the two local pieces of information. \Cref{asm:graph-partial} dictates that, after imposing the structured model, such overlaps are rich enough that one cannot change the latent score matrix in a nontrivial way without altering at least one observed comparison.
\end{remark}
\begin{assumption}[Local compact parameter chart]
\label{asm:local-compact}
There exists a compact set $\Theta_{r,0}\subset\Theta_r$ such that $\theta^*$ lies in its relative interior after fixing the local anchor signs of the columns of $U$ and $V$. And the population structured criterion
\[
m(\theta):=\sum_{(k,i,j)\in\Omega_w}w_{kij}\Big[-p_{kij}^*\eta_{kij}(\theta)+\log\big(1+\exp\{\eta_{kij}(\theta)\}\big)\Big],
\]
has $\theta^*$ as its unique minimizer over $\Theta_{r,0}$.
\end{assumption}
This assumption replaces the coercivity established in \Cref{lem:score-coercive} under \Cref{asm:graph} (full sampling setting). As the identification constraint becomes obscured for the entire space $\mathcal{S}$, this assumption, together with \Cref{asm:local-compact}, dictates that even though the likelihood may be flat along component-wise shifts of an unrestricted row, those shifts should either be infeasible within the fixed-rank structured model or break the $\Theta_{r,0}$ boundary that we set.

By the standard asymptotic theory of MLE, the asymptotic properties of $\hat{\theta}$ of \Cref{thm:asymptotics} still hold under \Cref{asm:pooled-connect,asm:graph-partial,asm:sampling,asm:graph-partial,asm:local-compact}.

In practice, however, we need to avoid this setting due to the difficulty in verifying \Cref{asm:local-compact,asm:graph-partial}, unless the sample size is really limited.

\clearpage

\section{Simulation Studies}\label{app:simulation}

This section details \Cref{subsec:comparison} on the specifications and settings.

\subsection{Main heterogeneous DGP}
\label{app:sim-main-dgp}

The main synthetic simulations use \(N=8\) items, \(K=4\) judges, true heterogeneity rank \(r=1\), and 50 independent repetitions for each condition. The sample-size experiment varies
\[
n_{\mathrm{cmp}} \in \{400,800,1200,1600,2000,2500,3000\}
\]
at heterogeneity scale \(h=1\). The heterogeneity experiment fixes \(n_{\mathrm{cmp}}=800\) and varies \(h\in\{0,0.5,1,2\}\). Error bars in Figure~\ref{fig:synthetic} report mean \(\pm 1.96\) standard errors across repetitions.

For each repetition, the true parameters are generated as follows. We draw \(\tilde\mu_i \overset{\iid}{\sim} \cN(0,1)\) and center it as
\[
\mu=\tilde\mu-\bar{\tilde\mu}\one_N ,
\]
so that \(\one_N^\top\mu=0\). We draw judge consensus sensitivities as \(\gamma=Kp\), where \(p\sim\mathrm{Dirichlet}(\one_K)\), so that \(\one_K^\top\gamma=K\). For \(r>0\), raw item and judge factor matrices are drawn with iid standard-normal entries. The item factor is column-centered and projected off the consensus direction,
\[
\widetilde V \leftarrow \widetilde V-\one_N\bar v^\top,\qquad
\widetilde V \leftarrow \widetilde V-\mu\frac{\mu^\top\widetilde V}{\max(\mu^\top\mu,\varepsilon)} .
\]
A QR decomposition is then applied in this constrained column space, yielding an item factor satisfying \(V^\top\one_N=0\) and \(\mu^\top V=0\). The judge factor is similarly column-centered,
\[
\widetilde U \leftarrow \widetilde U-\one_K\bar u^\top ,
\]
and orthogonalized in the zero-sum judge space, so that \(\one_K^\top U=0\).

The code assigns decreasing factor strengths \(r,r-1,\ldots,1\) and sets
\[
V=Q_V\operatorname{diag}\{\sqrt N\,\mathrm{strengths}\},
\qquad
U=Q_U\operatorname{diag}\{\sqrt K\,\mathrm{strengths}\}.
\]
The item factor is then re-centered, re-projected off \(\mu\), and the columns of \((U,V)\) are sign-canonicalized. Finally, both \(U\) and \(V\) are multiplied by \(\sqrt h\), so the heterogeneous score component \(UV^\top\) is scaled by \(h\). These operations preserve the zero-sum and orthogonality constraints.

By construction, the generated parameters satisfy
\[
\one_N^\top\mu=0,\qquad
\one_K^\top\gamma=K,\qquad
\one_K^\top U=0,\qquad
V^\top\one_N=0,\qquad
\mu^\top V=0.
\]
After QR normalization and canonical sign anchoring, the factors also satisfy the anchoring requirements in \Cref{cond:norm}. The true score matrix is then
\[
S=\gamma\mu^\top+UV^\top,
\]
and pairwise outcomes are generated according to the BTL model in \eqref{eq:binomial-model}.
We also use an auxiliary DGP from JA-Ranking, which uses the same \(\mu,\gamma\) draw but sets \(U=V=0\), so \(S_{ki}=\gamma_k\mu_i\). The experiment specifications and results are shown in \Cref{app:sim-ja-dgp}.

\subsection{Comparison design, fitting protocol, and metrics}
\label{app:sim-protocol}

For the design of experiment, comparisons are allocated using a near-balanced design over all judge-item-pair cells \((k,i,j)\) with \(i<j\). Let \(C=KN(N-1)/2\). Each cell receives either \(\lfloor n_{\mathrm{cmp}}/C\rfloor\) or \(\lceil n_{\mathrm{cmp}}/C\rceil\) comparisons, with the remaining cells selected uniformly without replacement. For each comparison in cell \((k,i,j)\), we draw
\[
Y_{kij}^{(\ell)}\sim \mathrm{Bernoulli}\{\sigma(S_{ki}-S_{kj})\},
\]
where \(Y_{kij}^{(\ell)}=1\) means that judge \(k\) prefers item \(i\) to item \(j\). The generated records are aggregated into upper-triangular comparison counts \(n_{kij}\) and win counts \(Y_{kij}\).

All methods are fit to the same aggregated counts. For pooled BTL, we aggregate comparisons over judges and fit a single centered score vector \(\hat\mu_{\mathrm{BTL}}\), then form
\[
\hat S_{\mathrm{BTL}}=\one_K\hat\mu_{\mathrm{BTL}}^\top .
\]
For JA-Ranking, we fit the sensitivity-only model
\[
\hat S_{\mathrm{JA}}=\hat\gamma\hat\mu^\top ,
\]
using the same centering and normalization convention as in the main text. For HJA, we fit the rank-\(r\) constrained model
\[
\hat S_{\mathrm{HJA}}=\hat\gamma\hat\mu^\top+\hat U\hat V^\top .
\]
All score-level metrics below are computed from these fitted score matrices.

For each method, score-entry intervals are obtained by applying the corresponding asymptotic covariance estimate to the fitted score-entry functional \(q(\theta)=S_{ki}\). For HJA, this uses the plug-in Fisher-information estimator from Theorem~\ref{thm:asymptotics}; for JA-Ranking and pooled BTL, we use the corresponding model-specific delta-method covariance estimates under their restricted parameterizations. For misspecified baselines, coverage is evaluated against the true heterogeneous score matrix \(S\), so undercoverage reflects both statistical uncertainty and structural bias.

We report four groups of metrics:
\begin{enumerate}[(1)]
    \item \textbf{Score-matrix error.}
    Score recovery is measured by the mean squared error
    \[
    \mathrm{MSE}
    =
    \frac{1}{KN}
    \sum_{k=1}^K\sum_{i=1}^N
    (\hat S_{ki}-S_{ki})^2 .
    \]

    \item \textbf{Consensus ranking recovery.}
    Let \(r_i^\star\) and \(\hat r_i\) denote the ranks of item \(i\) induced by the true consensus score \(\mu\) and the estimated consensus score \(\hat\mu\), respectively, with rank \(1\) assigned to the largest score. We report the Spearman rank correlation
    \[
    \rho_{\mathrm{sp}}(\mu,\hat\mu)
    =
    \operatorname{corr}
    \bigl(
    (r_1^\star,\ldots,r_N^\star),
    (\hat r_1,\ldots,\hat r_N)
    \bigr).
    \]
    When there are no ties, this is equivalently
    \[
    \rho_{\mathrm{sp}}(\mu,\hat\mu)
    =
    1-
    \frac{6\sum_{i=1}^N(r_i^\star-\hat r_i)^2}
    {N(N^2-1)}.
    \]

    We also report NDCG@\(N\). Let \(\hat\pi\) be the ordering of items sorted by decreasing \(\hat\mu_i\), and define the true relevance score
    \[
    \operatorname{rel}_i = 2^{N-r_i^\star}-1 .
    \]
    The discounted cumulative gain of the estimated ranking is
    \[
    \operatorname{DCG}@N
    =
    \sum_{\ell=1}^N
    \frac{\operatorname{rel}_{\hat\pi(\ell)}}{\log_2(\ell+1)}.
    \]
    Let \(\pi^\star\) be the ordering of items sorted by decreasing \(\mu_i\). The ideal discounted cumulative gain is
    \[
    \operatorname{IDCG}@N
    =
    \sum_{\ell=1}^N
    \frac{\operatorname{rel}_{\pi^\star(\ell)}}{\log_2(\ell+1)}.
    \]
    We define
    \[
    \operatorname{NDCG}@N
    =
    \frac{\operatorname{DCG}@N}{\operatorname{IDCG}@N}.
    \]

    \item \textbf{Uncertainty calibration.}
    Let \(C_{ki}^{0.95}\) denote the nominal 95\% confidence interval for the score entry \(S_{ki}\). Coverage is measured by
    \[
    \mathrm{Coverage}
    =
    \frac{1}{KN}
    \sum_{k=1}^K\sum_{i=1}^N
    \ind\{S_{ki}\in C_{ki}^{0.95}\}.
    \]
    Note that ``Unstructured BTL + SVD'' is excluded from this metric as truncated SVD is a non-smooth mapping prohibiting ordinary uncertainty quantification.

    \item \textbf{Judge-specific sign accuracy.}
    Judge-specific pairwise behavior is measured by the fraction of item-pair signs correctly recovered:
    \[
    \mathrm{SignAcc}
    =
    \frac{1}{K\binom{N}{2}}
    \sum_{k=1}^K
    \sum_{1\le i<j\le N}
    \ind
    \left\{
    \operatorname{sign}(S_{ki}-S_{kj})
    =
    \operatorname{sign}(\hat S_{ki}-\hat S_{kj})
    \right\}.
    \]
\end{enumerate}

\clearpage
\subsection{Convergence analysis} \label{app:subsec:sim-converge}

A run is declared converged when the relative likelihood decrease falls below the tolerance used in \Cref{alg:mle}.
\Cref{fig:convergence} offers a typical example of the convergence of \Cref{alg:mle}. Since each iteration step calculates a minimum, its convergence rate is fast.
\begin{figure}[!ht]
    \centering
    \includegraphics[width=0.9\linewidth]{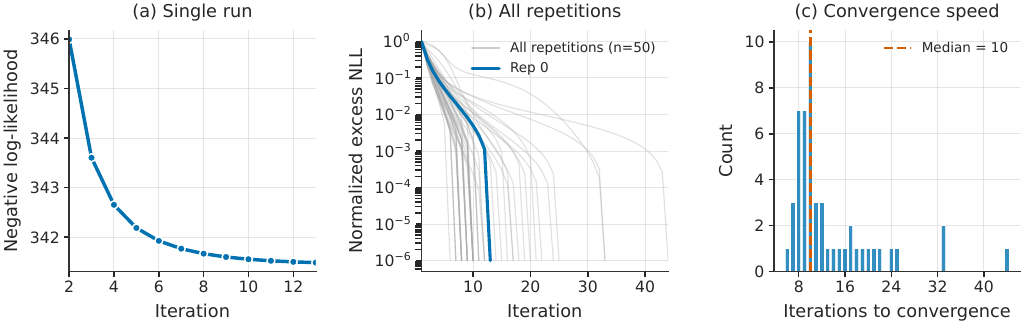}
    \caption{Convergence behavior of our proposed model. (a) Negative log-likelihood (NLL) decreases monotonically over iterations for a representative run. (b) Normalized excess NLL (log-scale) across $n=50$ repetitions shows consistent geometric convergence regardless of random seed. (c) Histogram of convergence iterations demonstrates that the algorithm consistently converges in a small, predictable number of steps.}\label{fig:convergence}
\end{figure}

\subsection{Auxiliary sensitivity-only DGP}
\label{app:sim-ja-dgp}

We also run an auxiliary simulation under a sensitivity-only data-generating process matching the structure of JA-Ranking. The true score matrix is
\[
S_{ki}=\gamma_k\mu_i,
\]
with no low-rank heterogeneous component. We vary
\[
n_{\mathrm{cmp}}\in\{400,800,1200,1600,2000,2500,3000\},
\]
and otherwise use the same comparison allocation, aggregation, and evaluation metrics as in \Cref{app:sim-protocol}. HJA is still fit with rank \(r=1\), so this experiment evaluates robustness to mild over-parameterization when the true heterogeneity rank is zero.

Figure~\ref{fig:sim-ja-dgp} reports the results. The proposed method remains stable despite fitting one redundant heterogeneous direction. In some finite-sample settings, HJA can match or slightly improve on the sensitivity-only fit. We interpret this as a numerical finite-sample effect rather than a structural advantage: when some generated \(\gamma_k\) values are close to zero, the log-parameterization used by the sensitivity-only baseline can become unstable, whereas the affine parameterization used by HJA does not require strict positivity of \(\gamma_k\).

\begin{figure}[!ht]
    \centering
    \includegraphics[width=1\linewidth]{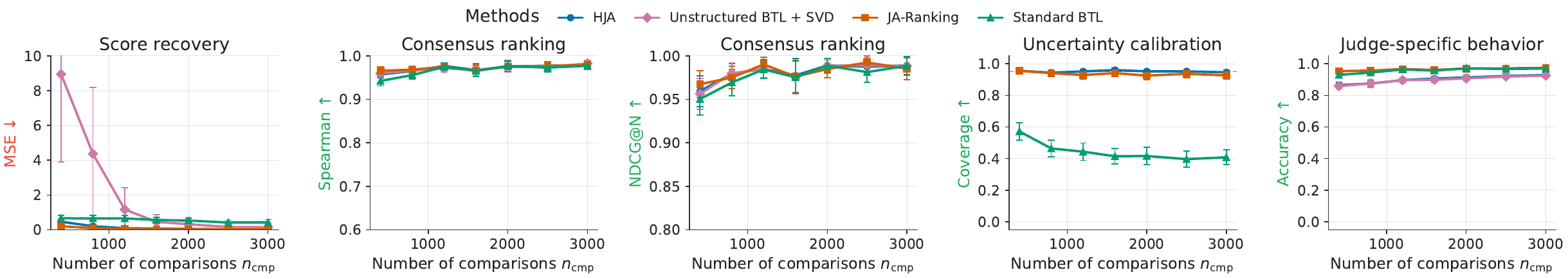}
    \caption{Synthetic performance under JA-Ranking data-generating process \cite{xu2026judgeaware}, where the true score matrix has \(S_{ki}=\gamma_k\mu_i\) and no low-rank interaction term.}
    \label{fig:sim-ja-dgp}
\end{figure}

\subsection{Ablation Study}\label{app:subsec:ablation}
We introduce a variant that performs the same alternating likelihood updates but omits ReAnchor as the ablated model. We also add an unstructured estimation scheme where $\bar S$ (subject to \eqref{eq:S-id}) is first estimated from the data and $\hat S$ and $\hat \theta$ are recovered via \Cref{prop:ident} and truncated SVD on $\tilde S$ defined by \eqref{eq:S-residual}. 
We compare the score recovery, consensus ranking, and sign prediction accuracy of our full model, the direct-S truncated-SVD model, and the ablated model where the ReAnchor Step (\Cref{alg:reanchor}) is removed from \Cref{alg:mle}. The settings are the same as in the main simulation experiment, although the plot of confidence interval coverage accuracy is removed due to the non-identifiability of the ablated model. 
\begin{figure}[!ht]
    \centering
    \includegraphics[width=0.9\linewidth]{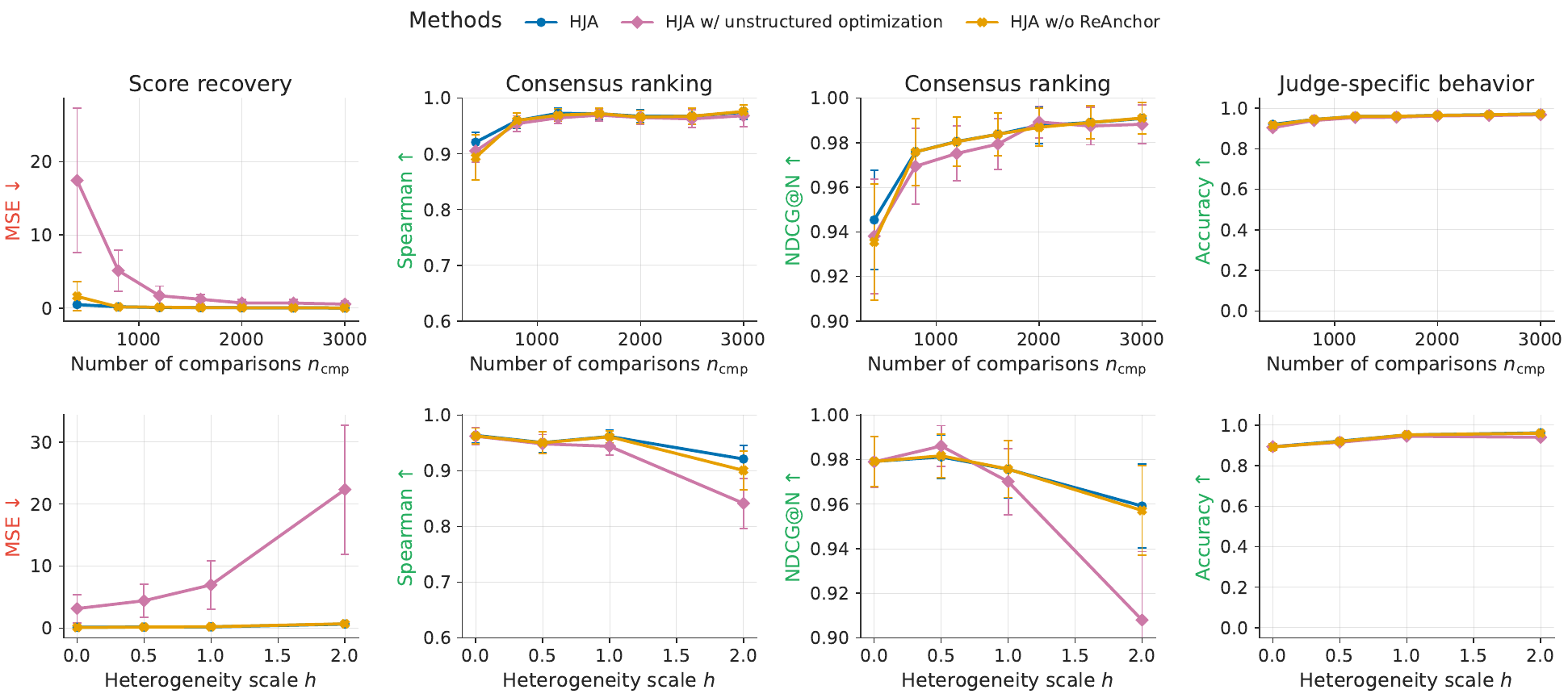}
    \caption{ReAnchor ablation under the synthetic DGP. We compare the proposed HJA estimator, the direct-S truncated-SVD estimator, and an ablated HJA variant that omits the ReAnchor step after each alternating update. Columns report score MSE, Spearman correlation, NDCG@N, and judge-specific sign accuracy; rows vary sample size and heterogeneity strength.}
    \label{fig:ablation}
\end{figure}

\Cref{fig:ablation} shows that the performance of the ablated models are worse than that of the full model. This shows that ReAnchor and structure are important for stabilizing the identified low-rank parameterization and improving score recovery, while the non-reanchored variant can still preserve much of the ordinal information needed for ranking and judge-specific sign prediction.
In summary, the design of \Cref{alg:mle} stabilizes score recovery while preserving the strong ranking and judge-behavior performance of the proposed estimator.

\subsection{Misspecified synthetic data on latent expertise in \Cref{subsec:latent-domain}(B)} \label{app:subsec:latent-domain}
We provide a hidden-domain experiment with $N=12$ items and $K = 9$ judges to illustrate how our model can identify the potential latent areas of judges' expertise, when the synthetic data do not agree with our model.

Specifically, each comparison is a coding or maths task, with frequencies $\pi = (\pi_{\text{coding}}, \pi_{\text{math}}) = (0.65, 0.35)$. Following \eqref{eq:known-domain} with its $U = 0$ and $d$ unobserved, set $\gamma_{kd} = 12$ when judge $k$ is good at domain $d$ and $0.005$ when $k$ is bad at $d$, up to small random perturbations. Hence, a judge is sensitive to a model's ability only within its capable latent domain.

Each item $i$ has a fixed latent skill vector $a_i = (a_{i,\text{code}}, a_{i,\text{math}})$, and the pair of interest has $a_A = (1.5, 0.2)$ and $a_B = (0.2, 1.0)$, so $A$ leads by $1.3$ on coding and $B$ leads by $0.8$ on maths. Conditional on judge $k$, pair $(i, j)$, and domain $d$, the outcome is
\[
Y_{kij} \mid D=d \sim \text{Bernoulli}\left[ \sigma(\gamma_{kd} (a_{id} - a_{jd})) \right].
\]
We define the equal-judge, domain-weighted latent score before generating outcomes as
\[
\bar{S}_i = \frac{1}{K} \sum_{k=1}^{K} \sum_d \pi_d \gamma_{kd} a_{id},
\]
and the contribution of judge $k$ to $A$ versus $B$ is $\Delta_k = \sum_d \pi_d \gamma_{kd} (a_{Ad} - a_{Bd})$. A maths-reliable judge contributes $0.65(0.005)(1.3) + 0.35(12)(-0.8) = -3.36$, while a coding-reliable judge contributes $0.65(12)(1.3) + 0.35(0.005)(-0.8) = 10.14$, hence
\[
\bar{S}_A - \bar{S}_B = (6(-3.36) + 3(10.14)) / 9 = 1.14 > 0,
\]
thus $A > B$, which an unbiased latent panel should report. However, since the domain info $D$ is hidden, the data identify only the domain-marginal win-probability by mixing out latent domains, i.e.,
\begin{equation}
p_{kij} = \sum_d \pi_d \sigma(\gamma_{kd} (a_{id} - a_{jd})) = 0.49 < 0.5.
\end{equation}
So pooled binary outcomes favor $B > A$ due to info loss. The mechanism is also misspecified relative to our model because the \emph{observed} probability $p_{kij}$ is a \emph{mixture model that cannot be expressed via our likelihood}. Thus, it's beyond HJA to identify or ``unmix'' it either, but, as shown below, $UV^\top$ can well \emph{approximate} such mixed hidden interaction.

When models are fit, they only receive judge identities, item pairs, and comparison outcomes, and no task-domain label is supplied. Every judge contributes 4,800 comparisons. Across 50 seeds, HJA recovers $A > B$ in 48 runs, against 35/0 for the sensitivity-only/pooled-BTL model. Held-out accuracy is 0.714/0.551/0.542 for HJA/pooled-BTL/sensitivity-only, showcasing our model's performance under a more general noise paradigm. Five-fold validation selects $r = 1$ in 49 runs, agreeing with the latent contrast of coding vs maths. This experiment shows the versatility of our model when some important labels are hidden in the data.

\clearpage

\section{Real Data} \label{app:sec:real}
This section provides three layers of real-data analysis: (i) experimental 
specifications and robustness benchmarks comparing HJA to baselines (\Cref{app:subsec:real-evaluation}); 
(ii) proper interpretations of patterns of model outputs (\Cref{app:subsec:interpretation}).
\subsection{Experiment specifications for Real-data robustness evaluation of \Cref{subsec:real-evaluation}(C)}
\label{app:subsec:real-evaluation}

The overall aim of real data comparison is to evaluate whether the proposed heterogeneous judge-aware model improves both prediction and ranking sensitivity relative to JA-Ranking and Standard BTL. The experiments are organized around held-out prediction, robustness to injected noisy judges, and accuracy on near-tie item pairs.


\noindent\textbf{Dataset description.}
We evaluate on four pairwise LLM-comparison datasets downloaded from \url{https://github.com/TanXZfra/A-Judge-Aware-Ranking-Framework-for-Evaluating-Large-Language-Models-without-Ground-Truth/tree/main/data}. Recall that $N$ denotes the number of candidate models, $K$ the number of judge models, and $n_{\mathrm{cmp}}$ the number of usable pairwise comparisons after preprocessing. Each comparison is represented as $(k,i,j,y)$, where $k$ indexes the judge, $i,j$ index the compared candidate models, and $y \in \{0,0.5,1\}$ records the judge preference, with ties mapped to $0.5$. We drop only records with unknown labels. For each judge-item-pair cell, $n_{kij}$ denotes the number of observed comparisons for judge $k$ on candidate pair $(i,j)$; the summaries below report the average, minimum, and maximum over nonzero $n_{kij}$ cells.

The Chatbot Arena dataset is derived from the LMSYS Chatbot Arena pairwise evaluation setting under the Apache 2.0 license, a crowdsourced battle platform for comparing chat assistants \cite{chiang2024chatbot,zheng2023mtbench}. The original file contains 10,000 pairwise records. After removing 63 unknown-label records, we use $n_{\mathrm{cmp}}=9{,}937$ comparisons, including 1,224 ties; the tie-excluded count is 8,713. The dataset has $N=20$ candidate models and $K=10$ judge models. Among the possible judge-item-pair cells, 1,737 of 1,900 are observed, and nonzero $n_{kij}$ has average/minimum/maximum $5.72/1/29$.

The MT-Bench dataset follows the LMSYS MT-Bench setting (Apache 2.0 license), a controlled multi-turn benchmark with 80 questions. The original file contains 10,000 pairwise records comparing $N=6$ candidate models, evaluated by $K=20$ judge models \cite{zheng2023mtbench}. After removing 294 unknown-label records, we use $n_{\mathrm{cmp}}=9{,}706$ comparisons, including 756 ties; the tie-excluded count is 8,950. All 300 possible judge-item-pair cells are observed, with nonzero $n_{kij}$ average/minimum/maximum $32.35/10/48$.

The UltraFeedback dataset is derived from UltraFeedback-style instruction-response comparisons (under MIT license), based on a large-scale AI-feedback resource built from diverse prompts and model responses \cite{cui2024ultrafeedback}. The original file contains 10,000 pairwise records. After removing 274 unknown-label records, we use $n_{\mathrm{cmp}}=9{,}726$ comparisons, including 761 ties; the tie-excluded count is 8,965. The dataset contains $N=17$ candidate models and $K=20$ judge models. We observe 2,379 of 2,720 possible judge-item-pair cells, and nonzero $n_{kij}$ has average/minimum/maximum $4.09/1/15$.

The in-house dataset is a private pairwise comparison set over recent LLMs collected by \cite{xu2026judgeaware}. The original file contains 36,000 pairwise records. After removing 161 unknown-label records, we use $n_{\mathrm{cmp}}=35{,}839$ comparisons, including 1,782 ties; the tie-excluded count is 34,057. This is the largest dataset in our experiments, with $N=45$ candidate models and $K=18$ judge models. We observe 15,374 of 17,820 possible judge-item-pair cells, with nonzero $n_{kij}$ average/minimum/maximum $2.33/1/9$.

\noindent\textbf{Experiment details.}
Stability and near-tie experiments use a record-level 80/20 train/test split. Our proposed model selects rank by 5-fold validation on negative log likelihood on the training split or full sample, depending on the workflow. Candidate ranks are \(0,\ldots,\min(K-1,N-2)\).
For the stability experiment, we fit each method on the training records and report held-out pairwise accuracy on the test records. For the robustness experiment, we fit base rankings on the original data, then add noisy judges sequentially. The saved default noisy-judge grid is 1 through 10 added noisy judges; the main table reports steps 1, 5, and 10 as low, mid, and high perturbation levels, with the numbers indicating exact-ranking accuracy versus the base ranking after noisy judges. For the near-tie slice, we select up to 20 item pairs from the training split with at least 20 training records per pair, sorted by empirical win-rate distance to 0.5 and ordered into low/mid/high tertiles, and evaluate the selected near-tie test records separately from all other test records after skipping tie-labeled test records.

For the \textbf{stress test} in \Cref{subsec:stress-test}(D),
we define the favored cluster as the bottom $\lfloor N/3\rfloor$
candidates under the clean HJA consensus. One injected source copies
the candidate-pair and question coverage of all raw records: 10,000
for each public dataset and 36,000 for in-house, approximately matching
the entire clean panel's comparison volume. It deterministically
selects the cluster candidate when exactly one belongs to the cluster;
otherwise, it selects either candidate with equal probability, without
ties. Across 20 label seeds (103--122), we vary only these random choices,
holding the clean fits, cluster, injection volume, and clean-selected
rank ($r=1$) fixed. We compute the Spearman correlation between each
contaminated consensus and that method's clean consensus, reporting
the mean and sample standard deviation across repetitions. This
variability reflects injected-label randomness conditional on the
clean panel. HJA uses $\tau=40$, tolerance $5\times10^{-5}$, at most
300 alternating steps, and 2,000 inner iterations. Sensitivity-only
uses a fixed 100-iteration budget, while pooled BTL allows 800 iterations.
All runs are retained: HJA and pooled BTL meet their convergence
criteria throughout, whereas sensitivity-only reaches its iteration
limit without meeting its convergence criterion. Its results therefore
describe performance under this computational budget.


All experiments were run on a CPU server with 12 cores, the type being Intel(R) Xeon(R) Gold 6126 CPU. One repetition of a whole Robustness experiment takes 45 minutes, while other experiments take minutes for one repetition. The optimization parameter $\tau=30$ is used by default, while it is finetuned to be 5 and 300 for the Robustness experiment on UltraFeedback (excluding low noisy injection) and the Near-tie experiment in in-house data. For each convex optimization task, i.e., Line \ref{line:argmin1}--\ref{line:argmin2} in \Cref{alg:mle}, we use L-BFGS-B method \cite{byrd1995limited} after reparameterization based on \Cref{cond:norm} \ref{cond:center}-\ref{cond:scale}.

\subsection{Interpretation} \label{app:subsec:interpretation}
\noindent\textbf{Outlier judge identification. }
We first show the composite analysis for MT-Bench and UltraFeedback in \Cref{fig:mt-feedback-composite-with-glm}. In both experiments, \texttt{Llama-4-Scout-17B-16E-In} appears noisy and not sensitive to the consensus of other judges.

\begin{figure}[htbp]
    \centering
    
    \begin{subfigure}[b]{1\textwidth}
        \centering
        \includegraphics[width=\textwidth]{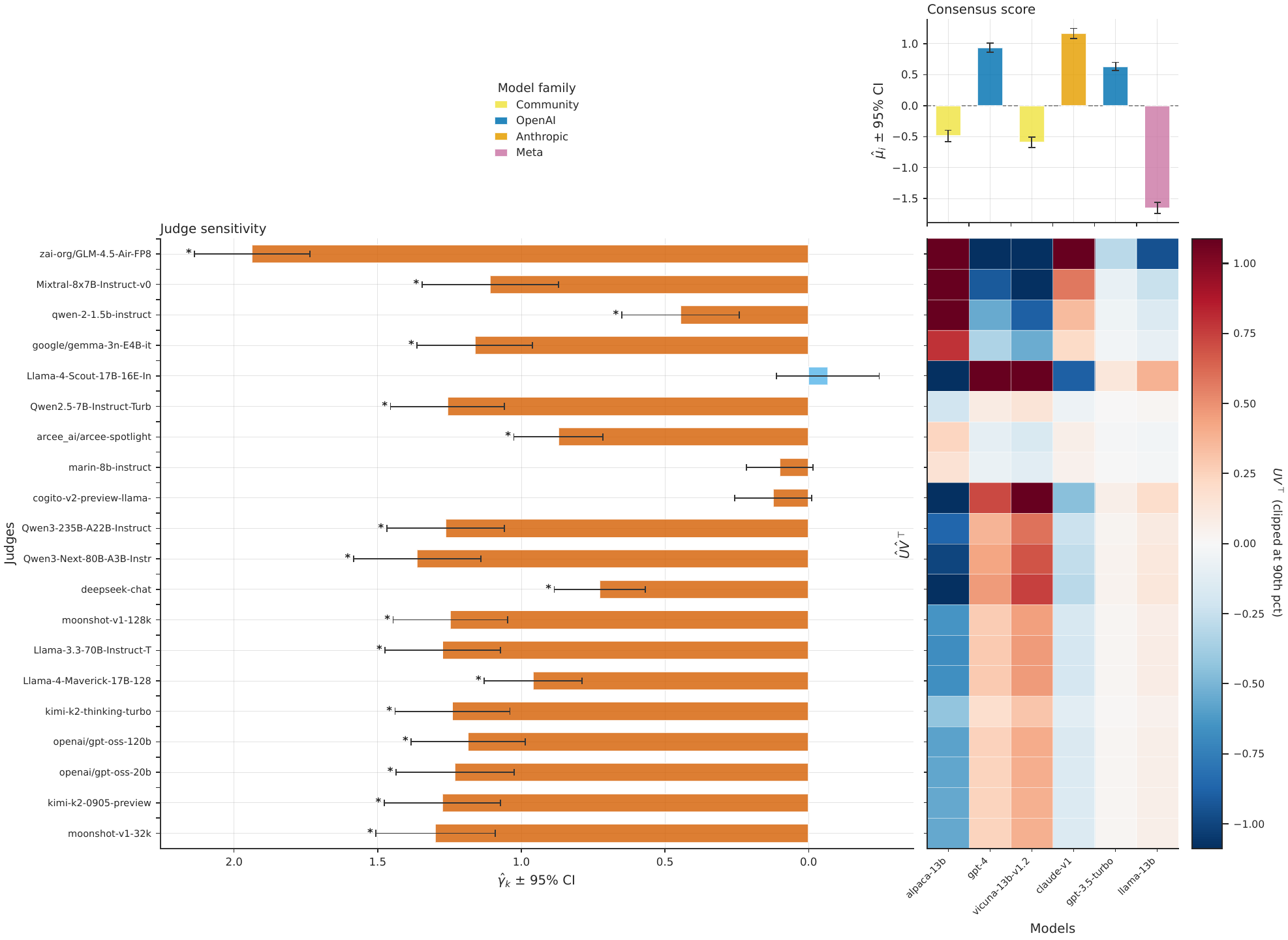}
        \caption{MT-Bench}
        \label{fig:a}
    \end{subfigure}
    
    \vspace{10pt}
    
    \begin{subfigure}[b]{1\textwidth}
        \centering
        \includegraphics[width=\textwidth]{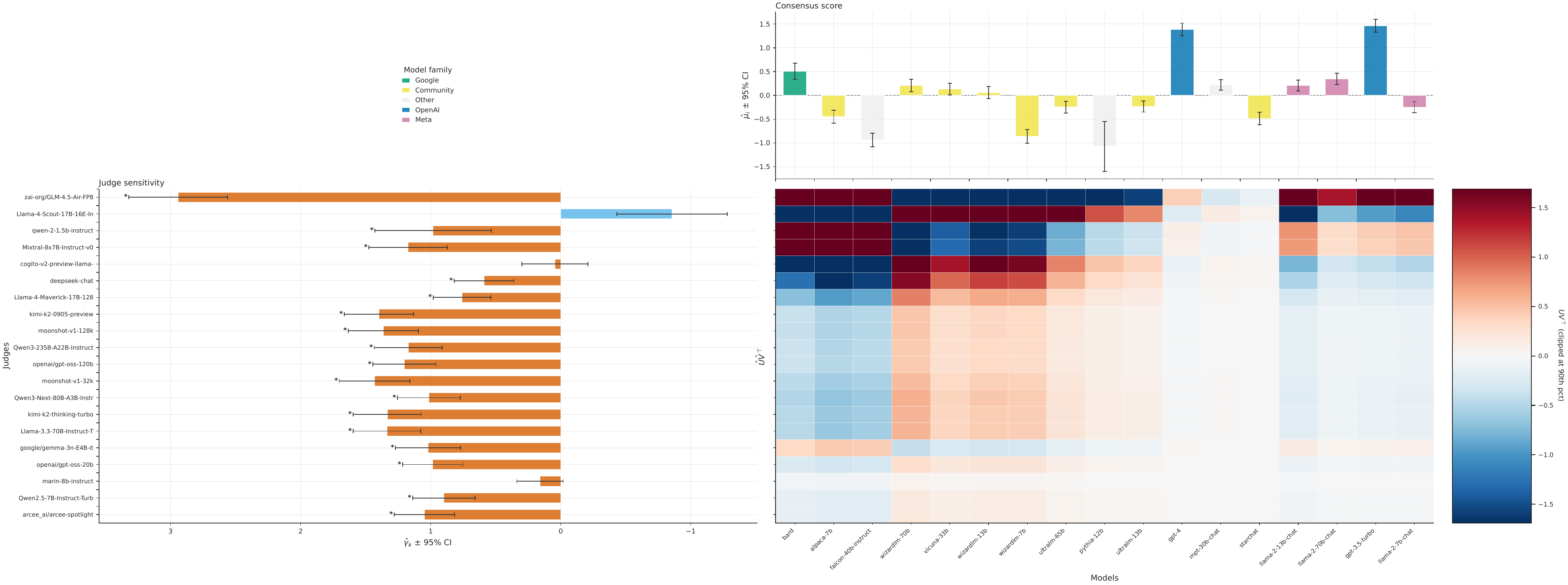}
        \caption{UltraFeedback}
        \label{fig:b}
    \end{subfigure}
    
    \caption{The composite figures showing the judge sensitivity, model rankings, and judge-model interaction by our proposed model. Error bars are given with near-normality assumption.}
    \label{fig:mt-feedback-composite-with-glm}
\end{figure}

However, in order to interpret fitted judge behavior, we need to distinguish consensus sensitivity from residual disagreement. 
A large value of \(\hat\gamma_k\) indicates that judge \(k\) responds strongly to the consensus ranking direction \(\hat\mu\), but it does not by itself imply that the judge is globally reliable. 
We therefore also examine the residual leverage 
\[
\hat h_k=\|(\hat U\hat V^\top)_{k,:}\|_2,
\]
or its relative version 
\[
\hat\rho_k=
\frac{\|(\hat U\hat V^\top)_{k,:}\|_2}
{|\hat\gamma_k|\|\hat\mu\|_2+\varepsilon}.
\]
Judges with large \(\hat\gamma_k\) and small \(\hat h_k\) are strongly aligned with the consensus signal, whereas judges with both large \(\hat\gamma_k\) and large \(\hat h_k\) are high-leverage judges whose evaluations combine sharp consensus discrimination with substantial off-consensus structure. 
In the rank-one fits used here, the sign of \(\hat U_{k1}\) further indicates whether this residual disagreement is aligned with or opposed to the dominant judge-side pattern.

\Cref{fig:leverage-glm} plots $\hat h_k$ against $\hat \gamma_k$. Because only the outliers matter under this metric, we only mark the outlier judge.
\begin{figure}
    \centering
    \includegraphics[width=0.6\linewidth]{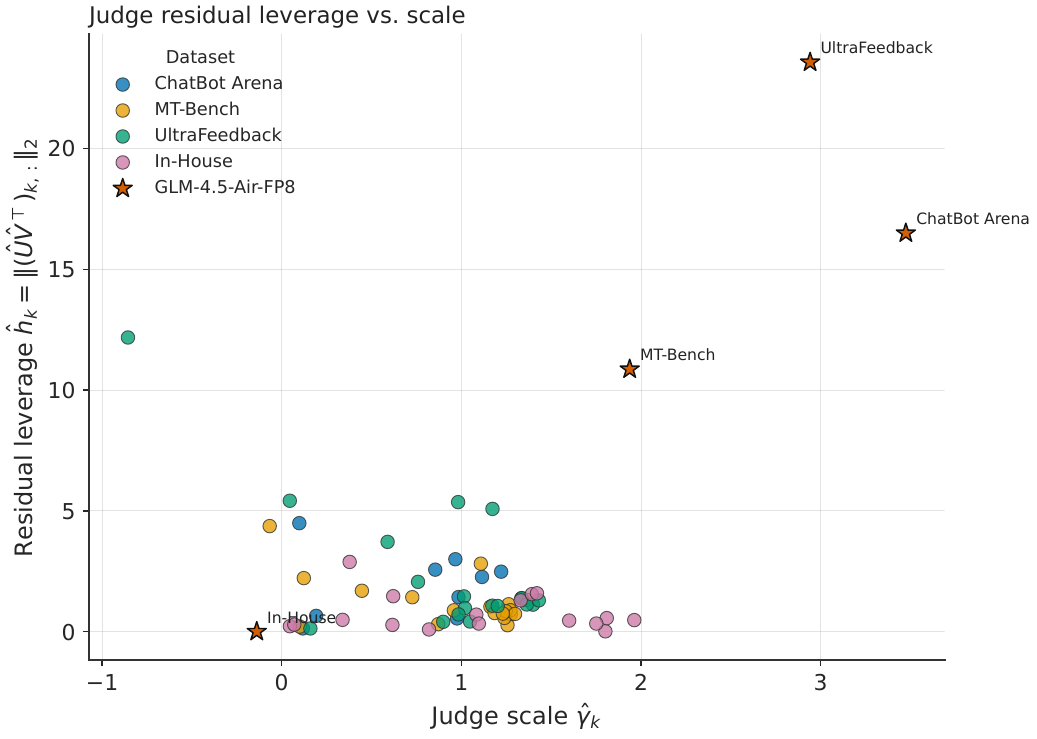}
    \caption{The scatter plot of leverages, where the outlier model is explicitly annotated.}
    \label{fig:leverage-glm}
\end{figure}

Under this diagnostic, \texttt{zai-org/GLM-4.5-Air-FP8} is not best interpreted as simply a highly reliable judge. 
Although it has the largest estimated consensus sensitivity, it also has unusually large residual leverage, with a residual direction opposed to the main judge cluster. 
Thus, GLM-4.5 behaves as a high-sensitivity but high-bias judge: it sharply separates models along the consensus direction while simultaneously imposing a strong idiosyncratic preference component.

In light of this, we rerun our model on MT-Bench and UltraFeedback without GLM 4.5 and visualize the results in \Cref{fig:overall-ultra-mt-no-glm}. From these results, we find that \texttt{Llama-4-Scout-17B-16E-In} is no longer a pure noise, but rather a judge with moderate sensitivity. This shows that the extreme leverage of GLM 4.5 distorts the fitted judge behavior of other judges, and its removal allows a more nuanced interpretation of the remaining judges.

\begin{figure}[htbp]
    \centering
    
    \begin{subfigure}[b]{1\textwidth}
        \centering
        \includegraphics[width=\textwidth]{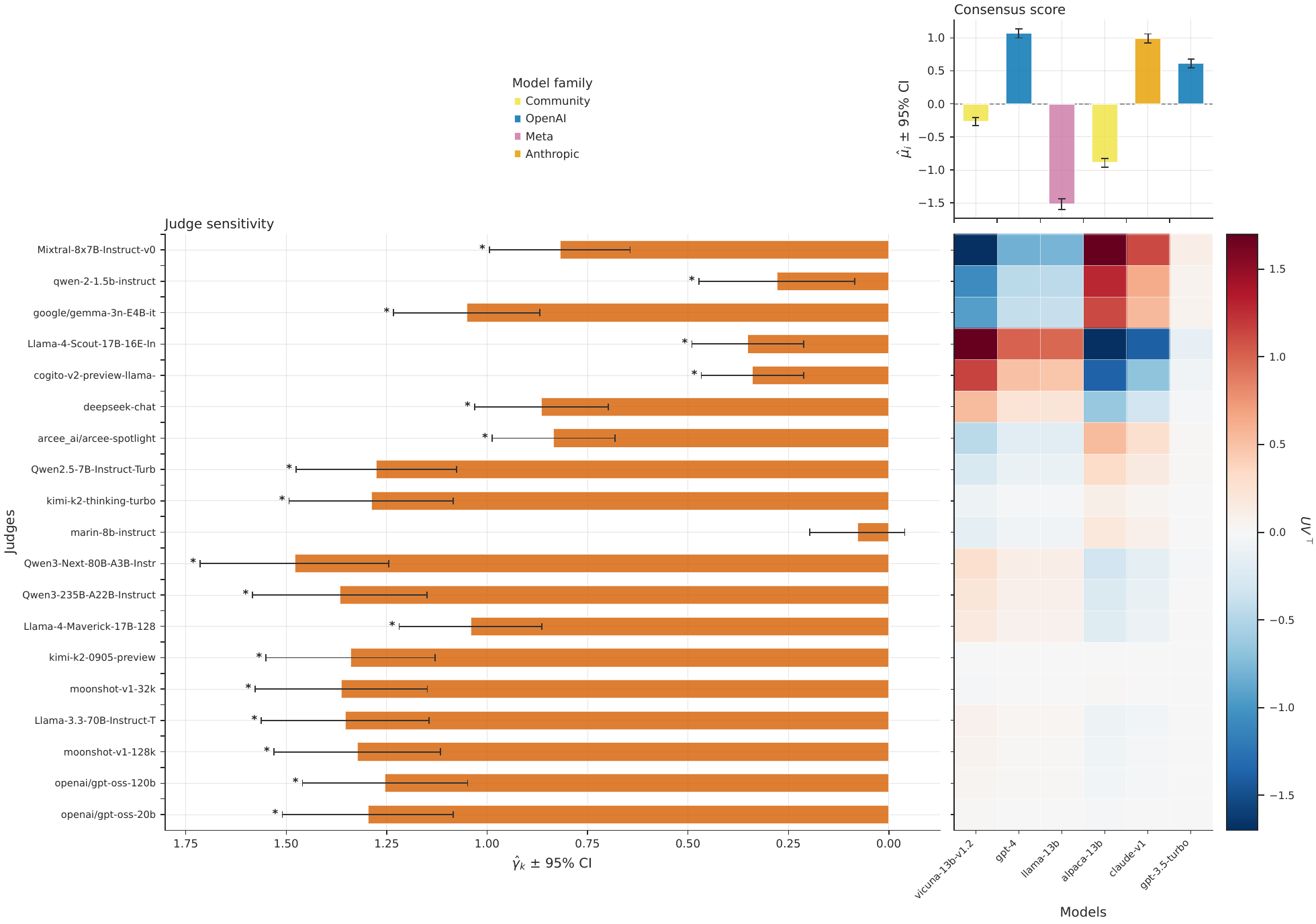}
        \caption{MT-Bench}
        \label{fig:mt-bench}
    \end{subfigure}
    
    \vspace{10pt}
    
    \begin{subfigure}[b]{1\textwidth}
        \centering
        \includegraphics[width=\textwidth]{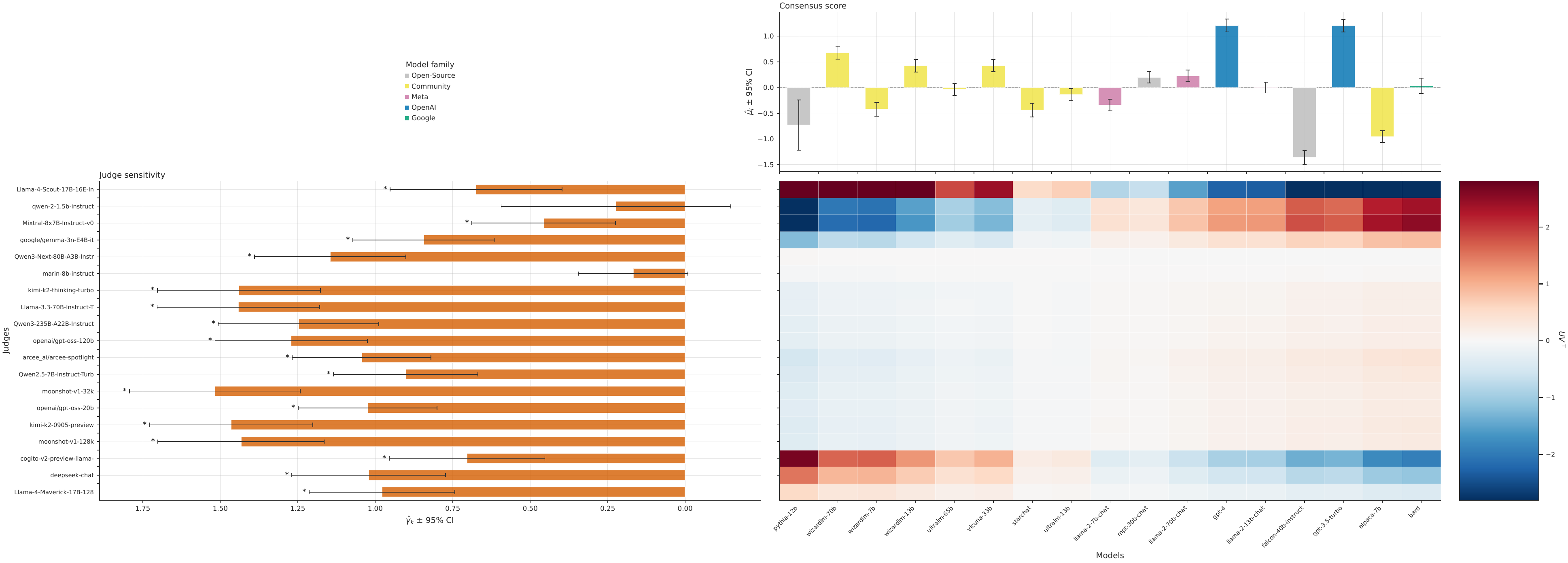}
        \caption{UltraFeedback}
        \label{fig:ultrafeedback}
    \end{subfigure}
    
    \caption{The composite figures showing the judge sensitivity, model rankings, and judge-model interaction by our proposed model with GLM 4.5 removed. Error bars are given with near-normality assumption.}
    \label{fig:overall-ultra-mt-no-glm}
\end{figure}

\noindent\textbf{Self-evaluation bias.} A well-documented phenomenon in automated evaluation is self-evaluation bias (or self-preference), where an LLM judge systematically assigns higher scores to responses generated by itself or models from the same provider family. While this bias poses a significant threat to naive aggregation methods when a judging panel is monopolized by a single model family, its impact is fundamentally mitigated in a diverse, multi-judge setting.
However, apart from selecting a dataset containing diversified judges and models, our formulation can also absorb a judge's systematic self-preference into the orthogonal residual disagreement term $UV^\top$. The heatmap of $UV^\top$ should show bright ``hotspots'' exactly where a judge prefers the output of its own or its family members. This mechanism ensures that $\mu$ retains a high-fidelity representation of consensus quality without the need to completely discard a biased judge. 

As an example, we plot the heatmap of $UV^\top$ to show the heterogeneous interaction between judges and the items on ChatBot Arena after the removal; see \Cref{fig:chatbot-uvt-heatmap}, which illustrates the interaction between the judges and the models.

\Cref{fig:chatbot-uvt-heatmap} shows that although different judges indeed prefer the styles from different models, there is no observed self-preferencing bias, possibly due to the double-blindness in the design and collection of these data.
\begin{figure}
    \centering
    \includegraphics[width=1\linewidth]{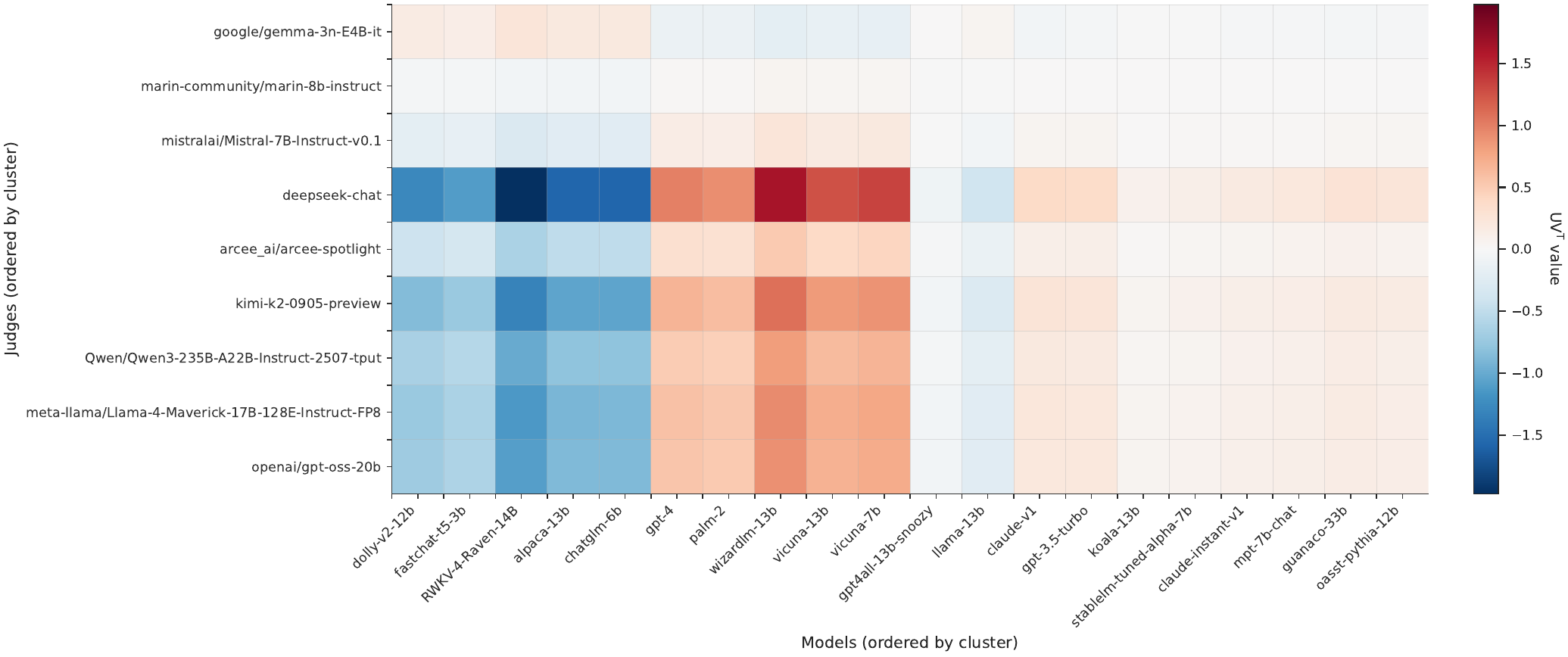}
    \caption{Judge-model heterogeneous interaction terms $UV^\top$ for ChatBot Arena clustered by judge-model affinity (extreme judge removed). The first three judges on the $y$-axis (\texttt{gemma-3n-E4B-it}, \texttt{marin-8b-instruct}, and \texttt{Mistral-7B-Instruct-v0.1}) show comparatively weak heterogeneous preferences, indicating closer alignment with the consensus ranking. Most disagreement happens mainly on the first ten models shown on the $x$-axis, with the majority of judges preferring model 6--10 than 1--5 on the $x$-axis.
    }
    \label{fig:chatbot-uvt-heatmap}
\end{figure}

We also try a biclustering analysis \cite{hartigan1972direct} on the estimated $\hat S$ of in-house data (with GLM 4.5) removed to mine some patterns; see \Cref{fig:bicluster}. We analyze in-house data because it contains the most diversified models and judges and is hence the most suitable for biclustering analysis. We also include the biclustering results for the other three datasets in \Cref{fig:biclustering-all}. The block structures validates HJA's decomposition: judges from the same family 
converge on similar preference profiles, suggesting that systematic disagreement 
is not random noise but reflects coherent, interpretable differences in evaluation philosophy. This makes the residual component $UV^\top$ a meaningful diagnostic tool rather than an unexplained variance sink.

\begin{figure}[htbp]
    \centering
    \includegraphics[width=\linewidth]{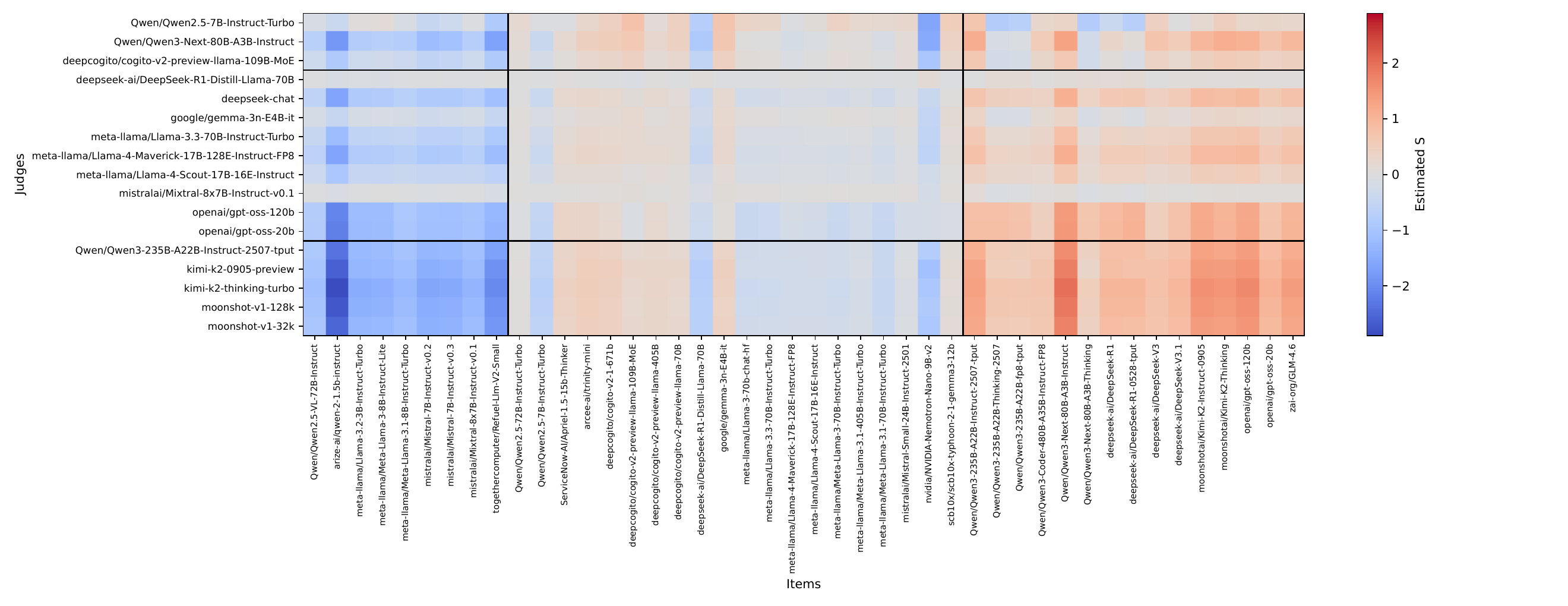}
    \caption{
    Biclustering of the estimated judge-specific score matrix \(\hat S\) on the in-house dataset after removing the high-leverage GLM-4.5 judge. 
    Rows and columns are reordered by biclustering, so contiguous blocks indicate judges and models with similar fitted preference profiles. 
    On the item side, the clustering separates earlier/smaller Qwen--Llama--Mistral-style models from a middle group of DeepCogito, DeepSeek-distilled, Gemma, Llama, and Mistral/Nemotron models, and a rightmost block dominated by recent Qwen3, DeepSeek, Kimi/Moonshot, OpenAI \texttt{gpt-oss}, and GLM models. 
    On the judge side, Qwen-family judges show similar scoring behavior, Kimi and Moonshot judges form a closely related cluster, and DeepSeek, Gemma, Llama, Mistral, and OpenAI judges occupy a more moderate shared cluster. 
    The block structure indicates that HJA recovers family-level affinity patterns in judge-specific rankings, rather than only isolated judge or model effects.
    }
    \label{fig:bicluster}
\end{figure}

\begin{figure}[htbp]
    \centering
    
    \begin{subfigure}[b]{\textwidth}
        \centering
        \includegraphics[width=0.5\textwidth]{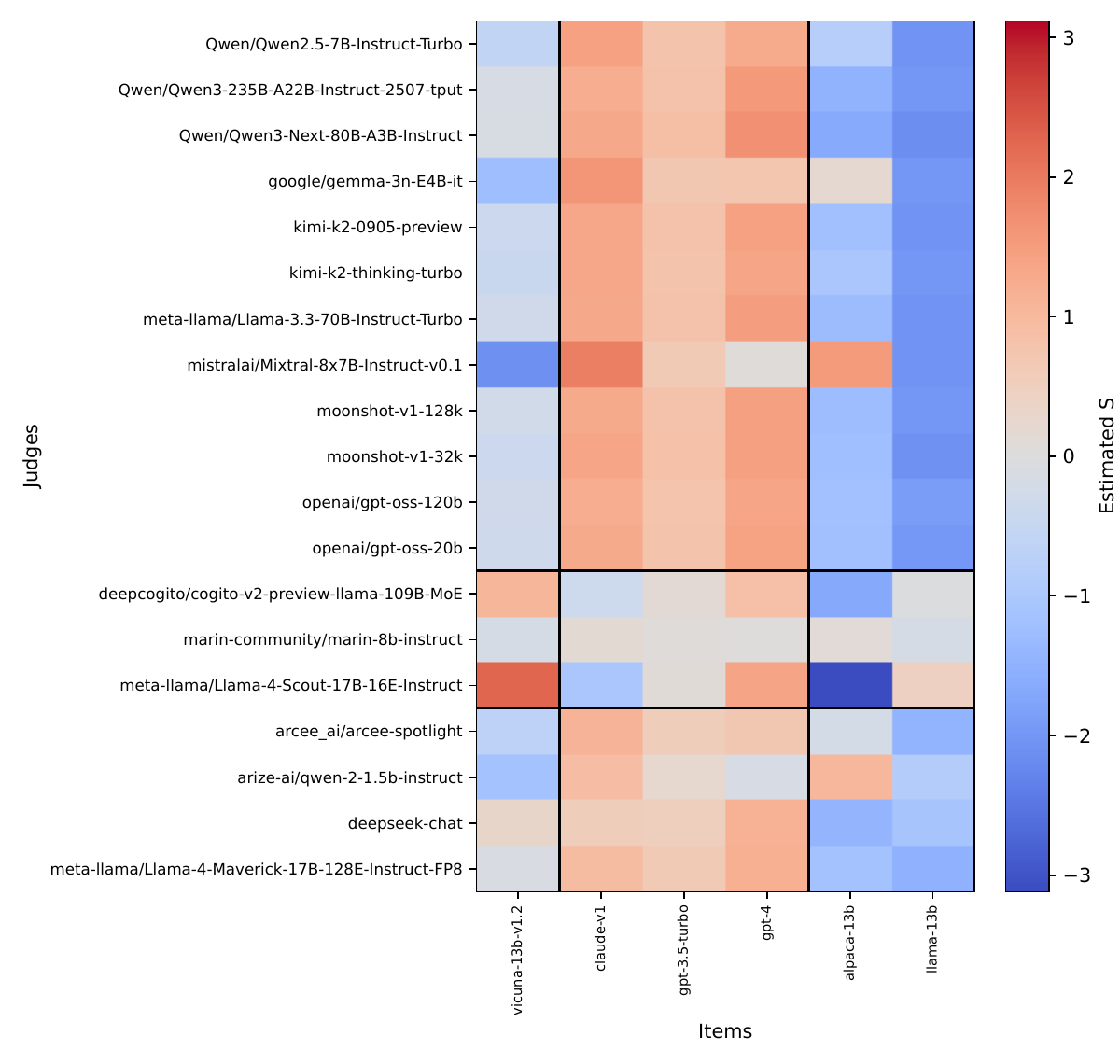}
        \caption{MT-Bench}
        \label{fig:mtbench-bicluster}
    \end{subfigure}
    
    \vspace{12pt}
    
    \begin{subfigure}[b]{0.5\textwidth}
        \centering
        \includegraphics[width=\textwidth]{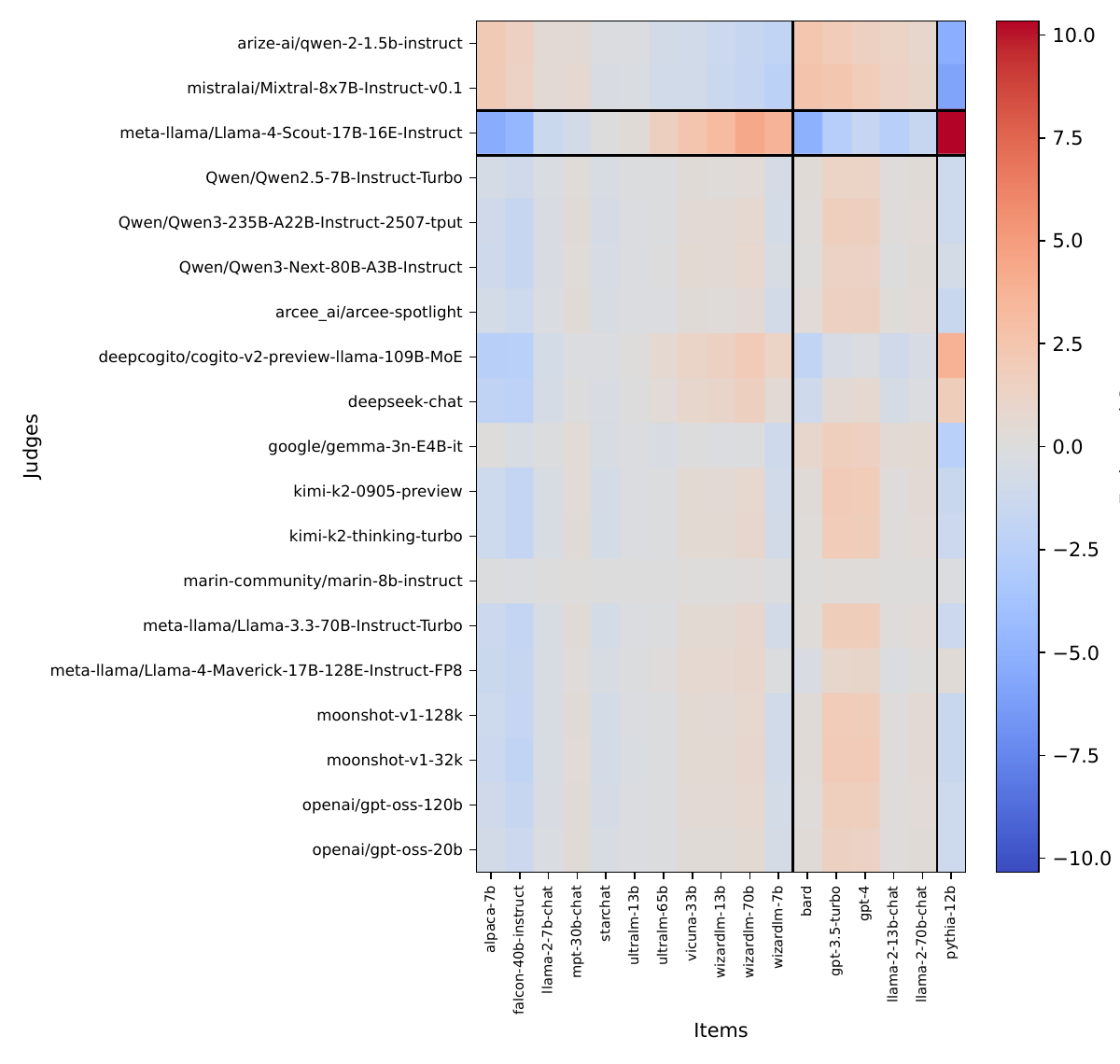}
        \caption{UltraFeedback}
        \label{fig:ultrafeedback-bicluster}
    \end{subfigure}
    
    \vspace{12pt}
    
    \begin{subfigure}[b]{\textwidth}
        \centering
        \includegraphics[width=0.5\textwidth]{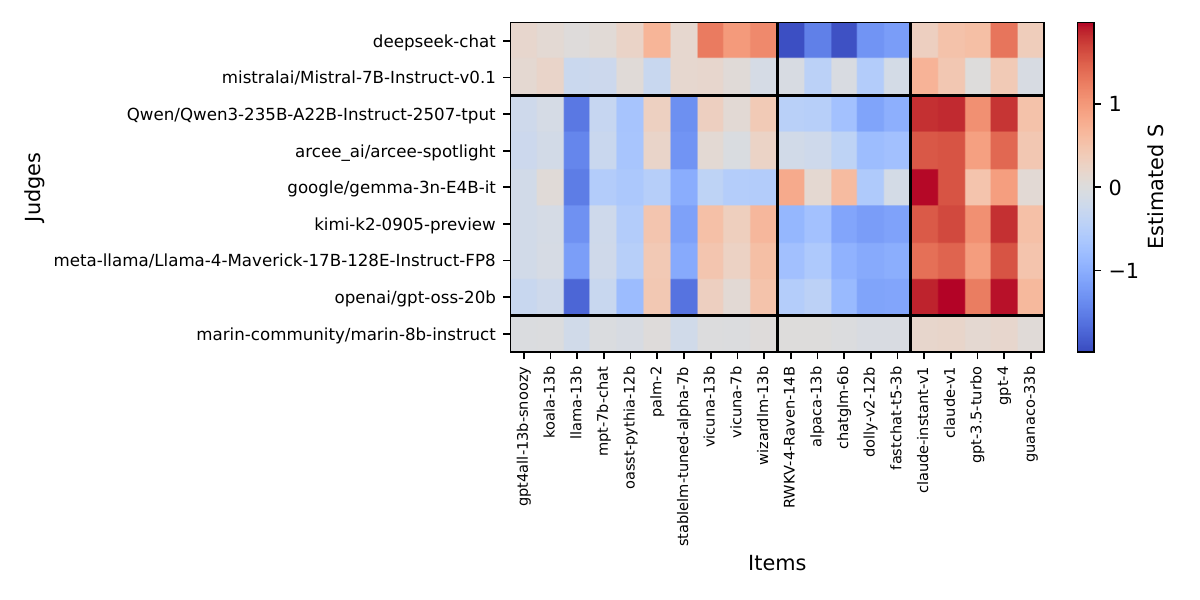}
        \caption{Chatbot Arena}
        \label{fig:arena}
    \end{subfigure}
    
    \caption{Biclustering of estimated $\hat S$ across (a) MT-Bench, (b) UltraFeedback, and (c) Chatbot Arena, all with GLM-4.5 removed. From (c), we find that Claude and GPT are preferred by most judges, and indeed, they have a high consensus score. Such a concordance pattern also appears in (a).}
    \label{fig:biclustering-all}
\end{figure}

\clearpage
\putbib[references]
\end{bibunit}

\end{document}